\documentclass[11pt]{article}
\usepackage{amsmath,amsfonts,amsthm}

\usepackage{txfonts}

\usepackage[textsize=scriptsize]{todonotes}
\usepackage{amssymb,latexsym}
\usepackage{graphicx}
\usepackage{pgf,pgfarrows,pgfnodes,pgfautomata,pgfheaps,pgfshade}
\usepackage[hidelinks]{hyperref}
\usepackage{tikz}
\usepackage{color}
\usepackage{paralist}
\usepackage[]{units}

\newtheorem{proposition}{Proposition}
\newtheorem{corollary}[proposition]{Corollary}
\newtheorem{lemma}{Lemma}

\newtheorem{remark}{Remark}
\newtheorem{definition}{Definition}

\newcounter{maintheoremcnt}
\newtheorem{maintheorem}{Theorem}[maintheoremcnt]

\newcommand{\zerovector}{\langle 0, \ldots, 0 \rangle}
\newcommand{\vot}{{{\mathrm{Vot}}}}
\newcommand{\powA}{{{S_k(A)}}}
\newcommand{\profiles}{{{\mathcal{P}}}}
\newcommand{\orders}{{{\Pi_{>}(A)}}}

\newcommand{\naturals}{{{\mathbb{N}}}}

\newcommand{\rationals}{{{\mathbb{Q}}}}
\newcommand{\reals}{{{\mathbb{R}}}}
\newcommand{\cvx}{{{\mathrm{conv}}}}

\newcommand{\inter}{{{\mathrm{int}}}}
\newcommand{\pos}{{{\mathrm{pos}}}}

\newcommand{\posf}{{{\lambda}}}

\newcommand{\scorefull}[3]{{{\mathrm{score}_{#1}({#2}, {#3})}}}
\newcommand{\pairscorefull}[4]{{{\mathrm{score}_{#1}({#2}, {#3}, {#4})}}}

\newcommand{\posweight}{{\mathrm{pos\hbox{-}weight}}}

\newcommand{\np}{{{\mathrm{NP}}}}

\newcommand{\lin}{{{\mathrm{lin}}}}

\newcommand{\sntv}{{{\mathrm{SNTV}}}}
\newcommand{\pav}{{{\mathrm{PAV}}}}
\newcommand{\Bloc}{{{\mathrm{Bloc}}}}
\newcommand{\maj}{{{\mathrm{maj}}}}
\newcommand{\kBorda}{{{k\text{-}\mathrm{Borda}}}}
\newcommand{\CC}{{{\mathrm{CC}}}}
\newcommand{\shortcite}{\cite}

\makeatletter
\newtheorem*{rep@theorem}{\rep@title}
\newcommand{\newreptheorem}[2]{%
\newenvironment{rep#1}[1]{%
 \def\rep@title{#2 \ref{##1}}%
 \begin{rep@theorem}}%
 {\end{rep@theorem}}}
\makeatother
\newreptheorem{theorem}{Theorem}
\newreptheorem{proposition}{Proposition}
\newreptheorem{lemma}{Lemma}
\newreptheorem{corollary}{Corollary}
\newreptheorem{maintheorem}{Theorem}

\title{Axiomatic Characterization of Committee Scoring Rules}
\author{Piotr Skowron\\
  University of Oxford\\
  Oxford, UK
  \and 
Piotr Faliszewski\\
  AGH University\\
  Krakow, Poland
  \and 
Arkadii Slinko\\
  University of Auckland\\
  Auckland, New Zealand
}

\begin{document}

\maketitle

\begin{abstract}
  Committee scoring rules form a rich class of aggregators of voters'
  preferences for the purpose of selecting subsets of objects with
  desired properties, e.g., a shortlist of candidates for an
  interview, a representative collective body such as a parliament, or
  a set of locations for a set of public facilities.  In the spirit of
  celebrated Young's characterization result that axiomatizes
  single-winner scoring rules, we provide an axiomatic
  characterization of multiwinner committee scoring rules.  We show
  that committee scoring rules---despite forming a remarkably general
  class of rules---are characterized by the set of four standard
  axioms, anonymity, neutrality, consistency and continuity, and by
  one axiom specific to multiwinner rules which we call committee
  dominance. In the course of our proof, we develop several new
  notions and techniques. In particular, we introduce and
  axiomatically characterize multiwinner decision scoring rules, a
  class of rules that broadly generalizes the well-known majority
  relation.
\end{abstract}

\section{Introduction}

One of the most influential results in social choice, Arrow's
impossibility theorem~\cite{arrow1963}, states that when voters have
three or more distinct alternatives (candidates) to choose from, no
social welfare function can map the ranked preferences of individuals
into a transitive social preference order while satisfying four axioms
called unrestricted domain, non-dictatorship, Pareto efficiency, and
independence of irrelevant alternatives.  Arrow's axioms are
reasonable at the individual level of cognition but appeared too
strong to require from a social perspective (this seems particularly
true for the independence of irrelevant alternatives axiom). The
result was negative but it had two important consequences. First,
knowing what is impossible to achieve is important. Second, Arrow
created a framework for developing a positive approach to the social
choice theory, i.e., a framework for investigations of what is
actually possible to achieve. Indeed, numerous axiomatic
characterizations of existing voting rules followed Arrow's work
(these are too numerous to list here, see the survey of Chebotarev and
Shamis~\cite{che-sha:j:scoring-rules} for a comprehensive list of such
characterizations, as well as Section~\ref{sec:related} where we
outline work related to ours).  This was foreshadowed to a certain
extent by May who in a highly original paper axiomatically
characterized the simple majority rule~\cite{mayAxiomatic1952} (but in
a narrow framework that did not allow for generalizations).

Most common voting rules have been introduced without normative
considerations. Hence a discovery of an axiomatic characterization for
a voting rule is hard to overestimate. When we axiomatically
characterize a rule, we are discovering sets of axioms that we know in
advance are consistent (in particular, the rule that is characterized
satisfies all of
them).  %This approach compare favorably to relaxing the conditions of Arrow's Impossibility Theorem in attempt to get something non-dictatorial.
This is, in fact, where the normative theory begins: a commitment to a
particular voting rule is a commitment to the set of axioms that
define this rule. Now, if electoral designers compare two voting
rules, they can look at them from different `angles' where each axiom
provides them with a certain `view' of the rule.  These `views' can be
interpreted as behavioral characteristics with normative
implications. Comparisons of such characteristics can cause an
electoral designer to prefer one rule to another.

In the process of investigating various voting rules, several axioms
were identified that are not only reasonable at the individual level
but also leave enough room for a wide class of procedures for
aggregating preferences of the society. Among them, one of the most
important is consistency,\footnote{In Smith's terminology,
  separability~\cite{smi:j:scoring-rules}.} introduced by
Smith~\cite{smi:j:scoring-rules} and adopted by
Young~\cite{young74}. Consistency says that if two societies decide on
the same set of options and if both societies prefer option $x$ to
option $y$, then the union of these two societies should also prefer
$x$ to $y$. Amazingly, together with the symmetry (which says that all
alternatives and all voters are treated equally) and
continuity,\footnote{In Smith's terminology,
  Archimedean~\cite{smi:j:scoring-rules}} consistency uniquely defines
the class of scoring social welfare
functions~\cite{smi:j:scoring-rules,young74}, which are also called
positionalist voting
rules~\cite{gardenfors73:scoring-rules,pattanaik2002positional} or,
perhaps more commonly, positional scoring rules. These rules are
defined as follows. Given voters' rankings over alternatives, each
alternative earns points from each voter's ranking depending on its
position in that ranking. The alternative $x$ is then at least as high
in the social order as $y$ if the total number of points that $x$
garnered from all the voters is at least as large as for~$y$. Young
also obtained a similar axiomatic characterization of social choice
functions~\cite{you:j:scoring-functions}, which, unlike social welfare
functions, determine only the winner(s). These characterizations of
scoring rules made it possible to axiomatize some particular scoring
rules, most notably Borda~\cite{youngBorda} and
Plurality~\cite{RichelsonPlurality} (see also the work of
Merlin~\shortcite{merlinAxiomatic} for a refined presentation of
Young's result, and the survey of Chebotarev and
Shamis~\cite{che-sha:j:scoring-rules} for a comprehensive list of
axiomatic characterizations of voting rules).

The study of single-winner voting rules is now
well-advanced~\cite{arrow2002handbook,arrow2010handbook}. This is not
the case for the multi-winner voting rules, i.e., for the rules that
aim at electing committees. The only success in their axiomatic study
was Debord's characterization of the $k$-Borda voting rule
\cite{deb:j:k-borda} by methods similar to Young's.
% PF: rephrased because committee scoring rules were not defined PF:
% before the first reference to them
In this paper we provide axiomatic characterization of committee
scoring rules---the multiwinner analogues of single-winner scoring
rules, recently introduced by Elkind et
al.~\shortcite{elk-fal-sko-sli:c:multiwinner-rules}---in the style of
Smith's and Young's 
% characterization 
results for the single-winner case~\cite{smi:j:scoring-rules,young74}.

% Moreover,
% committee scoring rules---the multiwinner analogues of single-winner
% scoring rules---have been introduced by Elkind et
% al.~\shortcite{elk-fal-sko-sli:c:multiwinner-rules} only recently.

%
%Committee scoring rules, recently introduced by Elkind et
%al.~\shortcite{elk-fal-sko-sli:c:multiwinner-rules}, are multiwinner
%rules that extend the idea of a single-winner
%scoring rule\footnote{A single-winner rule is a formal decision
%  process that given a set of preferences of the voters over the
%  candidates returns a winning candidate. We assume that preferences are
%  expressed by ranking the candidates from the best to the worst one.
%  A single-winner scoring rule gives to each candidate, for each
%  ballot, a number of points that depends on its position in the
%  ballot. The candidate with the most points is announced as the
%  winner.}  to the multiwinner setting. 

In our model of a multiwinner election, we are given a set of
candidates, a collection of voters with preferences over these
candidates, and an integer~$k$. A multiwinner voting rule is an
algorithm that allows us to compare any two committees (i.e., two
subsets of candidates of size~$k$) on the basis of preferences of the
voters and, in particular, it allows us to identify the best
committee.  In other words, multiwinner voting rules are assumed to
produce weak linear orders over the committees.  Multiwinner elections
of this type are interesting for a number of reasons, and, in
particular, due to a wide range of their applications.  For example,
we may use multiwinner elections to choose a country's parliament, to
identify a list of webpages a search engine should display in response
to a query, to choose locations for a set of facilities (e.g.,
hospitals or fire stations) in a city, to short-list a group of
candidates for a position or a prize, to decide which set of products
a company should offer to its customers (if there is a limited
advertising space), or even as part of a genetic
algorithm~\cite{fal-saw-sch-smo:c:multiwinner-genetic-algorithms}. There
are many other applications and we point the reader to the works of Lu
and Boutilier~\shortcite{budgetSocialChoice,bou-lu:c:value-directed},
Elkind et al.~\shortcite{elk-fal-sko-sli:c:multiwinner-rules}, and
Skowron et al.~\shortcite{owaWinner} for more detailed discussions of
them (including the applications mentioned above).
%with an AI focus.

%Unfortunately, so far our understanding of the properties of multiwinner voting rules is quite
%limited, especially in comparison with the single-winner voting rules that have been extensively studied. 
Multiwinner voting rules differ from their single-winner counterparts
in several important aspects.  First of all,
some %fundamental properties of
multiwinner voting rules take into account possible interdependence
between the committee members---the issue which does not exist when
the goal is to select a single winning candidate. The valuation of a
candidate may depend not only on the voters' preferences but also on
who the other committee members would
be. %When we shortlist and there are two equally strong candidates, then both should be included.
For example, in some cases it is important to diversify the committee,
e.g., when we are choosing locations for a set of facilities like
hospitals, when we are choosing a set of advertisements (within the
given budget) to reach the broadest possible audience of customers, or
when we want to provide a certain level of proportionality of
representation in a collective body such as a parliament.

%For instance,
%for achieving proportionality of representation one might prefer to
%include in a winning committee a Pareto dominated candidate (e.g., in case when the candidate who dominates her has already reached the capacity and cannot be assigned any more voters to represent).
%for the sake of achieving a better representation of a given minority.  
%Namely, a candidate supported by a minority of voters can be considered valuable only if the majority is already well
%represented. 
%  Diversification among committee members can be also
%desirable in the context of business applications. As an example,
%consider a company whose goal is to select a set of
%advertisements (within the given budget) to reach the broadest possible set of customers. 
%Diversification is also important if, for example, we are choosing a set of facility locations, e.g., hospitals.
%In this case diversification within the chosen set of advertisements is often desired.

%%% PF: Actually, we do not. Is there a point to do it? Aside advertisement?
%%% (we recall definitions of these rules in Section~\ref{sec:committeeScoringRules}).  

Identifying the class of committee scoring rules has been a recent,
important step on the way of getting a better understanding of multiwinner
voting rules~\cite{elk-fal-sko-sli:c:multiwinner-rules}.  Committee
scoring rules extend their single-winner counterparts as follows.  Let us recall that a
single-winner scoring rule is based on a scoring function that, given a position of a
candidate in a vote (that is, in the ranking of candidates provided by
the voter), outputs the number of points that the candidate gets from this particular voter.  %gives to this particular %candidate. 
The overall score of a candidate in the election is
the sum of the points she gets from all the votes, and the
candidate with the highest overall score wins.  
%(For example, the Borda score
%of candidate $c$ in vote $v$ is the number of candidates that $v$
%ranks below $c$; the Borda rule elects candidates with the highest
%total Borda score.)
%Committee scoring
%rules are a natural extension of this idea.  
In the case of committee scoring rules,
 we elect not a single candidate but a committee of size~$k$, so we
need a different notion of a position.
Specifically, we say that the position of a committee in a given vote
is the set of positions of its members in this vote.
%
% To this end, following Elkind et
% al.~\shortcite{elk-fal-sko-sli:c:multiwinner-rules}, we let the
% position of a committee be the set of positions of its members.
A committee scoring function then assigns points to each possible
position of a committee (with $m$ candidates and committee size $k$,
there are $m \choose k$ such committee positions) and the total score
of a committee is the sum of the points it gets from all the
voters. Then committee $X$ is at least as good as committee $Y$ if the
total number of points that committee $X$ receives is at least as
large as the number of points of committee $Y$.  We view committee
scoring rules as social welfare functions, generalized to the
multiwinner setting; in this respect, our approach is closer to that
of Smith~\cite{smi:j:scoring-rules}, Young~\cite{young74}, and
Merlin~\cite{merlinAxiomatic} rather than to that of
Young~\cite{you:j:scoring-functions}.

%The committee with the highest score
%wins.\footnote{Formally, we consider committee scoring rules as rules
%  that rank committees based on a given profile of votes. That is, a
%  committee scoring rule ranks the committees from the one with the
%  highest score to the one with the lowest one. This means that we
%  view committee scoring rules as social welfare functions,
%  generalized to the multiwinner setting. In this respect, our
%  approach is closer to that of Merlin~\cite{merlinAxiomatic} than to
%  that taken originally by Young~\cite{you:j:scoring-functions}.} 

%Given the above description, it is quite natural to expect that
%committee scoring rules are the natural generalization of
%single-winner scoring rules to the multiwinner setting (as much as one
%can judge naturalness of various notions). Yet, even in the world of
%scoring systems, multiwinner rules are much different from their
%single-winner variants.  

Although this generalization of single-winner voting rules to committee scoring rules is quite natural, one can expect much more diversity in the multiwinner case. And this is indeed the case:
the committee scoring rules form a remarkably wide class of
multiwinner election rules, which includes simple ``best $k$'' rules
such as SNTV or $k$-Borda (which select $k$ candidates that are ranked
first most often, or that have the highest Borda scores,
respectively), more involved rules, such as the Chamberlin--Courant
rule~\shortcite{ccElection} that focuses on providing proportional
representation, or even more complex selection procedures, such as the
variants of the OWA-based rules of Skowron et
al.~\shortcite{owaWinner} and Aziz et
al.~\shortcite{azi-gas-gud-mac-mat-wal:c:multiwinner-approval}, or
decomposable rules of Faliszewski et
al.~\cite{fal-sko-sli-tal:c:classification} with applications reaching
far beyond political science.

%In this paper, we would like to obtain a formal claim justifying that committee scoring rules are multiwinner analogous of single-winner scoring systems.
It is, therefore, remarkable that the committee scoring rules admit an axiomatic characterization very similar in spirit to 
%To this end, we consider 
the celebrated characterization of single-winner scoring rules. %~\cite{you:j:scoring-functions}. %(We explain the meaning of the axioms below.)
%
%\newtheorem*{theoremyoung}{Theorem A (Young's Characterization of
%  Single-Winner Scoring Rules)}
%\begin{theoremyoung}
%  A single-winner voting rule
%(i.e., a function that given an election
%  outputs a set of tied winners) 
%  is a scoring rule if and only if it is symmetric, consistent, and
%  continuous.
%   A single-winner voting rule
% %(i.e., a function that given an election
% %  outputs a set of tied winners) 
%   is a scoring rule if and only if it is symmetric (i.e., it treats
%   all the voters and all the candidates in a uniform way), consistent
%   (i.e., if a given candidate is a winner in two elections, then this
%   candidate also wins if we put these two elections together), and
%   continuous (i.e., for every election we can ensure a given
%   candidates' victory by adding sufficiently many copies of some
%   election where this candidate is a unique winner).
%\end{theoremyoung}
%
Our first main result is as follows.

\begin{maintheorem}[\bf Axiomatic Characterization of Committee Scoring Rules]\label{thm:theMainTheorem}
  A multiwinner voting rule is a committee scoring rule if and only if
  it is symmetric, consistent, continuous, and satisfies committee dominance.
\end{maintheorem}

% PF: revised a bit since we already introduced symmetry before
Let us give an informal description of the axioms in this
characterization and explain the appearance of committee dominance
among them. Symmetry, as in the single-winner case, simply means that
all the voters and candidates are treated in a uniform way.  This is a
standard, widely accepted requirement, and cannot be disputed if the
society adheres to the basic principles of equality both for the
voters and for the candidates.

The requirement of consistency is easily adapted to the multiwinner
case. It says that if there are two groups of voters and for both of
them our voting rule shows that committee $C_1$ is at least as good as
committee $C_2$, then the rule must show that $C_1$ is at least as
good as $C_2$ when the two groups join together in a single
electorate.  We saw that in the case of single-winner rules this
requirement is rather appealing.  Rejecting it would be difficult to
justify from the point of view of social philosophy as it would mean
that we treat large and small societies differently.

Let us now explain continuity. Again two societies are
involved. Suppose that for the first one the voting rule outputs that
committee $C_1$ is at least as good as committee $C_2$ and for the
second one it outputs the opposite conclusion, that committee $C_2$ is
strictly better than committee $C_1$. Then, if we join together the
first society and the second society cloned sufficiently many times,
then for the combined society the rule will output that $C_2$ is
strictly better than committee $C_1$.  That is, continuity ensures
that large enough majority of a population always gets its
choice.\footnote{Smith refers to continuity as the Archimedean
  property and this is a better name for it but, we stick to Young's
  terminology.}

Now, we move to the new axiom which we call committee dominance.  This
axiom requires that if there are two committees, $X$ and $Y$, such
that every voter can pair the candidates in $X$ with candidates in $Y$
into a sequence of pairs $(x_1,y_1),\ldots,(x_k,y_k)$ so that for
every pair $(x_i,y_i)$ this voter weakly prefers $x_i$ to $y_i$, then
the society weakly prefers $X$ to $Y$.  This is an incarnation, in the
case of multiwinner rules, of the famous Pareto Principle which is the
least disputable principle in social choice.  Any libertarian
philosopher would agree that if such a concept like social preference
is at all used, then it should be derived in some systematic way from
individual preferences, and this inevitably leads to the Pareto
Principle. The requirement of committee dominance is, in fact, a part
of the definition of committee scoring
rules~\cite{elk-fal-sko-sli:c:multiwinner-rules}, so it cannot be
avoided here.  Committee dominance can also be seen as a weak form of
monotonicity (see the works of Elkind et
al.~\cite{elk-fal-sko-sli:c:multiwinner-rules} and Faliszewski et
al.~\cite{fal-sko-sli-tal:c:classification} for extended discussions
of various multiwinner monotonicity notions). In his definition of
scoring functions, Young disregards monotonicity considerations and
his scoring functions can assign a higher score to a lower position,
but if one were to use the standard definition of a single-winner
scoring rule which is predominantly used in social choice and which
stipulates that a higher position yields a number of points that is at
least as high as for any lower position, then one would have to add to
Young's characterization an axiom enforcing the Pareto Principle too.
 
Unfortunately, the original Young's technique cannot be applied to prove Theorem~\ref{thm:theMainTheorem}.  
%which once again illustrates that multiwinner rules are distinct from their single-winner variants. 
%The reason is that some important properties of single-winner rules can be preserved
%under permutations of the set of candidates, but this is not the case for multiwinner rules.
%
Some observations critical to Young's approach cannot be extended to
multiwinner case.  For instance, Young's analysis heavily relies on
the fact that for any two ordered pairs of candidates $(a_1,a_2)$ and
$(b_1,b_2)$ there is a permutation of the set of candidates that maps
$a_1$ to $b_1$ and $a_2$ to $b_2$.  This however fails for two pairs
of committees $(C_1,C_2)$ and $(C_3,C_4)$ since the intersections
$C_1\cap C_2$ and $C_3\cap C_4$ may have different cardinalities. As a
result, the neutrality axiom (symmetry with respect to the candidates)
has much less bite in the context of multiwinner elections.

Our approach is based on the novel concept of a decision rule (or, a
$k$-decision rule if we fix the cardinality $k$ of the committees
involved).  Given a profile of the society and two committees of size
$k$, a $k$-decision rule tells us which committee is better for this
society (or that they are equally good).
% PF: I think that it is important to explicitly describe the difference
% PF: between decision rules and multiwinner rules
However, as opposed to our multiwinner rules, decision rules are not
required to be transitive (e.g., it is perfectly legal for a decision
rule to say that committee $C_1$ is better than $C_2$, that $C_2$ is
better than $C_3$, and that $C_3$ is better than $C_1$).
%  In other words, a decision
% rule is an instrument for pairwise comparisons of committees with no
% regard to transitivity of these comparisons.  
We note that all the properties of symmetry, consistency and
continuity are equally applicable to decision rules as to multiwinner rules.

In the class of decision rules, we distinguish the class of decision
scoring rules that is much broader than the class of committee scoring
rules.
A decision scoring rule stipulates that any linear order in the
profile, `awards' points (positive or negative) to pairs of
committees.  If in a ranking $v$ committees $C_1$ and $C_2$ have,
respectively, committee positions $I_1$ and $I_2$, then this pair of
committees gets $d(I_1,I_2)$ points from $v$, where $d$ is a certain
function that returns real values. The score of an ordered pair of
committees $(C_1,C_2)$ is the total number of points that this pair
gets from all linear orders of the profile. If the score is positive,
then $C_1$ is strictly preferred over $C_2$. If it is negative, then
$C_2$ is strictly preferred over $C_1$. Otherwise, if the score is
zero, the two committees are declared equally good.
%
% A decision scoring rule stipulates that any linear order in the
% profile, to any pair of ordered committee positions $(I_1,I_2)$
% awards $d(I_1,I_2)$ points, both positive and negative. An ordered
% pair of committees $(C_1,C_2)$ gets a score which is the total
% number of points that the corresponding pairs of their positions
% yield from all the voters. A decision scoring rule is a
% generalization of the concept of a majority relation in case $k=1$.
%
% If the total number of points is positive, then $C_1$ will be
% preferred to $C_2$ by the rule, if negative, then the other way
% around, if the number of points is zero, then the rule declares the
% indifference between $C_1$ and $C_2$.
%
Decision scoring rules, while a bit counterintuitive at first, are a
very general and useful notion. For example, one can easily show a decision scoring
rule that generates the standard majority relation, where alternative
$a$ is preferred to alternative $b$ if and only if more voters place
$a$ higher than $b$ than the other way around.

As indicated above, decision scoring rules are a very broad class that
goes far beyond committee scoring rules.  It is, therefore, quite
surprising that we can still obtain an axiomatic characterization for
them (especially that it uses the same axioms as Young's
characterization of single-winner scoring
rules~\cite{you:j:scoring-functions} adapted to the multiwinner
setting):\par\medskip

\begin{maintheorem}[\bf Axiomatic Characterization of Decision Scoring Rules]\label{thm:decision-rules}
  A decision rule is a decision scoring rule if and only if it is symmetric, consistent and continuous.
\end{maintheorem}

Since decision rules generalize the notion of the majority relation,
this result opens a possibility to use ideas from the theory of
tournament solution concepts in future research on multiwinner rules
(for an overview of tournament theory, see, e.g., the book of
Laslier~\cite{las:b:tournaments}).  Our theorem says that decision
scoring rules form the unique class of functions mapping voters'
preferences to tournaments and satisfying the above three axioms.

This paper is organized as follows.  First, in
Section~\ref{sec:related} we discuss related work and then, in
Section~\ref{sec:prelims}, we provide necessary background regarding
multiwinner elections and committee scoring rules. In
Section~\ref{sec:properties} we formally describe the axioms that we
use in our characterization. Section~\ref{sec:mainResult} contains our
main result and its proof. The proof is quite involved and is divided
into two parts.  First, we provide a variant of our characterization
for decision rules (for this part of the proof, we use a technique
that is very different from that used by Young). Second, we build an
inductive argument with Young's characterization providing us with the
induction base to obtain our final result (while this part of the
proof is inspired by Young's ideas, it uses new technical approaches
and tricks), using results from the first part as tools.  We conclude
in Section~\ref{sec:conclusions}.

%Our proof is inspired by that of Young as presented by
%Merlin~\shortcite{merlinAxiomatic}. However, we use new ideas that allow to
%obtain characterization in a significantly simpler way. These ideas can be
%also used to obtain a simpler proof for the case of single-winner rules. 

% On the technical side, we follow the steps of Young's proof, as
% presented by Merlin~\shortcite{merlinAxiomatic}. However, the
% multiwinner setting required us to extend and modify the proof and the
% exact axioms in a number of ways.

\section{Related Work}\label{sec:related}

Axiomatic characterizations of single-winner election rules have been
actively studied for quite a long time.  Indeed, the classical theorem
of Arrow~\cite{arrow1963} and the related, and equally important,
result of Gibbard~\cite{gib:j:polsci:manipulation} and
Satterthwaite~\cite{sat:j:polsci:manipulation} can be seen as
axiomatic characterizations of the dictatorial rule\footnote{Under the
  dictatorial voting rule, the winner is the candidate most preferred
  by a certain fixed voter (the dictator).}  (however, typically these
theorems are considered as impossibility results, taking the view of
an electoral designer).  Other well-known axiomatic characterizations
of single-winner rules include the characterizations of the majority
rule\footnote{The majority rule is defined for the set of two
  candidates only. It selects the one out of two candidates that is
  preferred by the majority of the voters.} due to
May~\cite{mayAxiomatic1952} and Fishburn~\cite{fishburn73SocChoice},
several different characterizations of the Borda
rule~\cite{youngBorda, hansson76, fishburnBorda, smi:j:scoring-rules} and the
Plurality rule~\cite{RichelsonPlurality, Ching1996298}, the
characterization of the Kemeny rule\footnote{The Kemeny rule, given
  the set of rankings over the alternatives, returns a ranking that
  minimizes the sum of the Kendall tau~\cite{kendall1938measure}
  distances to the rankings provided by the
  voters.}~\cite{lev-you:j:condorcet}, the
characterization of the Antiplurality rule~\cite{Barbera198249}, and
the characterizations of the approval voting rule\footnote{In the
  approval rule, each voter expresses his or her preferences by
  providing a set of approved candidates. A candidate that was
  approved by most voters is announced as the winner.} due to
Fishburn~\cite{fishburn78Approval} and Sertel~\cite{sertel88Approval}.
Freeman~et~al.~\cite{conf/aaai/FreemanBC14} proposed an axiomatic
characterization of runoff methods, i.e., methods that proceed in
multiple rounds and in each round eliminate a subset of candidates
(the single transferable vote (STV) rule, a
rule used, e.g., in Australia, is perhaps the best known example of
such a multistage elimination rule).  Still, in terms of axiomatic
properties, single-winner scoring rules appear to be the best
understood single-winner rules. Some of their axiomatic
characterizations were proposed by
G{\"a}rdenfors~\cite{gardenfors73:scoring-rules},
Smith~\cite{smi:j:scoring-rules} and
Young~\cite{you:j:scoring-functions} (we refer the reader to the
survey of Chebotarev and Shamis~\cite{che-sha:j:scoring-rules} for an
overview of these characterizations).  For a number of voting rules no
axiomatic characterizations are yet known.

Probabilistic single-winner election rules have also been a subject of
axiomatic studies.  For instance, Gibbard~\cite{10.2307/1911681}
investigated strategyproofness of probabilistic election systems and
blue his result can be seen as an axiomatic characterization of the
random dictatorship rule. Brandl~et~al.~\cite{Bran13a}, by studying
different types of consistency of probabilistic single-winner election
rules, characterized the function returning maximal lotteries, first
proposed by Fishburn~\cite{Fish84a}.

% Unfortunately, not all single-winner
% runoff rules have elegant (or, even, not so elegant) axiomatic
% characterizations. 

% For example, to the best of our knowledge there are
% no such characterizations of multiround methods, where in each round
% the voters need to cast their votes again, over a smaller set of
% candidates (possibly ranking these candidates differently from the
% previous rounds).  The rules used in presidential elections in Poland
% and in France are examples of such election systems.

The state of research on axiomatic characterizations  of multiwinner voting rules is
far less advanced. Indeed, we are aware of only one unconditional
characterization of a multiwinner rule: Debord has characterized the
$k$-Borda rule as the only rule that satisfies neutrality,
faithfulness, consistency, and the cancellation
property~\cite{deb:j:k-borda}.
%
% Nonetheless,
% there are also characterizations of rules within certain classes. For
% example, Faliszewski et al.~\cite{fal-sko-sli-tal:c:top-k-counting}
% characterize Bloc with the class of committee scoring rules, and Aziz
% et al.~\cite{justifiedRepresenattion} characterized Proportional
% Approval Voting rule within a certain class of approval-based rules.
%
Yet, there exists an interesting line of research, where the
properties of multiwinner election rules are studied. A large bulk of
this literature focuses on the principle of Condorcet
consistency~\cite{bar-coe:j:non-controversial-k-names,
  kay-san:j:condorcet-winners, fis:j:majority-committees,
  rat:j:condorcet-inconsistencies}, and on approval-based multiwinner
rules~\cite{kil-handbook, kil-mar:j:minimax-approval,
  azi-gas-gud-mac-mat-wal:c:multiwinner-approval,justifiedRepresenattion}. Properties
of other types of multiwinner election rules have been studied by Felsenthal
and Maoz~\cite{fel-mao:j:norms},
Elkind~et~al.~\shortcite{elk-fal-sko-sli:c:multiwinner-rules},
and---in a somewhat different
context---Skowron~\cite{skow:multiwinner-models}.

In their effort to analyze axiomatic properties of multiwinner rules,
Elkind et al.~\cite{elk-fal-sko-sli:c:multiwinner-rules} introduced
the notion of committee scoring rules, the main focus of the current
work. Committee scoring rules were later studied axiomatically and
algorithmically by
Faliszewski~et~al.~\cite{fal-sko-sli-tal:c:classification,fal-sko-sli-tal:c:top-k-counting}.
In particular, they have identified many
interesting subclasses of committee scoring rules and found that most committee scoring rules
are $\np$-hard to compute, but in many cases there are good
approximation algorithms (the work on the complexity of committee
scoring rules can be traced to the studies of the complexity of the
Chamberlin--Courant rule, initiated by Procaccia, Rosenschein and
Zohar~\cite{complexityProportionalRepr} and continued by Lu and
Boutilier~\cite{budgetSocialChoice}, Betzler et
al.~\cite{fullyProportionalRepr}, and Skowron et
al.~\cite{sko-fal-sli:j:multiwinner}).
The axiomatic part of the works of Faliszewski et
al.~\cite{fal-sko-sli-tal:c:classification,fal-sko-sli-tal:c:top-k-counting},
has lead, in particular, to characterizations of several multiwinner
voting rules within the class of committee scoring rules.  They showed
that SNTV is the only nontrivial weakly separable
representation-focused rule, Bloc is the only nontrivial weakly
separable top-$k$-counting rule, and the $k$-approval-based
Chamberlin--Courant rule is the only nontrivial representation-focused
and top-$k$-counting rule.\footnote{These characterizations are, in a
  sense, syntactic, because the properties they rely on describe
  syntactic features of committee scoring functions. Faliszewski et
  al.~\cite{fal-sko-sli-tal:c:classification,fal-sko-sli-tal:c:top-k-counting}
  also provide some semantic characterizations. For example, within
  the class of committee scoring rules, a rule is weakly separable if
  and only if it is non-crossing monotone and, if a rule is
  fixed-majority consistent, then it is top-k-counting. In effect,
  they characterize the Bloc rule as a committee scoring rule that is
  non-crossing monotone and fixed-majority consistent.}  (For brevity,
we omit exact description of these properties here and point the
readers to the original papers.) 

Skowron et al.~\cite{owaWinner} has studied a family of multiwinner
rules that are based on utility values of the alternatives instead of
preference orders, and where these utilities are aggregated using
ordered weighted average operators (OWA operators) of
Yager~\cite{yager1988}.  (The same class, but for approval-based
utilities, first appeared in early works of the Danish polymath
Thorvald N. Thiele~\cite{Thie95a} and was later studied by Forest Simmons\footnote{See the
  description in the overview of Kilgour~\cite{kil-handbook}.} and
Aziz~et~al.~\cite{azi-gas-gud-mac-mat-wal:c:multiwinner-approval,justifiedRepresenattion}).
It is easy to express these OWA-based rules as committee scoring
rules.

% PF: already discussed above
%
% Apart from committee scoring rules, there are many other, quite distinct,
% multiwinner election systems. For example, there are rules based on the
% extension of the single-winner Condorcet principle\footnote{Condorcet
%   principle says that if for some candidate $c$ it holds that for
%   every other candidate $d$ a majority of voters prefers $c$ to $d$,
%   then $c$ should be the unique winner of the given single-winner election.}
% (see, e.g., the works of
% Fishburn~\shortcite{fis:j:majority-committees}, Kaymak and
% Sanver~\shortcite{kay-san:j:condorcet-committee},
% Ratliff~\shortcite{rat:j:condorcet-inconsistencies}, and Barber\'a and
% Coelho~\shortcite{barberaNonControversial}), there are rules based on
% approval ballots, where voters express which candidates are and which
% are not acceptable for them (see, e.g., the works of
% Kilgour~\shortcite{kil-handbook}, Kilgour and
% Marshall~\shortcite{kil-mar:j:minimax-approval}, and Aziz et
% al.~\shortcite{justifiedRepresenattion}), and there are numerous other
% rules based on ordinal preference orders. 

As we mentioned
in the introduction, the decision rules---studied
in~Section~\ref{sec:nontransitiveRules}---can be seen as
generalizations of majority relations in the case of single-winner
elections. In the world of single-winner elections, majority relations
are often seen as inputs to election procedures (known as tournament solution
concepts). For example, according to the Copeland
method~\cite{cop:m:copeland} the candidate with the greatest number of
victories in pairwise comparisons with other candidates is a
winner. The Smith set~\cite{smi:j:scoring-rules} is another example of
such a rule: it returns the minimal (in terms of inclusion) subset of
candidates, such that each member of the set is preferred by the
majority of voters over each candidate outside the
set. Fishburn~\cite{fishburn77socChooiceFunctions} describes nine
other tournament solution concepts that explore the
Condorcet principle for majority graphs. For an overview of
tournament solution
concepts we refer the reader to the book of
Laslier~\cite{las:b:tournaments} (and to the chapter of Brandt, Brill,
and Harrenstein~\cite{BrandtEtAlChapter} for a more computational
perspective).  We believe that it would be a fascinating topic of
research to explore the properties (computational or axiomatic) of the
generalized tournament solutions for multiwinner rules generated by
our decision rules.

Intransitive preference relations have also been studied by Rubinstein~\cite{rubinstein80Tournament} and by Nitzan and Rubinstein~\cite{nitzan81BordaChar}.
Rubinstein~\cite{rubinstein80Tournament} shows axiomatic characterization of scoring systems among rules which take input preferences
in the form of tournaments, i.e., complete, assymetric (possibly intransitive) relations. 
Nitzan and Rubinstein~\cite{nitzan81BordaChar}, on the other hand, provide axiomatic characterization of the Borda rule,
assuming each voter can have intransitive preferences.

% The decision rules that we study in
% Section~\ref{sec:nontransitiveRules} can also be seen as mechanisms
% for aggregating and comparing sets of candidates (of the same size).

%Elkind et al.~\cite{elk-fal-sko-sli:c:multiwinner-rules} and
%Faliszewski et
%al.~\cite{fal-sko-sli-tal:c:top-k-counting,fal-sko-sli-tal:c:classification}
%discuss a number of subclasses of committee scoring rules, showing
%that each of the above rules is, in some well-defined sense, rather
%different from the other ones. Indeed, using their nomenclature (not
%described here), SNTV is the only nontrivial separable
%representation-focused rule, Bloc is the only nontrivial weakly
%separable top-$k$-counting rule, $k$-Borda is separable (but neither
%representation-focused nor top-$k$-counting), and the Borda-based
%Chamberlin--Courant rule is representation-focused (but neither
%separable nor top-$k$-counting).  Elkind et
%al.~\cite{elk-fal-sko-sli:c:multiwinner-rules} and Faliszewski et
%al.\cite{fal-sko-sli-tal:c:top-k-counting,fal-sko-sli-tal:c:classification}
%analyze these subclasses of committee scoring rules and show that they
%differ quite significantly both in terms of axiomatic properties and
%in terms of computational properties of the rules included in them (we
%provide more discussion in Section~\ref{sec:related}).  %Yet, in this
%work we deal with the class of committee scoring rules as a whole and
%do not consider its internal structure. 

\section{Multiwinner Voting and Decision Rules}\label{sec:prelims}

In this section we provide necessary background regarding multiwinner
elections and committee scoring rules, as well as a definition of our
novel concept of decision rules. For each positive integer $t$, by
$[t]$ we mean the set $\{1, \ldots, t\}$, and by $[t]_k$ we mean the
set of all $k$-element subsets of $[t]$.  For each set $X$ and each $k
\in \naturals$, by $S_k(X)$ we denote the set of all $k$-element
subsets of $X$ (so, in particular, we have that $S_k([t]) = [t]_k$).
For a given set $X$, by $\Pi_{>}(X)$ and $\Pi_{\geq}(X)$ we denote the
set of all linear orders over $X$ and the set of all weak orders over
$X$, respectively.

\subsection{Multiwinner Elections}\label{sec:prelim1}
Let $A = \{a_1, \ldots, a_m\}$ be the set of all the candidates, and
let $N = \{1, 2, \ldots\}$ be the set of all possible voters. We refer
to the members of $\powA$ as size-$k$ committees, or, simply, as
committees, when $k$ is clear from the context. For each finite subset
$V \subseteq N$, by $\profiles(V)$ we denote the set of all
$|V|$-element tuples of elements from $\orders$, indexed by elements
of $V$. We refer to elements of $\profiles(V)$ as preference profiles
for the set of voters $V$. We set $\profiles = \{P \in \profiles(V):
V$ is a finite subset of $N\}$ to be the set of all possible
preference profiles. For each preference profile $P \in
\profiles$, by $\vot(P)$ we denote the set of all the voters in $P$
(in particular, we have that for each $P \in \profiles(V)$ it holds
that $\vot(P)=V$). For each profile $P$ and each voter $v \in
\vot(P)$, by $P(v)$ we denote the preference order of $v$~in~$P$.

% PF: rephrased to avoid repetitions
Our proof approach crucially relies on using what we call
\emph{$k$-decision rules}.  A \emph{$k$-decision rule} $f_k$,
\[ 
  f_k\colon \profiles \to \Big(\powA \times \powA \to \{-1, 0, 1\}\Big),
\]
is a function that for each preference profile $P \in \profiles$
provides a mapping, $f_k(P) \colon \powA \times \powA \to \{-1, 0,
1\}$, such that for each two size-$k$ committees $C_1$ and $C_2$ it
holds that $f_k(P)(C_1,C_2) = -f_k(P)(C_2,C_1)$.
%  that compares size-$k$ committees.  We require
%   that for each profile $P$ and each two size-$k$ committees $C_1$ and
%   $C_2$ it holds that $f_k(P)(C_1,C_2) = -f_k(P)(C_2,C_1)$.
% % \colon \powA \times \powA \to \{-1, 0, 1\}$, such that $f_k(P)(C_1,C_2) = - f_k(P)(C_2,C_1)$ for all $P\in\profiles$ and all $C_1,C_2\in S_k(A)$.
We interpret $f_k(P)(C_1,C_2) = 1$ as saying that at profile $P$ the
society prefers committee $C_1$ over committee $C_2$ and we denote
this as $C_1 \succ_P C_2$ (we omit $f_k$ from this notation because it
will always be clear from the context).  Similarly, we interpret
$f_k(P)(C_1,C_2) = 0$ as saying that at profile $P$ the society views
the committees as equally good (denoted $C_1 =_P C_2$) and
$f_k(P)(C_1,C_2) = -1$ as saying that at profile $P$ the society
prefers $C_2$ to $C_1$ (denoted $C_2 \succ_P C_1$).  We write $C_1
\succeq_P C_2$ if $C_1 \succ_P C_2$ or $C_1 =_P C_2$, which is
equivalent to $f_k(P)(C_1, C_2) \ge 0$.
%
%
% Informally, $f_k$ compares size-$k$ committees from the societal point
% of view, when the profile of the society is $P$. Given a preference
% profile $P \in \profiles$ and two size-$k$ committees $C_1, C_2 \in
% \powA$, the decision rule declares that, at profile $P$, committee
% $C_1$ is preferred to committee $C_2$ if and only if $f_k(P)(C_1, C_2)
% = 1$ and we write in this case $C_1 \succ_P C_2$. If $f_k(P)(C_1, C_2)
% = 0$ we say that our voting rule declares these two committees as
% equally good which we denote as $C_1 =_P C_2$.  We write $C_1
% \succeq_P C_2$ if $C_1 \succ_P C_2$ or $C_1 =_P C_2$, which is
% equivalent to $f_k(P)(C_1, C_2) \ge 0$.}
 %
Sometimes, when $P$ is a more involved expression, we write $C_1\:
{\succeq}[P]\: C_2$ instead of $C_1 \succeq_P C_2$ and $C_1\: {=}[P]\:
C_2$ instead of $C_1 =_P C_2$.
%(Note that for decision rules we do not require the
%relations $\succ_P$ and $\succeq_P$ to be transitive; however, we will add this
%requirement for multiwinner election rules.)
%
By $C_1 \prec_P C_2$ we mean $C_2 \succ_P C_1$, and by $C_1 \preceq_P
C_2$ we mean $C_2 \succeq_P C_1$.

% Since the just-introduced notation does not specify the rule $f_k$
% explicitly, the reader may be worried about possible confusion.
% While, on one hand, including $f_k$ in the notation would make the
% discussion more explicit, it would also make the notation even more
% heavy than it already is.\footnote{As the reader will see in
% Section~\ref{sec:nontransitiveRules}, our notation is quite explicit
% and---unfortunately---quite heavy in a number of places. We
% experimented with several approaches and the current one is a result
% of many weeks of work on improving the clarity and readability of
% the manuscript.}  Since most of the time we use the notation with
% respect to a single $k$-decision rule (fixed at the beginning of our
% proof), we believe that the benefits of the abbreviated notation
% outweigh the drawbacks.

A \emph{$k$-winner election rule} $f_k$ is a $k$-decision rule that
additionally satisfies the transitivity requirement, i.e., it is a
$k$-decision rule such that for each profile $P$ and each three
committees $C_1$, $C_2$, and $C_3$ of size $k$ it satisfies the
following condition:
\[ C_1 \succeq_P C_2\ \  \text{and}\ \  C_2 \succeq_P C_3 \implies C_1 \succeq_P
C_3 \textrm{.}
\]
A multiwinner election rule $f$ is a family $(f_k)_{k \in \naturals}$
of $k$-winner election rules, with one $k$-winner rule for each
committee size~$k$. We remark that often multiwinner rules are defined
to simply return the set of winning committees, whereas in our case
they implicitly define weak orders over all possible committees of a
given size. Since the number of such committees is huge, we believe
that giving a concise algorithm for comparing committees---this is
what a transitive decision rule is---is the right way to define a
multiwinner analog of a social welfare function.

%We provide some examples of multiwinner rules in the next section.

% We note that election rules that output (weak) orders over some
% space of possible winners rather than sets of (tied) winner are also
% referred to as social choice correspondences in the literature.

\subsection{Committee Scoring Rules and Decision Scoring
  Rules}\label{sec:committeeScoringRules}

\paragraph{Committee Scoring Rules}
For a preference order $\pi \in \orders$, by $\pos_{\pi}(a)$ we denote
the position of candidate $a$ in $\pi$ (the top-ranked candidate has
position $1$ and the bottom-ranked candidate has position $m$). A
single-winner scoring function $\gamma\colon [m]\to \reals$ assigns a
number of points to each position in a preference order so that
$\gamma(i)\geq \gamma(i+1)$ for all $i\in [m-1]$.  For example, the
Borda scoring function, $\beta \colon [m] \rightarrow
\naturals$, is defined as $\beta(i) = m-i$.  Similarly, for each $t
\in [m]$, we define the $t$-Approval scoring function, $\alpha_t$, so
that $\alpha_t(i) = 1$ for $i \leq t$ and $\alpha_t(i) = 0$
otherwise. $1$-Approval scoring function is known as the plurality
scoring function.

We extend the notion of a position of a candidate to the case of
committees in the following way. For a preference order $\pi \in
\orders$ and a committee $C \in \powA$, by $\pos_\pi(C)$ we mean the
%
% PF: IMPORTANT! In this paper we define a position of a committee 
% PF: to be a set of positions of its members. In other papers we
% PF: indeed used the definition based on sorted sequences, but here
% PF: we made the choice early on and now it would be very difficult
% PF: to change it. The notions, of course, are equivalent. I suppose
% PF: that sets of positions are a bit simpler and this slightly easier
% PF: to work with in this paper, whereas sorted sequences are FAR 
% PF: EASIER if one wants to define specific rules (e.g., defining PAV
% PF: without a sorted sequence is one bucketfull of pain). 
%
%element of $[m]_k$ obtained from the 
set $\pos_\pi(C) = \{\pos_\pi(a):
a \in C\}$.
% by sorting its elements in the increasing order.  
By a \emph{committee scoring function for committees of size $k$}, we
mean a function $\posf\colon [m]_k \to \reals$, that for each element
of $[m]_k$, interpreted as a position of a committee in some vote,
assigns a score. A committee scoring function must also satisfy the
following dominance requirement. %(An
%intuitive understanding is that the score measures the satisfaction of
%a voter from a committee that he or she ranks at the given position.)
Let $I$ and $J$ be two sets from $[m]_k$ (i.e., two possible committee
positions) such that $I = \{i_1, \ldots, i_k\}$, $J = \{j_1, \ldots,
j_k\}$ with $i_1 < \cdots < i_k$ and $j_1 < \cdots < j_k$. We say that
$I$ dominates $J$ if for each $t \in [k]$ we have $i_t \leq j_t$ (note
that this notion might be referred to as ``weak dominance'' as well,
since a set dominates itself).  If $I$ dominates $J$, then we require that $\posf(I) \geq \posf(J)$.
For each set of voters $V \subseteq N$ and each preference profile $P
\in \profiles(V)$, by $\scorefull{\posf}{C}{P}$ we denote the total score that the
voters from $V$ assign to committee $C$. Formally, we have that
$\scorefull{\posf}{C}{P} = \sum_{v \in \vot(P)} \posf(\pos_{P(v)}(C))$.
By a \emph{committee scoring function} we mean a family of committee
scoring functions, one for each possible size of the committee.

\begin{definition}[Committee scoring rules]
A multiwinner election rule is a \emph{committee scoring rule} if
there exists a committee scoring function $\posf$ such that for each
two equal-size committees $C_1$ and $C_2$, we have that $C_1 \succ_P
C_2$ if and only if $\scorefull{\posf}{C_1}{P} > \scorefull{\posf}{C_2}{P}$, and $C_1 =_P C_2$ if
and only if $\scorefull{\posf}{C_1}{P} = \scorefull{\posf}{C_2}{P}$.  
\end{definition}

Committee scoring rules were introduced by Elkind et
al.~\shortcite{elk-fal-sko-sli:c:multiwinner-rules} and were later
studied by Faliszewski et
al.~\cite{fal-sko-sli-tal:c:top-k-counting,fal-sko-sli-tal:c:classification}
(closely related notions were considered by Thiele~\cite{Thie95a},
Skowron~et~al.~\cite{owaWinner} and by
Aziz~et~al.~\cite{justifiedRepresenattion,azi-gas-gud-mac-mat-wal:c:multiwinner-approval}).
Below we present some examples of committee scoring rules by
specifying the actions of the corresponding committee scoring
functions on $I=\{i_1,\ldots,i_k\}$ with $i_1 < \cdots < i_k$:
\begin{enumerate}
\item The single non-transferable vote (SNTV) rule uses scoring
  function $\posf_\sntv(I) = \sum_{t=1}^k\alpha_1(i_t)$. In other
  words, for a given voter it assigns score $1$ to every committee
  that contains her highest-ranked candidate, and it assigns score zero otherwise.
  That is, SNTV selects the committee of $k$ candidates with the highest
  plurality scores.
  
\item The Bloc rule uses function $\posf_\Bloc(I) =
  \sum_{t=1}^k\alpha_k(i_t)$, i.e., the score a committee gets from a
  single vote is the number of committee members positioned among the
  top $k$ candidates in this vote. Bloc selects the committee with the
  highest total score accumulated from all the voters (and one can see
  that it selects $k$ candidates with the highest $k$-Approval
  scores).
  
\item The $k$-Borda rule uses function $\posf_\kBorda(I) =
  \sum_{t=1}^k\beta(i_t)$, i.e., the score a committee gets from a
  single vote is the sum of the Borda scores of the committee
  members. It selects the committee with the highest total score (and
  one can see that this committee consists of $k$ candidates with the
  highest Borda scores).
  
\item The classical Chamberlin--Courant rule~\shortcite{ccElection} is
  defined through the scoring
  function $\posf_{\beta\text{-}\CC}(I) = \beta(i_1)$ (recall
  that we assumed that $i_1 < \cdots < i_k$). Intuitively, under the
  Chamberlin--Courant rule the highest-ranked member of a committee is
  treated as the representative for the given voter, and this voter
  assigns the score to the committee equal to the Borda score of his
  or her representative.

\item The
  $t$-Approval-Based Proportional Approval Voting rule~\cite{Thie95a,
    kil-handbook, justifiedRepresenattion,
    azi-gas-gud-mac-mat-wal:c:multiwinner-approval, owaWinner} (the
  $\alpha_t$-PAV rule, for short) is defined by the scoring function
  $\posf_{\alpha_t\text{-}\pav}(I) = \sum_{j = 1}^{k}
  \frac{1}{j}\alpha_t(i_j)$.  Thus, the score that a voter $v$ assigns
  to a committee $C$ increases (almost) logarithmically with the
  number of members of $S$ located among the top $t$ preferred
  candidates by $v$.  The use of logarithmic function, implemented by
  the sequence of harmonic weights $(1, \nicefrac{1}{2},
  \nicefrac{1}{3}, \ldots)$ ensures some interesting properties
  pertaining to proportional representation of the
  voters~\cite{justifiedRepresenattion}, and allows one to view
  $\alpha_t$-PAV as an extension of the d'Hondt method of
  apportionment~\cite{Puke14a} to the setting where voters can vote
  for individual candidates rather than for
  parties~\cite{per:pav-dhondt}.
\end{enumerate}
Naturally, many other rules can be interpreted as committee scoring
rules; Faliszewski et
al.~\cite{fal-sko-sli-tal:c:top-k-counting,fal-sko-sli-tal:c:classification}
provide specific examples.

\paragraph{Decision Scoring Rules}
\emph{Decision scoring rules} are our main example of $k$-decision
rules. These rules are similar to committee scoring rules, but with
the difference that the scores of two committees cannot be computed
independently.
% PF: commented out because it is repetitive wrt the next sentence
% and, instead, the score is assigned to every pair. 
Specifically, for each pair of committee positions $(I_1,I_2)$ we
define a numerical value, the score that a voter assigns to the pair
of committees $(C_1, C_2)$ under the condition that $C_1$ and $C_2$
stand in this voter's preference order on positions $I_1$ and $I_2$,
respectively. If the total score of a pair of committees $(C_1, C_2)$
is positive, then $C_1$ is preferred over $C_2$; if it is negative,
then $C_2$ is preferred over $C_1$; if it is equal to zero, then $C_1$
and $C_2$ are seen by this decision rule as equally good.

\begin{definition}[Decision scoring rules]\label{def:cpsr}
  Let $d\colon [m]_k \times [m]_k \to \reals $ be a \emph{decision
    scoring function}, that is, a function that for each pair of
  committee positions $(I_1, I_2)$, where $I_1, I_2 \in [m]_k$,
  returns a score value (possibly negative), such that for each $I_1$
  and $I_2$, it holds that $d(I_1, I_2) = -d(I_2, I_1)$. For each
  preference profile $P \in \profiles$ and for each pair of committees $(C_1, C_2)$, we define
  the score:
  \begin{align}
    \label{def_D_P}
    \pairscorefull{d}{C_1}{C_2}{P} = \sum_{v \in \vot(P)}d(\pos_{P(v)}(C_1), \pos_{P(v)}(C_2))
    \textrm{.}
  \end{align}
  A $k$-decision rule is a \emph{decision scoring rule} if there
  exists a decision scoring function $d$ such that for each preference
  profile $P$ and each two committees $C_1$ and $C_2$ it holds that:
  (i) $C_1 \succeq_P C_2$ if and only if
  $\pairscorefull{d}{C_1}{C_2}{P} \geq 0$, and (ii) $C_1 =_P C_2$ if
  and only if $\pairscorefull{d}{C_1}{C_2}{P} = 0$.
\end{definition}

As we have indicated throughout the introduction and the related work
discussion, one of the arguments in favor of decision scoring rules is
that they generalize the notion of a majority relation: For committee
size $k=1$ we define $d_\maj(\{i_1\},\{i_2\})=1$ if $i_1<i_2$ and
$d_\maj(\{i_1\},\{i_2\})=-1$ if $i_1>i_2$.  A candidate $x$ is
preferred to candidate $y$ if and only if more voters place $x$ ahead
of $y$ than the other way around.  Naturally, each committee scoring
rule is an example of (a transitive) decision scoring rules as well.

\section{Axioms for Our Characterization}\label{sec:properties}

In this section we describe the axioms that we use in
our characterization of committee scoring rules.  %Apart from committee dominance, 
The properties expressed by these axioms are
natural, straightforward generalizations of the respective properties 
from the world of single-winner elections.
%
% None of the considered properties requires
% transitivity, so we formulate them for $k$-decision rules. 
%
%However,
%since each $k$-winner election rule is a $k$-decision rule, these
%properties are relevant also in the context of $k$-winner election
%rules.  
%
We formulate them for the case of $k$-decision rules (for a given
value of $k$), but since $k$-winner rules are a type of $k$-decision
rules, the properties apply to $k$-winner rules as well. For each of
our properties $\mathfrak{P}$, we say that a multiwinner election rule
$f = \{f_k\}_{k \in \naturals}$ satisfies $\mathfrak{P}$ if $f_k$
satisfies $\mathfrak{P}$ for each $k \in \naturals$.

We start by recalling the definitions of anonymity and neutrality,
these two properties ensure that the election is fair to all voters
and all candidates. %Intuitively,
Anonymity means that none of the voters is neither privileged nor
discriminated against, whereas neutrality says the same for the
candidates.
% , i.e., the names of the
% voters have no influence on their voting power and the rule cares only
% about the set of votes and not about who cast them. Neutrality, on the
% other hand, says that all the candidates are treated in the uniform
% way.

\begin{definition}[Anonymity]
  We say that a $k$-decision rule $f_k$ is \emph{anonymous} if for
  each two (not necessarily different) sets of voters $V, V' \subseteq
  N$ such that $|V| = |V'|$, for each bijection $\rho\colon V \to V'$ and
  for each two preference profiles $P_1 \in \profiles(V)$ and $P_2 \in
  \profiles(V')$ such that $P_1(v) = P_2(\rho(v))$ for each $v \in V$,
  it holds that $f_k(P_1) = f_k(P_2)$.
\end{definition}

% Let $\sigma$ be a permutation of the set of candidates $A$ and let
% $\pi \in \orders$ be a linear order over $A$. By $\sigma(\pi)$ we
% denote the linear order such that for each two candidates $a$ and
% $b$ we have $a \mathrel{\pi} b \iff \sigma(a) \mathrel{\sigma(\pi)}
% \sigma(b)$. For a given profile $P$, we define $\sigma(f_k(P))$ in
% the natural, analogous way.

Let $\sigma$ be a permutation of the set of candidates $A$.  For a
committee $C$, by $\sigma(C)$ we mean the committee $\{ \sigma(c)
\colon c \in C \}$.  For a linear order $\pi \in \orders$, by
$\sigma(\pi)$ we denote the linear order such that for each two
candidates $a$ and $b$ we have $a \mathrel{\pi} b \iff \sigma(a)
\mathrel{\sigma(\pi)} \sigma(b)$. For a given $k$-decision rule $f_k$
and profile $P$, by $\sigma(f_k(P))$ we mean the function such that
for each two size-$k$ committees $C_1$ and $C_2$ it holds that
$\sigma(f_k(P))(\sigma(C_1),\sigma(C_2) = f_k(P)(C_1,C_2)$.

\begin{definition}[Neutrality]
  A $k$-decision rule $f_k$ is \emph{neutral} if for each permutation
  $\sigma$ of $A$ and each two preference profiles $P_1, P_2$
  over the same set of voters $V$, such that $P_1(v) = \sigma(P_2(v))$ for each $v \in
  V$, it holds that $f_k(P_1) = \sigma( f_k(P_2) )$.
\end{definition}

A rule that is anonymous and neutral is called \emph{symmetric}.  We
note that our definition of anonymity resembles the ones used by
Young~\shortcite{you:j:scoring-functions} or
Merlin~\shortcite{merlinAxiomatic} rather than the traditional ones,
as presented by May~\shortcite{mayAxiomatic1952} or
Arrow~\shortcite{arrow1963}.  The difference comes from the fact that
we (just like Young and Merlin) need to consider elections with
variable sets of voters. The next axiom, in particular, describes a situation
where  two elections with disjoint sets of voters are merged. Given two profiles $P$
  and $P'$ over the same set of alternatives and with disjoint sets of voters, by
  $P+P'$ we denote the profile that consists of all the voters from $P$
  and $P'$ with their respective preferences.

% , who assume that the set of voters is
% fixed. In our case, however, decision rules may operate on arbitrary
% subsets of the set of all voters. Such approach and such definition of
% rules allows us to define axioms that describe what happens if we
% merge different sets of voters. The example of such an axiom is
% consistency, defined below.

\begin{definition}[Consistency]
  A $k$-decision rule $f_k$ is \emph{consistent} if for each two
  profiles $P$ and $P'$ over disjoint sets of voters, $V \subset N$
  and $V' \subset N$, %$V \cap V' = \emptyset$, 
  and each two committees
  $C_1, C_2 \in \powA$, (i) if $C_1 \succ_P C_2$ and $C_1
  \succeq_{P'} C_2$, then it holds that $C_1 \succ_{P+P'} C_2$, and (ii) if $C_1 \succeq_P C_2$ and $C_1
  \succeq_{P'} C_2$, then it holds that $C_1 \succeq_{P+P'} C_2$.
\end{definition}

In some sense, consistency is the most essential element of Young's
characterization, that distinguishes single-winner scoring rules from
the other single-winner rules. In particular, it
means that the rule treats small electorates in the same way as it
treats large ones.

In the framework of social welfare functions (and our decision rules
are analogues of those) it is important to distinguish consistency and
reinforcement axioms.  If we were to express the reinforcement axiom
in our language, it would be worded in the same way as the consistency
axiom, except that the conclusion would only apply to profiles $P$ and
$P'$ such that $f_k(P) = f_k(P')$ (i.e., when the entire rankings
produced by the rule for profiles $P$ and $P'$ coincide).
% That is, reinforcement puts a condition on the overall result of a
% decision rule, provided that the rule is applied to some profile
% $P+P'$ such that the result of comparing any two committees is the
% same both under $P$ and under $P'$.
On the other hand, the premise of consistency requires only that
$f_k(P)$ and $ f_k(P')$ agree on the ranking of $C$ and $C'$ which is
a much weaker requirement than $f_k(P) = f_k(P')$.  As a result, the
consistency axiom is much stronger than the reinforcement axiom.
For example, Kemeny's social welfare function satisfies reinforcement
but not consistency. (We point the reader to the work of Young and
Levenglick for an axiomatic characterization of the Kemeny's
rule~\cite{lev-you:j:condorcet} using the reinforcement axiom.)

\begin{remark}\label{rem:using-consistency}
  In our proofs, we often use the consistency axiom in the following
  way. Let $C_1$ and $C_2$ be two committees and let $P$ and $Q$ be
  two profiles over disjoint voter sets, such that $C_1 \succ_{P+Q}
  C_2$ and $C_1 =_P C_2$. Using consistency, we conclude that $C_1
  \succ_Q C_2$. Indeed, if, for example, it were the case that $C_2
  \succeq_Q C_1$ then by consistency (as applied to merging profiles $P$
  and $Q$) we would have to conclude that $C_2 \succeq_{P+Q} C_1$ which is
 the opposite to what is  assumed.
\end{remark}

The next axiom concerns the dominance relation between committee
positions and specifies a basic monotonicity condition (it can also be
viewed as a form of Pareto dominance).

\begin{definition}[Committee Dominance]
  A $k$-decision rule $f_k$ has the committee dominance property if
  for every profile $P$ and every two committees $C_1, C_2 \in \powA$
  the following holds: If for every vote $v \in \vot(P)$ we have that
  $\pos_{P(v)}(C_1)$ dominates $\pos_{P(v)}(C_2)$, then $C_1 \succeq_P
  C_2$.
% PF rephrased to maintain consistency with the wording of other
% PF axioms
  % Let $f_k$ be a $k$-decision rule. Suppose that for any profile $P$
  % and any two committees $C_1, C_2 \in \powA$ the fact that
  % $\pos_{P(v)}(C_1)$ dominates $\pos_{P(v)}(C_2)$ for any vote $v$ in
  % $P$ implies $C_1 \succeq_P C_2$. Then the $k$-decision rule $f_k$ is
  % said to have the \emph{committee dominance} property.
\end{definition}

The definition of committee scoring rules requires that if $\lambda$
is a committee scoring function (for committee size~$k$) and $I$ and
$J$ are two committee positions such that $I$ dominates $J$, then
$\lambda(I) \geq \lambda(J)$. That is, committee scoring rules have
the committee dominance property by definition and, thus, we include
this axiom in our characterization.
Young~\cite{you:j:scoring-functions} did not include axioms of this
form because he allowed scoring functions to assign lower scores to
higher positions.

% PF: rephrase because almost the same text was in the introduction
% Thus the requirement of committee
% dominance is, in fact, a part of the definition of a committee scoring
% rule~\cite{elk-fal-sko-sli:c:multiwinner-rules} and it cannot be
% avoided here. In his definition of scoring functions Young disregards
% monotonicity considerations and his scoring functions can assign a
% higher score to a lower position. This is why he did not need an axiom
% analogous to committee dominance.  }

Finally, we define the continuity property, which ensures that if a
certain set of voters $V$ prefers $C_1$ over $C_2$, then for each set
of voters $V'$, disjoint from $V$, there exists some (possibly large)
number $n$, such that if we clone $V$ exactly $n$ times and add such
a profile to $V'$, then in this final profile $C_1$ will be preferred to
$C_2$ (note that when we speak of cloning voters, we implicitly assume
that the decision rule is anonymous and that the identities of the
cloned voters do not matter). Thus, continuity might be viewed as a
kind of ``large enough majority always gets its choice'' principle.

\begin{definition}[Continuity]
  An anonymous $k$-decision rule $f_k$ is \emph{continuous} if for
  each two committees $C_1, C_2 \in \powA$ and each two profiles $P_1$
  and $P_2$ where $C_1 \succ_{P_2} C_2$, there exists a number $n \in
  \naturals$ such that for the profile $Q=P_1+nP_2$ (that consists of the
  profile $P_1$ and of $n$ copies of the profile $P_2$), it holds that
  $C_1 \succ_{Q} C_2$.
\end{definition}

Although we call this axiom continuity (after Young), we note that
there are many axioms of this nature in decision theory, where they
are called ``Archimedean.''  Such axioms usually rule out the
existence of parameters which are infinitely more important than some
other parameters; mathematically, this is usually expressed in terms
of rules that use lexicographic orders (see, for example, axiom A3 in
the work of Gilboa, Schmeidler, and Wakker~\cite{gilboa2002utility} on
Case-Based Decision Theory).
% of the paper on the Case-Based Decision Theory
% \cite{gilboa2002utility} which are close in spirit to our paper and
% the motivation for it).
In Young's characterization the continuity
axiom plays a similar role. For more discussion on continuity, we
refer the reader to the original work of
Young~\shortcite{you:j:scoring-functions}.  
%\par\medskip

\section{Proofs of Main Results}\label{sec:mainResult}

We now start
proving our main results---the axiomatic characterizations of
committee scoring rules and of decision scoring rules, i.e.,
Theorems~\ref{thm:theMainTheorem} and~\ref{thm:decision-rules}. In
fact, Theorem~\ref{thm:decision-rules} will be proved first and will
serve as an intermediate step in proving
Theorem~\ref{thm:theMainTheorem}. Here is the roadmap of the proof.

Since anonymity allows us to ignore the order of linear orders in
profiles, in Section~\ref{sec:nontransitiveRules} we change the domain
of our rules from the set of preference profiles to
the set of voting situations. A voting situation is an
$m!$-dimensional vector with non-negative integers specifying how many
times each linear order representing a vote repeats in the voters'
preferences.  We use this new representation of the domain of decision
rules in Section~\ref{sec:nontransitiveRules} and we conclude this
section by proving Theorem~\ref{thm:decision-rules}.

In Section~\ref{sec:toolsforcsrs} we further extend the
domain of our rules to generalized voting
situations, allowing fractional and negative multiplicities of linear
orders; the voting situations in such extended domain can then be
identified with the elements of $\rationals^{m!}$. We use
characterization from Section~\ref{sec:nontransitiveRules} to prove
that for each symmetric, consistent,
committee-dominant, continuous $k$-winner election rule $f_k$ and
each two committees of size $k$
the set of voting situations for which $C_1$ and $C_2$ are equivalent
is a hyperplane in $\rationals^{m!}$. This will be an important
technical tool in the subsequent proof. In particular, this
observation will be used in Lemma~\ref{thm:nontransitive2} which
implies that for the proof of Theorem~\ref{thm:theMainTheorem} it
would be sufficient to find a committee scoring rule that correctly
identifies the voting situations for which given committees are equivalent
under $f_k$.

In Sections~\ref{sec:second-part-a} and~\ref{sec:specialCase_base} we
then concentrate solely on proving Theorem~\ref{thm:theMainTheorem}.
%The second part of the proof is included in
%Sections~\ref{sec:second-part-a} and~\ref{sec:specialCase_base}.  
In Section~\ref{sec:second-part-a}, we prove
Theorem~\ref{thm:theMainTheorem} for the case where $f$ is used to
recognize in which profiles a certain committee $C_1$ is preferred
over some other committee $C_2$, when $|C_1 \cap C_2| = k-1$. If $|C_1
\cap C_2| = k-1$ then there are only two candidates, let us refer to
them as $c_1$ and $c_2$, such that $C_1 = (C_1 \cap C_2) \cup
\{c_1\}$, and $C_2 = (C_1 \cap C_2) \cup \{c_2\}$. Thus, this case
closely resembles the single-winner setting, studied by
Young~\shortcite{you:j:scoring-functions} and
Merlin~\shortcite{merlinAxiomatic}. For each two candidates $c_1$ and
$c_2$, Young and Merlin %show a base
present a basis of the vector space of preference profiles that
satisfies the following two properties:
\begin{enumerate}
\item[(i)] For each preference profile in the basis, the scores of $c_1$ and $c_2$ are
  equal according to every possible scoring function.
\item[(ii)] Candidates $c_1$ and $c_2$ are ``symmetric'' and, thus,
  every neutral and anonymous voting rule has to judge them as equally
  good.
\end{enumerate}
These observations allow one to use geometric arguments to note that
the set of profiles in which $c_1$ is preferred over $c_2$ can be
separated from the set of profiles in which $c_2$ is preferred over
$c_1$ by a hyperplane. The coefficients of the linear equation that
specifies this hyperplane define a single-winner scoring rule, and
this scoring rule is exactly the voting rule that one started with.
In Section~\ref{sec:second-part-a} we use the same geometric
arguments, but the construction of the appropriate basis is more
sophisticated. Indeed, finding this
basis is the core technical part of Section~\ref{sec:second-part-a}.

In Section~\ref{sec:specialCase_base} we extend the result from
Section~\ref{sec:second-part-a} to the case of any two committees
(irrespective of the size of their intersection), concluding the
proof. Here, finding an appropriate basis seems
even harder and, consequently, we use a different technique. To deal
with committees $C_1$ and $C_2$ that have fewer than $k-1$ elements in
common, we form a third committee, $C_3$, whose intersections with
$C_1$ and $C_2$ have more elements than the intersection of $C_1$ and
$C_2$. Then, using an inductive argument, we conclude that the space
of profiles $P$ where $C_1=_P C_3$ is $(m!-1)$-dimensional, and that
the same holds for the space of profiles $P$ such that $C_2 =_P C_3$.
An intersection of two vector spaces with this dimension has dimension
at least $m!-2$ and, so, we have a subspace of profiles $P$ such that
$C_1 =_P C_2$ whose dimension is at least $(m-2)!$.  Using
combinatorial arguments, we find a profile $P'$ which does not belong
to the space but for which $C_1 =_{P'} C_2$ still holds. This gives us
our $(m-1)!$-dimensional space. By applying results from the first
part of the proof, this suffices to conclude that the committee
scoring function that we found in Section~\ref{sec:second-part-a} for
committees that differ in at most one element works for all other
committees as well.

\subsection{Characterization of Decision
  Rules}\label{sec:nontransitiveRules}

We start our analysis by considering $k$-decision rules.  Recall that
the outcomes of $k$-decision rules do not need to be transitive.  That
is, for a $k$-decision rule $f_k$ it is possible to have a profile $P$
and three committees such that $C_1 \succ_P C_2$, $C_2 \succ_P C_3$,
and $C_3 \succ_P C_1$. The remaining part of this section is devoted
to proving Theorem~\ref{thm:decision-rules}.

The whole discussion, i.e., this and the following sections, is
divided into small subsections, each with a title describing its main
outcome. These section titles are intended to help the reader navigate
through the proof, but otherwise one can read the text as a continuous
piece. In particular, all the notations, conventions, and definitions
carry over from one subsection to the next, and so on.

\paragraph{Setting up the Framework.}
Let us fix, for the rest of the proof, a positive integer $k$, the
size of the committee to be elected, and a symmetric, consistent,
continuous $k$-decision rule $f_k$. Our immediate goal is to show that
this rule must be a decision scoring rule. For this, we
need to find a function $d\colon [m]_k \times [m]_k \to \reals $ such
that for each profile $P$ and each two committees $C_1$, $C_2$ it
holds that $C_1 \succeq_P C_2$ if and only if
$\pairscorefull{d}{C_1}{C_2}{P} \geq 0$.

Our function $d$ will be piecewise-defined.  For each $s \in [k]$ we
will define a function $d_s$ which applies only to pairs $(I_1,I_2)\in
[m]_k \times [m]_k$ satisfying $|I_1\cap I_2|=s$, outputs real values
and such that the score:
\begin{align}
    \label{def_D_P_s}
    \pairscorefull{d_s}{C_1}{C_2}{P} = \sum_{v \in \vot(P)}d_s(\pos_{P(v)}(C_1), \pos_{P(v)}(C_2))
\end{align}
calculated with the use of this function satisfies the following
condition: if $|C_1\cap C_2|=s$, then $C_1 \succeq_P C_2$ if and only
if $\pairscorefull{d_s}{C_1}{C_2}{P} \geq 0$.  Pursuing this idea, for
the rest of the proof we will fix $s$ and restrict ourselves to pairs
of committees satisfying $|C_1\cap C_2|=s$. The restriction of $f_k$
to such pairs of committees will be denoted $f_{k,s}$.
%The function $d_s$ will be called $d$ as it will cause no confusion.

%FIND

\paragraph{The First Domain Change.} Anonymity of $f_k$ allows us to
use a more convenient domain for representing preference profiles.
Indeed, under anonymity the order of votes in any profile is
no longer meaningful and the outcome of any symmetric rule is fully
determined by the {\em voting situation} that specifies how many times each
linear order repeats in a given profile. In particular, for any $\pi \in \orders$ and voting
situation $P$, by $P(\pi)$ we mean the number of voters in $P$ with
preference order $\pi$.  Fixing some order on possible votes from
$\orders$, a voting situation can, thus, be viewed as an
$m!$-dimensional vector with non-negative integer coefficients.

Correspondingly, we can view $f_k$ as a function:
\[ 
   f_k: \naturals^{m!} \to \big(\powA \times \powA \to \{-1, 0, 1\}\big),
\]
with the domain $\naturals^{m!}$ instead of $\profiles$ (recall the
definition of a $k$-decision rule in Section~\ref{sec:prelim1}). 
%Each element $P$ of $\naturals^{m!}$, referred to as a \emph{voting
%  situation}, describes how many voters cast each of the $m!$ possible
%different votes. We view voting situations as vectors of nonnegative
%integers, indexed with preference orders from $\orders$ (that is, 
Representing profiles by voting situations will be helpful in our
further analysis, since algebraic operations on vectors from
$\naturals^{m!}$ become meaningful: for a voting situation $P$ and a
constant $c \in \naturals$, $cP$ is the voting situation that
corresponds to $P$ in which each vote was replicated $c$
times. Similarly, for two voting situations $P$ and $Q$, the sum $P+Q$
is the voting situation obtained by merging $P$ and $Q$. Subtraction
of voting situations can sometimes be meaningful as well.

Given a voting situation $P$, when we speak of ``some vote $v$ in
$P$,'' we mean ``some preference order that occurs within $P$.''  We
sometimes treat each vote $v$ (i.e., each preference order) as a
standalone voting situation that contains this vote only. When we say
that we modify some vote within some voting situation $P$, we mean
modifying only one copy of this vote, and not all the votes that have
the same preference order.

Let $d'\colon [m]_k \times [m]_k \to
\reals $ be some decision scoring function. Naturally, we can speak of
applying the corresponding decision scoring rule to voting situations
instead of applying them to preference profiles as in \eqref{def_D_P}.
For a voting situation $P \in \profiles$, the score of a committee
pair $(C_1, C_2)$ is:
\begin{align}
  \label{def_D_P_vot_sit}
  \pairscorefull{d'}{C_1}{C_2}{P} = \sum_{v \in \orders} P(v) \cdot d'(\pos_{v}(C_1), \pos_{v}(C_2)) \textrm{.}
\end{align}
%(Recall that for $v \in \orders$, $P(v)$ is the number of voters with preference order $v$ within $P$.)
 
%\vspace{-3mm}

\paragraph{Independence of Committee Comparisons from Irrelevant Swaps.}
We will now show that for
each two committees $C_1$ and $C_2$, the result of their comparison
according to $f_k$ depends only on the positions on which $C_1$ and
$C_2$ are ranked by the voters (and do not depend on the positions of
candidates not belonging to $C_1 \cup C_2$). In particular, if a
committee $C_1$ is better than committee $C_2$ in some election, then
it will also be better after we permute the set of candidates in some
of the votes but without changing the positions of committees $C_1$
and $C_2$ in these votes.

For $v\in \orders$, we write $v[a \leftrightarrow b]$ to denote the
vote obtained from $v$ by swapping candidates $a$ and $b$. Further, if
$v$ is a vote in $P$, by $P[v, a \leftrightarrow b]$ we denote the
voting situation obtained from $P$ by swapping $a$ and $b$ in $v$, and
by $P[a \leftrightarrow b]$ we denote the voting situation obtained
from $P$ by swapping $a$ and $b$ in every vote.

\begin{lemma}\label{lemma:othersIrrelevant}
  Let $C_1$ and $C_2$ be two size-$k$ committees, $P$ be a voting
  situation, $a,b$ be two candidates such that one of the following conditions is satisfied:
  \begin{inparaenum}[(i)]
  \item $a, b \notin C_1 \cup C_2$,
  \item $a, b \in C_1 \cap C_2$,
  \item $a, b \in C_1 \setminus C_2$, or
  \item $a, b \in C_2 \setminus C_1$.
  \end{inparaenum} Then for
  each vote $v$ in $P$, $C_1 \succeq_{P} C_2$ if and only if $C_1
  \succeq_{P[v, a \leftrightarrow b]} C_2$.
\end{lemma}
\begin{proof}
  Let us assume that $C_1 \succeq_{P} C_2$. Our
  goal is to show that $C_1 \succeq_{P[v, a \leftrightarrow b]} C_2$,
  so for the sake of contradiction we assume that $C_2 \succ_{P[v, a
    \leftrightarrow b]} C_1$ holds.

  We rename the candidates so that $C_1 \setminus C_2 = \{a_1, \ldots,
  a_{\ell}\}$ and $C_2 \setminus C_1 = \{b_1, \ldots, b_{\ell}\}$, and
  we define $\sigma$ to be a permutation (over the set of candidates)
  that for each $x \in [\ell]$ swaps $a_{x}$ with $b_{x}$, but leaves all the other candidates
  intact.
  That is, $\sigma(a_i) = b_i$ and $\sigma(b_i) = a_i$ for all $i\in [\ell]$, and for each
  candidate $c \notin \{a_1, \ldots, a_\ell, b_1, \ldots, b_\ell\}$ it
  holds that $\sigma(c) = c$.  Since $C_1 \succeq_{P} C_2$, by
  neutrality we have that $C_2 \succeq_{\sigma(P)} C_1$. Due to our
  assumptions, it holds that $C_2 \succ_{P[v, a \leftrightarrow b]}
  C_1$ and, by consistency, %, we get that:
  \begin{equation}\label{eq:iis:1}
     C_2 \:{\succ}{\big[\sigma(P) + P[v, a \leftrightarrow b]\big]}\: C_1.
  \end{equation}

  Let $Q = v[a \leftrightarrow b] + \sigma(v)$ be a voting situation
  that consists just of two votes, $v[a \leftrightarrow b]$ and
  $\sigma(v)$.  We observe that $\sigma(P) - \sigma(v) = \sigma(P[v, a
  \leftrightarrow b] - v[a \leftrightarrow b])$. This is because $P[v,
  a \leftrightarrow b] - v[a \leftrightarrow b]$ is the same as $P -
  v$. Since:
  \[ 
     \sigma(P) + P[v, a \leftrightarrow b] - Q\ \ = \ \ \underbrace{(\sigma(P) - \sigma(v))}_{R'} \ \ +\ \ \underbrace{(P[v, a \leftrightarrow b] - v[a \leftrightarrow b])}_{R''},
  \]
  and %the two voting situations in the final sum
  both summands on the right-hand-side are symmetric with respect to
  $\sigma$ (i.e., $\sigma(R')=R''$, $\sigma(R'')=R'$, and $\sigma^2$
  is an identity permutation),
%  is an involution, i.e., $\sigma^2$ is an
%  identity permutation) and committees $C_1$ and $C_2$, 
  by symmetry of $f_k$ we have: %that:
  \begin{equation}\label{eq:iis:2}
    C_2 \:{=}{[\sigma(P) + P[v, a \leftrightarrow b] - Q]}\: C_1.
  \end{equation}
  Thus, by consistency---as applied to
  equations~\eqref{eq:iis:1} and~\eqref{eq:iis:2} in the way described
  in Remark~\ref{rem:using-consistency}---we get that $C_2 \succ_{Q}
  C_1$. By neutrality, we also infer that $C_2 \succ_{Q[ a
    \leftrightarrow b]} C_1$.
  %(by $Q[ a \leftrightarrow b]$ we mean a
 % voting situation that is identical to $Q$ except that we swap $a$
 % and $b$ in every preference order). 
  This follows because for each of the four conditions for $a, b$ from
  the statement of the lemma it holds that permutation $a
  \leftrightarrow b$ maps committee $C_1$ to committee $C_1$ and
  committee $C_2$ to committee $C_2$.  Next, by consistency we get
  that
 % \begin{align*}
 $
     C_2 \:{\succ}_{\big[Q + Q[ a \leftrightarrow b]\big]}\: C_1.
 $
%  \end{align*}
 However, we observe that:
  \begin{align*}
     Q + Q[ a \leftrightarrow b] &= \Big(v[a \leftrightarrow b] + \sigma(v)\Big) + \Big(v + \sigma(v)[a \leftrightarrow b]\Big) \\
                                 &= \underbrace{\Big(v[a \leftrightarrow b] + \sigma(v)[a \leftrightarrow b]\Big)}_{Q'} + \underbrace{\Big(v + \sigma(v)\Big)}_{Q''} \textrm{.}
 \end{align*} 
 Furthermore, if $a, b \notin C_1 \cup C_2$, or $a, b \in C_1 \cap
 C_2$, then $\sigma(v)[a \leftrightarrow b] = \sigma(v[a
 \leftrightarrow b])$. On the other hand, if $a, b \in C_1 \setminus
 C_2$ or $a, b \in C_2 \setminus C_1$, then $\sigma(v)[a
 \leftrightarrow b] = (\sigma \circ [a \leftrightarrow b]) (v[a
 \leftrightarrow b])$. In other words, there always exists a
 permutation $\tau$ such that $Q' = \tau(Q')$, $C_1 = \tau(C_2)$, and
 $C_2 = \tau(C_1)$ ($\tau$ is either $\sigma$ or $\sigma \circ [a
 \leftrightarrow b]$), and, similarly, we have $Q'' = \sigma(Q'')$,
 $C_2 = \sigma(C_1)$, $C_1 = \sigma(C_2)$.  Thus, by neutrality, we
 get that:
 \begin{align*}
   C_2 \:{=}{\big[v[a \leftrightarrow b] + \sigma(v[a \leftrightarrow
     b])\big]}\: C_1 \quad \text{and} \quad C_2 \:{=}{\big[v +
     \sigma(v)\big]}\: C_1 \text{.}
 \end{align*}
 Thus, by consistency, we infer that $C_2 \:{=}_{\big[Q + Q[ a
   \leftrightarrow b]\big]}\: C_1$, which contradicts our previous conclusion.
 This completes the proof.
\end{proof}

\paragraph{Putting the Focus on Two Fixed Committees.}
Recall that we have assumed $f_k$ to be symmetric, consistent and
continuous.  Now, we fix a pair of size-$k$ committees, $C_1$ and
$C_2$ with $|C_1\cap C_2|=s$, and define $f_{C_1, C_2}$
to be the rule that acts on voting situations in the same way as
$f_{k,s}$ does, but with the difference that it only distinguishes, at
any voting situation $P$, whether (i)
$C_1$ is preferred over $C_2$, or (ii) $C_1$ and $C_2$ are seen as
equally good, or (iii) $C_2$ is preferred over $C_1$. In other words,
we set $f_{C_1,C_2}(P)$ to be $-1$, $0$ or $1$ depending on
$f_{k,s}(P)$ ranking $C_1$ lower than, equally to, or higher than $C_2$, respectively
($f_{k,s}$ can be viewed as the collection of rules $f_{C_1', C_2'}$,
one for each possible pair of committees $C_1'$ and $C_2'$).

%%% PF: rho was used for a permutation---we should pick a different name
%We define function $f_{k, |C_1 \cap C_2|}$ so that for each voting situation $P$,
%relation $f_{k, |C_1 \cap C_2|}(P)$ is the restriction of $f_k(P)$ to pairs of
%committees $C'_1, C'_2$ such that $|C'_1 \cap C'_2| = |C_1\cap
%C_2|$.

\newcommand{\spcvote}[4]{v(#1 \rightarrow #2, #3 \rightarrow #4)}

\paragraph{Defining Distinguished Profiles.}
For each two committee positions $I_1$ and $I_2$ such that $|I_1 \cap
I_2| = |C_1 \cap C_2|=s$, we consider a single-vote voting situation
$\spcvote{C_1}{I_1}{C_2}{I_2}$, where $C_1$ and $C_2$ are ranked on
positions $I_1$ and $I_2$, respectively, and all the other candidates
are ranked arbitrarily, but in some fixed, predetermined order.

Let us consider two cases. First, let us assume that for each two
committee positions $I_1$ and $I_2$ such that $|I_1 \cap I_2| = s$, it
holds that $C_1$ is as good as $C_2$ relative to $\spcvote{C_1}{I_1}{C_2}{I_2}$,
i.e., $C_1 \: {=}\big[\spcvote{C_1}{I_1}{C_2}{I_2}\big] \: C_2$.
% holds for each two
%committee positions $I_1$, $I_2$ such that $|I_1 \cap I_2| = |C_1 \cap C_2|$). 
By Lemma~\ref{lemma:othersIrrelevant}, we infer that for each
single-vote voting situation $v$ we have $C_1 =_v C_2$ (because in any
vote the set of positions shared by $C_1$ and $C_2$ always has the
same cardinality $s$). Further, by consistency, we conclude that
$f_{C_1, C_2}$ is trivial, i.e., for every voting situation $P$ it
holds that $C_1 =_P C_2$. By neutrality, we get that $f_{k,s}$ is also
trivial (i.e., it declares equally good each two committees whose
intersection has $s$ candidates). Of course, in this case $f_{k, s}$
is a decision scoring rule (with trivial scoring function
$d_s(I_1,I_2) \equiv 0$).

If the above case does not hold, then there are some two committee
positions, $I_1^*$ and $I_2^*$, such that $|I_1^* \cap I_2^*| = s$ and
$C_1$ is not equivalent to $C_2$ relative to
$\spcvote{C_1}{I_1^*}{C_2}{I_2^*}$. Without loss of generality we
assume that:
\begin{equation}
\label{main_assumption}
C_1 \: {\succ}[\spcvote{C_1}{I_1^*}{C_2}{I_2^*}] \: C_2. 
\end{equation}
We note that, by neutrality, this implies:
\begin{equation}
\label{main_assumption_consequence}
C_2 \: {\succ}[\spcvote{C_2}{I_1^*}{C_1}{I_2^*}] \: C_1. 
\end{equation}

Let us fix any two such $I_1^*$ and $I_2^*$ for now. As we will see
throughout the proof, any choice of $I_1^*$ and
$I_2^*$ with the aforementioned property will suffice for our arguments.

For each two committee
positions $I_1$ and $I_2$ with $|I_1 \cap I_2| = |I_1^* \cap I_2^*| =
s$, and for each two nonnegative integers $x$ and $y$, we define the
following voting situation:
\[ 
P_{x(C_1 \rightarrow I_1, C_2 \rightarrow I_2)}^{y(C_1 \rightarrow
  I_1^*, C_2 \rightarrow I_2^*)}
= 
y \cdot \big( \spcvote{C_1}{I_1^*}{C_2}{I_2^*} \big)
+
x \cdot \big( \spcvote{C_1}{I_1}{C_2}{I_2} \big),
\] 
where there are $y$ voters that rank $C_1$ and $C_2$ on positions
$I_1^*$ and $I_2^*$, respectively, and there are $x$ voters that rank
$C_1$ and $C_2$ on positions $I_1$ and $I_2$, respectively.

\paragraph{Deriving the Components for the Decision Scoring
  Function for \boldmath{$f_{C_1,C_2}$.}}
We now proceed toward defining a decision scoring function for
$f_{k,s}$. To this end, we define the value $\Delta_{I_1, I_2}$ as:
\begin{align}
\label{eq:diffDefinition}
\Delta_{I_1, I_2}  =
  \begin{cases}
    \:\:\:\:\!\! \sup \Big\{\frac{y}{x} \colon C_2 \; {\succ}\Big[P_{x(C_1 \rightarrow I_2, C_2 \rightarrow I_1)}^{y(C_1 \rightarrow I_1^*, C_2 \rightarrow I_2^*)}\Big]\; C_1, \;\; x, y \in \naturals \Big\}       & \hspace{-2mm} \text{for } C_1 \; {\succ}\big[ \spcvote{C_1}{I_1}{C_2}{I_2} \big] \; C_2,\\[2mm]
    -\inf \Big\{\frac{y}{x} \colon C_2  \;{\succ}\Big[P_{x(C_1 \rightarrow I_2, C_2 \rightarrow I_1)}^{y(C_1 \rightarrow I_2^*, C_2 \rightarrow I_1^*)}\Big]\; C_1, \;\; x, y \in \naturals \Big\}       & \hspace{-2mm} \text{for } C_2 \; {\succ}\big[ \spcvote{C_1}{I_1}{C_2}{I_2}  \big] \; C_1,\\[2mm]
    \phantom{-} 0  & \hspace{-2mm} \text{for } C_1 \; {=}\big[ \spcvote{C_1}{I_1}{C_2}{I_2} \big] \; C_2 \textrm{.}\\ 
  \end{cases}
\end{align}
This definition certainly might not seem intuitive at first. However,
we will show that the values $\Delta_{I_1, I_2}$, for all possible
$I_1$ and $I_2$ with $|I_1\cap I_2|=s$, in essence, define a decision
scoring function for $f_{k,s}$. The next few lemmas should build an
intuition for the nature of these values.
However, let us first argue that the values $\Delta_{I_1, I_2}$ are
well defined. Let us fix some committee positions $I_1$ and $I_2$ (such
that $|I_1 \cap I_2| = |I_1^* \cap I_2^*|=s$). Due to continuity of $f_k$, we see
that the appropriate sets in Equation~\eqref{eq:diffDefinition} are
non-empty. For example, if $C_1 \:
{\succ}\big[\spcvote{C_1}{I_1}{C_2}{I_2}\big] \: C_2$ and, thus, $C_2
\: {\succ}\big[ \spcvote{C_1}{I_2}{C_2}{I_1}\big] \: C_1$, then
continuity of $f_k$ ensures that there exists (possibly large) $x$
such that $C_2 \:{\succ}\Big[P_{x(C_1 \rightarrow I_2, C_2 \rightarrow
  I_1)}^{(C_1 \rightarrow I_1^*, C_2 \rightarrow I_2^*)}\Big]\:
C_1$. This proves that the set from the first condition
of~\eqref{eq:diffDefinition} is nonempty. An analogous reasoning
proves the same fact for the set from the second condition
in~\eqref{eq:diffDefinition}. Further, we claim that the value
$\Delta_{I_1, I_2}$ is finite.  This is evident for the case where we
take the infimum over the set of positive rational numbers. For the
case where we take the supremum, this follows from
Lemma~\ref{lemma:reverseDelta}, below.
%
% Further, we
% show that for each two committee positions $I_1$ and $I_2$,
% $\Delta_{I_2, I_1} = -\Delta_{I_1, I_2}$.

\begin{lemma}
\label{lemma:reverseDelta}
For each two committee positions $I_1$ and $I_2$ with $|I_1 \cap I_2|
= s$, it holds that $\Delta_{I_2, I_1} = -\Delta_{I_1, I_2}$.
\end{lemma}
\begin{proof}
  We assume, without loss of generality, that $C_1 \; {\succ}\big[
  \spcvote{C_1}{I_1}{C_2}{I_2} \big] \; C_2$.\footnote{This assumption
    is without loss of generality because the condition from the
    statement of the lemma, $\Delta_{I_2, I_1} = -\Delta_{I_1, I_2}$,
    is symmetric; if it held that $C_2 \;
    {\succ}\big[\spcvote{C_1}{I_1}{C_2}{I_2}\big] \; C_1$ then we
    could simply swap $I_2$ and $I_1$, and we would prove that
    $\Delta_{I_1, I_2} = -\Delta_{I_2, I_1}$.}  Let us consider two
  sets:
  \begin{align*}
    U = \Big\{\frac{y}{x} \colon C_2 \:{\succ}\Big[P_{x(C_1 \rightarrow I_2,
      C_2 \rightarrow I_1)}^{y(C_1 \rightarrow I_1^*, C_2 \rightarrow
      I_2^*)}\Big]\: C_1, \ \ x, y \in \naturals \Big\}
  \end{align*}
  ($U$ is the set that we take supremum of in
  Equation~\eqref{eq:diffDefinition}), and:
  \begin{align*}
    L = \Big\{\frac{y}{x} \colon C_2 \:{\succ}\Big[P_{x(C_1 \rightarrow I_1,
      C_2 \rightarrow I_2)}^{y(C_1 \rightarrow I_2^*, C_2 \rightarrow
      I_1^*)}\Big]\: C_1, \ \ x, y \in \naturals \Big\}
  \end{align*}
  (thus, $L$ is the set that we take infimum of in
  Equation~\eqref{eq:diffDefinition}, for $\Delta_{I_2,I_1}$). We will
  show that $\sup U = \inf L$. First, we show that $\sup U \leq \inf
  L$. For the sake of contradiction, let us assume that this is not
  the case, i.e., that there exists $\frac{y}{x} \in U$ and
  $\frac{y'}{x'} \in L$ such that $\frac{y}{x} > \frac{y'}{x'}$.
  Since $\frac{y}{x} \in U$ and $\frac{y'}{x'} \in L$, we get:
  \[
  C_2 \:{\succ}\Big[P_{x(C_1 \rightarrow I_2, C_2 \rightarrow
    I_1)}^{y(C_1 \rightarrow I_1^*, C_2 \rightarrow I_2^*)}\Big]\: C_1
  \quad\text{and}\quad
  C_2 \:{\succ}\Big[P_{x'(C_1 \rightarrow I_1, C_2 \rightarrow
    I_2)}^{y'(C_1 \rightarrow I_2^*, C_2 \rightarrow I_1^*)}\Big]\:
  C_1.
  \]
  Let us consider the voting situation:
  \begin{align*}
    S = 
    y'\cdot P_{x(C_1 \rightarrow I_2, C_2 \rightarrow I_1)}^{y(C_1
      \rightarrow I_1^*, C_2 \rightarrow I_2^*)} +
    y \cdot P_{x'(C_1 \rightarrow I_1, C_2 \rightarrow I_2)}^{y'(C_1
      \rightarrow I_2^*, C_2 \rightarrow I_1^*)} \textrm{.}
  \end{align*}
  By consistency, we have that $C_2 \succ_S C_1$.  However, let us
  count the number of voters in $S$ that rank committees $C_1$ and
  $C_2$ on particular positions. There are $yy'$ voters that rank
  $C_1$ and $C_2$ on positions $I_1^*$ and $I_2^*$, respectively, and
  the same number $yy'$ of voters that rank $C_1$ and $C_2$ on
  positions $I_2^*$ and $I_1^*$. Due to neutrality and consistency,
  these voters cancel each other out. (Formally, if $S'$ were a voting
  situation limited to these voters only, we would have $C_1 =_{S'}
  C_2$. This is so due to the symmetry of $f_k$ and the fact that for
  any permutation $\sigma$ that swaps all the members of $C_1
  \setminus C_2$ with all the members of $C_2\setminus C_1$, we have
  $S' = \sigma(S')$.)  Next, there are $x'y$ voters that rank $C_1$
  and $C_2$ on positions $I_1$, and $I_2$, and $xy'$ of voters that
  rank $C_1$ and $C_2$ on positions $I_2$ and $I_1$, respectively.
  Since we assumed that $\frac{y}{x} > \frac{y'}{x'}$, we have that
  $x'y > xy'$.  So, $xy'$ voters from each of the two aforementioned
  groups cancel each other out (in the same sense as above), and we
  are left with considering $x'y-xy' > 0$ voters that rank $C_1$ and
  $C_2$ on positions $I_1$ and $I_2$. Thus, we conclude that $C_1
  \succ_S C_2$.  However, this contradicts the
  fact that $C_2 \succ_S C_1$ and we conclude that $\sup U \leq \inf
  L$.
  
  % However, in $S$ there is the same number of voters who rank $C_1$
  % and $C_2$ on positions $I_1^*$ and $I_2^*$ as those that rank them
  % on positions $I_2^*$ and $I_1^*$, respectively ($yy'$ voters), yet
  % there are more voters who rank $C_1$ and $C_2$ on positions $I_1$
  % and $I_2$ than those that rank them on positions $I_2$ and $I_1$,
  % respectively ($xy' < x'y$). Thus, $C_2$ cannot be preferred over
  % $C_1$ in $S$. This gives a contradiction and proves that $\sup U
  % \leq \inf L$.

  Next, we show that $\sup U \geq \inf L$. To this end, we will show
  that there are no values $\frac{y}{x}$ and $\frac{y'}{x'}$ such that
  $\sup U < \frac{y}{x} < \frac{y'}{x'} < \inf L$. Assume on the contrary  that this is not the case and that such
  values exist. It must be the case that $\frac{y}{x}$ is not in $U$
  and, so, we have:
  \begin{equation}\label{l2:i}
     C_1 \:{\succeq}\Big[P_{x(C_1 \rightarrow I_2, C_2 \rightarrow
    I_1)}^{y(C_1 \rightarrow I_1^*, C_2 \rightarrow I_2^*)}\Big]\: C_2.
  \end{equation}
  Since $\frac{y}{x}$ also cannot be in $L$, we have:
  \begin{equation}\label{l2:ii}
    C_1 \:{\succeq}\Big[P_{x(C_1 \rightarrow I_1, C_2 \rightarrow
    I_2)}^{y(C_1 \rightarrow I_2^*, C_2 \rightarrow I_1^*)}\Big]\: C_2.
  \end{equation}

  By neutrality (applied to~\eqref{l2:ii}, and any permutation $\sigma$ that swaps candidates from
  $C_1 \setminus C_2$ with those from $C_2 \setminus C_1$), we have that:
  \begin{equation}\label{l2:iii}
    C_2 \:{\succeq}\Big[P_{x(C_1 \rightarrow
      I_2, C_2 \rightarrow I_1)}^{y(C_1 \rightarrow I_1^*, C_2
      \rightarrow I_2^*)}\Big]\: C_1.
  \end{equation}
  By putting together Equations~\eqref{l2:i} and~\eqref{l2:iii}, and
  by noting that the same reasoning can be repeated for
  $\frac{y'}{x'}$ instead of $\frac{y}{x}$, we conclude that it must
  be the case that:
  \begin{equation}\label{l2:added}
  C_1 \:{=}\Big[P_{x(C_1 \rightarrow I_2, C_2 \rightarrow I_1)}^{y(C_1
    \rightarrow I_1^*, C_2 \rightarrow I_2^*)}\Big]\: C_2
  \quad \text{and} \quad
  C_1 \:{=}\Big[P_{x'(C_1 \rightarrow I_2, C_2 \rightarrow
    I_1)}^{y'(C_1 \rightarrow I_1^*, C_2 \rightarrow I_2^*)}\Big]\:
  C_2. 
  \end{equation}
  After applying neutrality
  %that swaps all the members
 % of $C_1 \setminus C_2$ with all the members of $C_2\setminus C_1$ 
  to the first voting situation  in~\eqref{l2:added} (and copying the second
  part of~\eqref{l2:added})
  we obtain:
  \begin{equation}\label{l2:iv}
  C_1 \:{=}\Big[P_{x(C_1 \rightarrow I_1, C_2 \rightarrow I_2)}^{y(C_1
    \rightarrow I_2^*, C_2 \rightarrow I_1^*)}\Big]\: C_2
  \quad \text{and} \quad
  C_1 \:{=}\Big[P_{x'(C_1 \rightarrow I_2, C_2 \rightarrow
    I_1)}^{y'(C_1 \rightarrow I_1^*, C_2 \rightarrow I_2^*)}\Big]\:
  C_2.
  \end{equation}
  % Then, since $\frac{y}{x} \notin U$, we have $C_1
  % \:{\succeq}\Big[P_{x(C_1 \rightarrow I_2, C_2 \rightarrow
  %   I_1)}^{y(C_1 \rightarrow I_1^*, C_2 \rightarrow I_2^*)}\Big]\:
  % C_2$ and since $\frac{y}{x} \notin L$, we have $C_1
  % \:{\succeq}\Big[P_{x(C_1 \rightarrow I_1, C_2 \rightarrow
  %   I_2)}^{y(C_1 \rightarrow I_2^*, C_2 \rightarrow I_1^*)}\Big]\:
  % C_2$ and by neutrality $C_2 \:{\succeq}\Big[P_{x(C_1 \rightarrow
  %   I_2, C_2 \rightarrow I_1)}^{y(C_1 \rightarrow I_1^*, C_2
  %   \rightarrow I_2^*)}\Big]\: C_1$. We conclude that $C_1
  % \:{=}\Big[P_{x(C_1 \rightarrow I_2, C_2 \rightarrow I_1)}^{y(C_1
  %   \rightarrow I_1^*, C_2 \rightarrow I_2^*)}\Big]\: C_2$ and, by
  % similar reasoning, we get that $C_1 \:{=}\Big[P_{x'(C_1 \rightarrow
  %   I_2, C_2 \rightarrow I_1)}^{y'(C_1 \rightarrow I_1^*, C_2
  %   \rightarrow I_2^*)}\Big]\: C_2$. 
  We now define voting situation: %$Q$:
  \begin{align*}
    Q = 
      x' \cdot P_{x(C_1 \rightarrow I_1, C_2 \rightarrow I_2)}^{y(C_1
      \rightarrow I_2^*, C_2 \rightarrow I_1^*)} 
      + 
      x \cdot P_{x'(C_1
      \rightarrow I_2, C_2 \rightarrow I_1)}^{y'(C_1 \rightarrow
      I_1^*, C_2 \rightarrow I_2^*)} \textrm{.}
  \end{align*}
  From Equation~\eqref{l2:iv} (and consistency), we get that $C_1 =_Q
  C_2$. In $Q$ there is the same number of voters who rank $C_1$ and
  $C_2$ on positions $I_1$ and $I_2$ as those that rank them on
  positions $I_2$ and $I_1$, respectively (so these voters cancel each
  other out). On the other hand, there are $yx'$ voters who rank $C_1$
  and $C_2$ on positions $I_2^*$ and $I_1^*$, and $y'x$ voters who
  rank these committees on positions $I_1^*$ and $I_2^*$,
  respectively. Since $yx' < y'x$, we get that $C_1 \succ_Q C_2$,
  which contradicts our earlier observation
  that $C_1 =_Q C_2$. We conclude that it must be the case that $\sup
  U \geq \inf L$.

  Finally, since we have shown that $\sup U \leq \inf L$ and $\sup U
  \geq \inf L$, we have that $\sup U = \inf L$. This proves that
  $\Delta_{I_2,I_1} = -\Delta_{I_1,I_2}$.
\end{proof}

The next lemma shows that $\Delta_{I_1,I_2}$ provides a threshold
value for proportions of voters in distinguished profiles with respect
to the relation between $C_1$ and $C_2$.

\begin{lemma}\label{lemma:deltaDefinition}
%For each two $k$-element subsets, $I_1$ and $I_2$, and 
  Let $I_1$ and $I_2$ be two committee positions such that $|I_1 \cap
  I_2| = s$, and let $x,y$ be two positive
  integers.  The following two implications hold:
  \begin{enumerate}
  \item if $C_1 \; {\succ}\big[\spcvote{C_1}{I_1}{C_2}{I_2}\big] \;
    C_2$ and $\displaystyle \frac{y}{x} < \Delta_{I_1, I_2}$, then $C_2
    \;{\succ}\Big[P_{x(C_1 \rightarrow I_2, C_2 \rightarrow
      I_1)}^{y(C_1 \rightarrow I_1^*, C_2 \rightarrow I_2^*)}\Big]\;
    C_1$,
  \item if $C_2 \; {\succ}\big[ \spcvote{C_1}{I_1}{C_2}{I_2} \big] \;
    C_1$ and $\displaystyle \frac{y}{x} > -\Delta_{I_1, I_2}$, then $C_2
    \;{\succ}\Big[P_{x(C_1 \rightarrow I_2, C_2 \rightarrow
      I_1)}^{y(C_1 \rightarrow I_2^*, C_2 \rightarrow I_1^*)}\Big]\;
    C_1$.
  \end{enumerate}
\end{lemma}
\begin{proof}
  Let us start with proving the first implication.  Assume that $C_1
  \; {\succ}\big[\spcvote{C_1}{I_1}{C_2}{I_2}\big] \; C_2$, and, for
  the sake of contradiction, that:
  \begin{equation}\label{l3:i}
  C_1\: \:{\succeq}\left[P_{x(C_1 \rightarrow I_2, C_2 \rightarrow
    I_1)}^{y(C_1 \rightarrow I_1^*, C_2 \rightarrow I_2^*)}\right]\: \:  C_2.
  \end{equation}
  It follows from the definition of $\Delta_{I_1, I_2}$ that there exist
  two numbers $x', y' \in \naturals$, such that $\frac{y}{x} <
  \frac{y'}{x'} \leq \Delta_{I_1,I_2}$ and:
  \begin{equation}\label{l3:ii}
    C_2 \:{\succ}\left[P_{x'(C_1 \rightarrow
      I_2, C_2 \rightarrow I_1)}^{y'(C_1 \rightarrow I_1^*, C_2
      \rightarrow I_2^*)}\right]\: C_1. 
  \end{equation}
  Let us consider a voting situation that is obtained from $P_{x(C_1
    \rightarrow I_2, C_2 \rightarrow I_1)}^{y(C_1 \rightarrow I_1^*,
    C_2 \rightarrow I_2^*)}$ (i.e., from the voting situation that
  appears in~\eqref{l3:i}) by swapping positions of $C_1$ and $C_2$,
  i.e., let us consider voting situation $P_{x(C_1 \rightarrow I_1,
    C_2 \rightarrow I_2)}^{y(C_1 \rightarrow I_2^*, C_2 \rightarrow
    I_1^*)}$. Naturally, in such a voting situation $C_2$ is weakly
  preferred over $C_1$:
  \begin{equation}\label{l3:iii}
    C_2 \:{\succeq}\left[  P_{x(C_1 \rightarrow I_1,
      C_2 \rightarrow I_2)}^{y(C_1 \rightarrow I_2^*, C_2 \rightarrow
      I_1^*)}\right]\: C_1. 
  \end{equation}
  By Equations~\eqref{l3:iii}, \eqref{l3:ii}, and consistency of
  $f_k$, we observe that in the voting situation:
  \[ P = x' \cdot P_{x(C_1 \rightarrow I_1, C_2
    \rightarrow I_2)}^{y(C_1 \rightarrow I_2^*, C_2 \rightarrow
    I_1^*)} + x \cdot P_{x'(C_1 \rightarrow I_2, C_2 \rightarrow
    I_1)}^{y'(C_1 \rightarrow I_1^*, C_2 \rightarrow I_2^*)}
  \]
  committee $C_2$ is strictly preferred over $C_1$ (i.e., $C_2 \succ_P
  C_1$). Let us now count the voters in $P$.  There are $xx'$ of them who put
  $C_1$ and $C_2$ on positions $I_1$ and $I_2$, respectively, and
  there are $xx'$ voters who put $C_1$ and $C_2$ on positions $I_2$
  and $I_1$, respectively. By the same arguments as used in the proof
  of Lemma~\ref{lemma:reverseDelta}, these voters cancel each other
  out. Next, there are $y'x$ voters who put $C_1$ and $C_2$ on
  positions $I_1^*$ and $I_2^*$, respectively, and  $x'y$
  voters who put $C_1$ and $C_2$ on positions $I_2^*$ and $I_1^*$,
  respectively. Since $y'x > yx'$. we conclude that $C_1 \succ_P C_2$
  (again, using the same reasoning as we used in
  Lemma~\ref{lemma:reverseDelta} for similar arguments).  This is a
  contradiction with our earlier observation that $C_2 \succ_P C_1$.
  This completes the proof of the first part of the lemma.

  % In $P$ there are at least $xx'$ voters who put $C_1$ and $C_2$ on
  % positions $I_1$ and $I_2$, respectively (there are exactly $x x'$
  % such voters if $\{I_1, I_2\} \neq \{I_1^{*}, I_2^{*}\}$); further,
  % there are at least $x x'$ voters who put $C_1$ and $C_2$ on
  % positions $I_2$ and $I_1$, respectively (as before, there are
  % exactly $x x'$ such voters if $\{I_1, I_2\} \neq \{I_1^{*},
  % I_2^{*}\}$). Let $P_r$ be the voting situation that consists of
  % those $2xx'$ votes from $P$, where $C_1$ and $C_2$ stand on
  % positions $I_2$ and $I_1$ ($xx'$ votes) and on positions $I_1$ and
  % $I_2$ (the other $xx'$ votes), respectively. From
  % the neutrality it follows that $C_1 \:{=}[P_r]\:
  % C_2$. Thus, from consistency it follows that we can remove these
  % $2xx'$ votes from $P$ and obtain a new voting situation $Q = P -
  % P_r$ in which $C_2$ is still preferred over $C_1$. Furthermore, in
  % $Q$, $y'x$ voters put $C_1$ and $C_2$ on positions $I_1^{*}$ and
  % $I_2^{*}$, respectively; and $x'y$ voters put $C_1$ and $C_2$ on
  % positions $I_2^{*}$ and $I_1^{*}$. Since $y'x > x'y$, we infer
  % that
  % $C_1 \:{\succ}[Q]\: C_2$, which is a contradiction.

  The proof of the second implication is similar and we provide it for
  the sake of completeness. We assume that $C_2 \; {\succ}\big[
  \spcvote{C_1}{I_1}{C_2}{I_2} \big] \; C_1$ and, for the sake of
  contradiction, that:
  \begin{equation}\label{l3:iv}
    C_1
    \;{\succeq}\left[P_{x(C_1 \rightarrow I_2, C_2 \rightarrow
        I_1)}^{y(C_1 \rightarrow I_2^*, C_2 \rightarrow I_1^*)}\right]\;
    C_2. 
  \end{equation}
  From the definition of $\Delta_{I_1, I_2}$ we know that there must
  be two numbers $x', y' \in \naturals$, such that $\frac{y}{x} >
  \frac{y'}{x'} \geq -\Delta_{I_1,I_2}$ and:
  \begin{equation}\label{l3:v}
    C_2 \:{\succ}\big[P_{x'(C_1 \rightarrow I_2, C_2
      \rightarrow I_1)}^{y'(C_1 \rightarrow I_2^*, C_2 \rightarrow
      I_1^*)}\big]\: C_1.
  \end{equation}
  If we swap the positions of committees $C_1$ and $C_2$ in the voting
  situation used in Equation~\eqref{l3:iv}, then by neutrality we have that:
  \begin{equation}\label{l3:vi}
    C_2 \;{\succeq}\left[P_{x(C_1 \rightarrow I_1, C_2 \rightarrow I_2)}^{y(C_1 \rightarrow I_1^*,
        C_2 \rightarrow I_2^*)}\right]\; C_1
  \end{equation}
  We now form voting situation:
  \[
  Q = x' \cdot P_{x(C_1 \rightarrow I_1, C_2 \rightarrow I_2)}^{y(C_1
    \rightarrow I_1^*, C_2 \rightarrow I_2^*)} + x \cdot P_{x'(C_1
    \rightarrow I_2, C_2 \rightarrow I_1)}^{y'(C_1 \rightarrow I_2^*,
    C_2 \rightarrow I_1^*)}
  \]
  By Equations~\eqref{l3:vi}, \eqref{l3:v}, and consistency of $f_k$
  we have that $C_2 \succ_Q C_1$. However, counting voters again leads
  to a contradiction. Indeed, we have $xx'$ voters who put $C_1$ and
  $C_2$ on positions $I_1$ and
  $I_2$, respectively, and $xx'$ voters who put $C_1$ and $C_2$ on
  positions $I_2$ and $I_1$, respectively. These voters cancel each
  other out. Then we have $y'x$ voters who put $C_1$ and $C_2$ on
  positions $I_2^*$ and $I_1^*$, respectively, and we have $x'y$
  voters who put $C_1$ and $C_2$ on positions $I_1^*$ and $I_2^*$,
  respectively.  Since $y'x < x'y$, we have that $C_1 \succ_Q C_2$,
  which is a contradiction with our previous conclusion that $C_2
  \succ_Q C_1$. This proves the second part of the lemma.
\end{proof}

\paragraph{Putting Together the Decision Scoring Function for
  \boldmath{$f_{C_1,C_2}$}.}
We are ready to define a decision scoring function $d_s$ for
$f_{C_1,C_2}$. For any two committee positions $I_1$ and $I_2$, with
$|I_1\cap I_2|=s$, we set:
\[
d_s(I_1, I_2) = \Delta_{I_1, I_2} \textrm{.}
\]
We note that our $d_s$
formally depends on the choice of $I_1^*$ and $I_2^*$, however this is
not a problem. We simple need a decision scoring function that behaves
correctly and each choice of $I_1^*$ and $I_2^*$ would give us one.
Intuitively, we can think of $d_s(I_1,I_2)$ as an (oriented) distance
between $I_1$ and $I_2$. The next lemma shows that we treat the
distance between $I_1^*$ and $I_2^*$ as a sort of gauge to measure
distances between other positions.

\begin{lemma}
  \label{unit_distance}
  It holds that $\Delta_{I_1^*,  I_2^*} = 1$.
\end{lemma}

% and for how this decision scoring rule operates.  As a quick sanity
% check, the reader can verify that $\Delta_{I_1^*, I_2^*} = 1$.
  
\begin{proof}
  We note that for each positive integer $z$, we have $C_1 \: {=}\big[
  P_{z(C_1 \rightarrow I_2^*, C_2 \rightarrow I_1^*)}^{z(C_1
    \rightarrow I_1^*, C_2 \rightarrow I_2^*)} \big] \: C_2$. Further,
  due to consistency of $f_k$ (used as in
  Remark~\ref{rem:using-consistency}) and
  by the choice of $I_1^*$ and $I_2^*$ (recall Equations
  \eqref{main_assumption} and \eqref{main_assumption_consequence}), we
  observe that $C_1 \: {\succ}\big[ P_{x(C_1 \rightarrow I_2^*, C_2
    \rightarrow I_1^*)}^{y(C_1 \rightarrow I_1^*, C_2 \rightarrow
    I_2^*)} \big] \: C_2$ whenever $y > x$ and $C_2 \: {\succ}\big[
  P_{x(C_1 \rightarrow I_2^*, C_2 \rightarrow I_1^*)}^{y(C_1
    \rightarrow I_1^*, C_2 \rightarrow I_2^*)} \big] \: C_1$ whenever
  $y < x$. We conclude that $\Delta_{I_1^*,I_2^*} = \sup\left\{
    \frac{y}{x} \colon y < x, \text{ for }x,y \in \naturals_+\right\}
  = 1.$
\end{proof}

The next lemma shows that $d_s$ is a decision scoring function for
$f_{C_1,C_2}$. Based on this result, we will later argue that it works
for all pairs of committees, not only for $(C_1,C_2)$, and hence that
it is a decision scoring function for $f_{k,s}$.

\begin{lemma}\label{thm:nontransitive}
  Let $C_1,C_2$ and $d_s$ be as defined in the
  above discussion.  Then for each voting situation $P$ the following
  three implications hold:
%  \begin{enumerate}
%  \item[(i)] 
  (i) if $\pairscorefull{d_s}{C_1}{C_2}{P} > 0$, then $C_1 \succ_P C_2$;
%  \item[(ii)] 
  (ii) if $\pairscorefull{d_s}{C_1}{C_2}{P} = 0$, then  $C_1 =_P C_2$;
%  \end{enumerate}
  (iii) if $\pairscorefull{d_s}{C_1}{C_2}{P} < 0$, then $C_2 \succ_P C_1$.
\end{lemma}
\begin{proof}
  We start by proving (i). Let $P$ be a voting situation such
  that $\pairscorefull{d_s}{C_1}{C_2}{P} > 0$.  For the sake of contradiction we assume
  that $C_2 \succeq_{P} C_1$.

  The idea of the proof is to perform a sequence of transformations of
  $P$ so that the result according to $f_k$ does not change
  (due to the imposed axioms), but,
  eventually, in the resulting profile each voter puts committees
  $C_1$ and $C_2$ either on positions $I_1^*$, $I_2^*$ or the other
  way round.  Let $t$ be the total number of transformations that
  we perform to achieve
  this and let $P_i$ be the voting situation that we obtain after the
  $i$-th transformation.  We will ensure that for each voting
  situation $P_i$ it holds that $\pairscorefull{d_s}{C_1}{C_2}{P_i} >
  0$ and $C_2 \succeq_{P_i} C_1$. In particular, for the final voting
  situation $P_t$ we will have $C_2 \succeq_{P_t} C_1$,
  $\pairscorefull{d_s}{C_1}{C_2}{P_t} > 0$, and each voter
  will have committees $C_1$ and $C_2$ on
  positions $I_1^*$ and $I_2^*$ or the other way round.  Therefore,
  we will have:
  \[
  P_t=x(C_1 \rightarrow I_1^*, C_2 \rightarrow I_2^*) + y(C_1 \rightarrow I_2^*, C_2 \rightarrow I_1^*)
  \]
  for some nonnegative integers $x$ and $y$, and by
  Lemmas~\ref{lemma:reverseDelta}~and~\ref{unit_distance} we will have:
  \[
  \pairscorefull{d_s}{C_1}{C_2}{P_t} = xd_s(I_1^*,I_2^*) + yd_s(I_2^*,I_1^*)=x-y
  \]
  However, from $\pairscorefull{d_s}{C_1}{C_2}{P_t} > 0$ we will
  conclude that $x>y$, i.e., there must be more voters who put $C_1$
  and $C_2$ on positions $I_1^*$ and $I_2^*$ than on positions $I_2^*$
  and $I_1^*$. By our choice of $I_1^*$ and
  $I_2^*$ (recall Equation \eqref{main_assumption} as in the proof of
  Lemma~\ref{unit_distance}) we will conclude that $C_1 \succ_{P_t}
  C_2$.  This will be a contradiction with $C_2 \succeq_{P_t} C_1$.
\par\smallskip

  We now describe the transformations.  We set $P_0 = P$. We perform
  the $i$-th transformation in the following way. If for each voter in
  $P_{i-1}$, committees $C_1$ and $C_2$ stand on positions $I_1^*$ and
  $I_2^*$ (or the other way round), we finish our sequence of
  transformations. Otherwise, we take a preference order of an
  arbitrary voter from $P_{i-1}$, for whom the set of committee
  positions of $C_1$ and $C_2$ is not $\{I_1^*, I_2^*\}$.  Let us
  denote this voter by $v_i$. Let $z$ denote the number of voters in
  $P_{i-1}$ who rank $C_1$ and $C_2$ on the same positions as $v_i$,
  including $v_i$ (so $z \geq 1$). Let $I_1$ and $I_2$ denote the
  positions of the committees $C_1$ and $C_2$ in the preference order
  of $v_i$, respectively. Let $\epsilon = \pairscorefull{d_s}{C_1}{C_2}{P_{i-1}}/2z >
  0$.

  \begin{description}
  \item{\bf Case 1:} If $C_1 \, {=}\big[ \spcvote{C_1}{I_1}{C_2}{I_2}
    \big] \, C_2$, then we obtain $P_i$ by removing from $P_{i-1}$ all
    $z$ voters with the same preference order as $v_i$. By consistency
    of $f_k$, it follows that in the resulting voting situation $P_i$
    it still holds that $C_2 \succeq_{P_i} C_1$ (this is, in essence,
    the same canceling out of voters that we already used in
    Lemmas~\ref{lemma:reverseDelta}
    and~\ref{lemma:deltaDefinition}). Also, by definition of $\Delta_{I_1,I_2}$
    in Equation~\eqref{eq:diffDefinition}, we have $\Delta_{I_1,I_2} =
    0$. Hence, it still holds that $\pairscorefull{d_s}{C_1}{C_2}{P_i}
    > 0$.

  \item{\bf Case 2:} If $C_1 \;
    {\succ}\big[\spcvote{C_1}{I_1}{C_2}{I_2} \big] \; C_2$, then let
    $x$ and $y$ be such integers that $\Delta_{I_1, I_2} - \epsilon <
    \frac{y}{x} < \Delta_{I_1, I_2}$ (recall that $\epsilon$ is
    defined just above Case 1, and that $z$ is the number of voters
    with the same preference order as $v_i$).  We define two new
    voting situations:
    \[
    R_{i-1} = z \cdot P_{x(C_1 \rightarrow I_2, C_2 \rightarrow
      I_1)}^{y(C_1 \rightarrow I_1^*, C_2 \rightarrow I_2^*)} \quad
    \text{and} \quad Q_{i-1} = x \cdot P_{i-1} + R_{i-1} \textrm{.}
    \]
    From Lemma~\ref{lemma:deltaDefinition} it follows that $C_2
    \succ_{R_{i-1}} C_1$ and, by consistency, we get that $C_2
    \succ_{Q_{i-1}} C_1$.
%
% We define a new voting situation $Q_{i-1}$:
% \begin{align*}
%    Q_{i-1} = x \cdot P_{i-1} + z \cdot P_{x(C_1 \rightarrow I_2, C_2 \rightarrow I_1)}^{y(C_1 \rightarrow
%     I_1^*, C_2 \rightarrow I_2^*)} \textrm{.} 
% \end{align*}
% From
%   Lemma~\ref{lemma:deltaDefinition} it follows that:
%   \[ 
%   C_2 \;{\succ}\Big[P_{x(C_1 \rightarrow I_2, C_2 \rightarrow
%     I_1)}^{y(C_1 \rightarrow I_1^*, C_2 \rightarrow I_2^*)}\Big]\;
%   C_1,
%   \]
%   and, by consistency, we get that $C_2
%   \succeq_{Q_{i-1}} C_1$. 
%
    Let us now calculate % the value of
    $\pairscorefull{d}{C_1}{C_2}{R_{i-1}}$. We note that $R_{i-1}$
    consists of $zx$ voters who rank $C_1$ and $C_2$ on positions
    $I_2$ and $I_1$ (and who contribute $zx \Delta_{I_2, I_1} = -zx
    \Delta_{I_1, I_2}$ to the value of
    $\pairscorefull{d}{C_1}{C_2}{R_{i-1}}$) and of $zy$ voters who
    rank $C_1$ and $C_2$ on positions $I_1^*$ and $I_2^*$,
    respectively (who contribute value $zy \Delta_{I_1^*, I_2^*} =
    zy$). That is, we have $\pairscorefull{d}{C_1}{C_2}{R_{i-1}} = -zx
    \Delta_{I_1, I_2} + zy$.  Further, by definition of $\epsilon$, we
    have that $\pairscorefull{d}{C_1}{C_2}{P_{i-1}} = 2z\epsilon$.
    In consequence, we have that:
    \begin{align*}
      \pairscorefull{d}{C_1}{C_2}{Q_{i-1}} &= x \cdot \pairscorefull{d}{C_1}{C_2}{P_{i-1}} + \pairscorefull{d}{C_1}{C_2}{R_{i-1}}\\
      &= 2zx\epsilon + ( - zx \Delta_{I_1, I_2} + zy) \\
      &= 2zx\epsilon + zx(-\Delta_{I_1, I_2} + \frac{y}{x}) \geq
      2zx\epsilon - zx\epsilon > 0 \textrm{.}
    \end{align*}
    The first inequality (in the final row) follows from the fact that
    we assumed $\Delta_{I_1, I_2} - \epsilon < \frac{y}{x} <
    \Delta_{I_1, I_2}$.  We now move on to Case 3, where we also build
    voting situation $Q_{i-1}$ with a similar property, and then
    describe how to obtain $P_i$ from $Q_{i-1}$'s.

  \item{\bf Case 3:} If $C_2 \; {\succ}\big[
    \spcvote{C_1}{I_1}{C_2}{I_2} \big] \; C_1$, then our reasoning is
    very similar to that from Case~$2$. Let $x$ and $y$ be such
    integers that $-\Delta_{I_1, I_2} < \frac{y}{x} < -\Delta_{I_1,
      I_2} + \epsilon$. We define two voting situations
    \begin{align*}
      R_{i-1} = z \cdot P_{x(C_1 \rightarrow I_2, C_2 \rightarrow
        I_1)}^{y(C_1 \rightarrow I_2^*, C_2 \rightarrow I_1^*)}
      \quad \text{and} \quad
      Q_{i-1} = x \cdot P_{i-1} + R_{i-1}.
    \end{align*}
     Lemma~\ref{lemma:deltaDefinition} implies that $C_2
    \succ_{R_{i-1}} C_1$, and, thus, from consistency, we get that
    $C_2 \succ_{Q_{i-1}} C_1$. Further, using similar analysis as in
    Case~2, we get that:
    \begin{align*}
      \pairscorefull{d}{C_1}{C_2}{Q_{i-1}} &= x \cdot \pairscorefull{d}{C_1}{C_2}{P_{i-1}} + zx \Delta_{I_2, I_1} - zy \\
      &= 2zx\epsilon + zx(-\Delta_{I_1, I_2} - \frac{y}{x}) \geq
      2zx\epsilon - zx\epsilon > 0 \textrm{.}
    \end{align*}
    The first inequality (in the final row) follows from the
    assumption that $-\Delta_{I_1, I_2} < \frac{y}{x} < -\Delta_{I_1,
      I_2} + \epsilon$. Below we describe how to obtain $P_i$ from
    $Q_{i-1}$ (for both Cases~2 and 3).
  \end{description}
  \medskip

  In Cases~2 and~3, in the voting situation $Q_{i-1}$ exactly $zx$
  voters have $C_1$ and $C_2$ on positions $I_2$ and $I_1$,
  respectively (for both cases, these voters are introduced in voting
  situation $R_{i-1}$). Further, there are exactly $zx$
  voters who rank $C_1$ and $C_2$ on positions $I_1$ and $I_2$,
  respectively (these are the cloned-$x$-times voters that were
  originally in $P_{i-1}$).  We define $P_{i}$ as $Q_{i-1}$ with these
  $2zx$ voters removed. Since we removed the same number of voters who
  rank $C_1$ and $C_2$ on positions $I_2$ and $I_1$, respectively, as
  the number of voters who rank these committees on positions $I_1$
  and $I_2$, respectively, we conclude that $\pairscorefull{d}{C_1}{C_2}{P_i} =
  \pairscorefull{d}{C_1}{C_2}{Q_{i-1}} > 0$ and that $C_2 \succeq_{P_i} C_1$.

  We note that after the just-described transformation none of the
  voters has both $C_1$ and $C_2$ on positions $I_1$ and $I_2$,
  respectively, and that we only added a number of voters that rank
  $C_1$ and $C_2$ on positions $I_1^*$ and $I_2^*$ (or the other way
  round) or we
  cloned voters already present. Hence, if we perform such
  transformations for all possible pairs of committee positions $I_1$
  and $I_2$, we will obtain our final voting situation, $P_t$, for which it holds that
  the following three conditions are satisfied: (i)
  $\pairscorefull{d}{C_1}{C_2}{P_t} > 0$, (ii) $C_2 \succeq_{P_t}
  C_1$, and (iii) in $P_t$ each voter ranks $C_1$ and $C_2$ on
  positions $I_1^*$ and $I_2^*$ (or the other way round).  Given (i)
  and (iii) we conclude that in $P_t$ there are more votes in which
  $C_1$ stands on position $I_1^*$ and $C_2$ stands on position
  $I_2^*$ than there are voters where the opposite holds. However,
  this implies that $C_1 \succ_{P_t} C_2$ and contradicts the fact
  that $C_2 \succeq_{P_t} C_1$.  This completes the proof of the first
  part of the lemma.

  Next, we consider part (ii) of the theorem.  Let $P$ be some voting
  situation such that $\pairscorefull{d}{C_1}{C_2}{P} = 0$. For the
  sake of contradiction we assume that $C_2 \neq_{P} C_1$, and,
  without loss of generality, we assume that
  $C_2 \succ_{P} C_1$. Since for the voting situation
  $\spcvote{C_1}{I_1^*}{C_2}{I_2^*}$ it holds that
  $\pairscorefull{d}{C_1}{C_2}{\spcvote{C_1}{I_1^*}{C_2}{I_2^*}} > 0$,
  then for each $n \in \naturals$, in the voting situation
  \[
  Q_n=nP + \spcvote{C_1}{I_1^*}{C_2}{I_2^*},
  \]
  we have $\pairscorefull{d}{C_1}{C_2}{Q_n} > 0$, and---from part (i) of the
  theorem---we get that $C_1 \succ_{Q_n} C_2$. On the other hand,
  continuity requires that there exists some value of $n$ such that
  $C_2 \succ_{Q_n} C_1$. 

  To prove part (iii) of the theorem, it suffices to observe that if
  $\pairscorefull{d_s}{C_1}{C_2}{P} < 0$, then
  $\pairscorefull{d_s}{C_2}{C_1}{P} =
  -\pairscorefull{d_s}{C_1}{C_2}{P} > 0$ and use part (i) of the
  theorem to conclude that in such case we have $C_2 \succ_P C_1$.
  This gives a contradiction and completes the proof.
\end{proof}

  \paragraph{Completing the Proof of Theorem~\ref{thm:decision-rules}.}
  We have dealt with a fixed pair of committees $(C_1,C_2)$ and we
  have proven Lemma~\ref{thm:nontransitive} which justifies that $d_s$
  is a decision scoring function for $f_{C_1, C_2}$. From neutrality
  it follows that $d_s$ will give us a decision scoring function for
  $f_{k, s}$. However, as we noted at the beginning of this section,
  $f_k$ can be viewed as a collection of independent functions $f_{k,
    s}$ for $s \in \{0 \ldots k-1\}$, thus this observation is
  sufficient to prove Theorem~\ref{thm:decision-rules}, a Young-Style
  characterization of decision scoring rules.
%\end{proof}

\subsection{The Tools to Deal with Committee Scoring Rules}
\label{sec:toolsforcsrs}

%\paragraph{Getting Ready for the Second Domain Change}
We have proved Theorem~\ref{thm:decision-rules}, which will serve as a
useful tool for proving Theorem~\ref{thm:theMainTheorem}.  However, to
complete the proof of Theorem~\ref{thm:theMainTheorem} we still need
to derive one more technical tool---Lemma~\ref{thm:nontransitive2}
below---that applies the results obtained so far to committee scoring
rules. To achieve this goal, we need to change our domain from
$\naturals^{m!}$ to $\rationals^{m!}$, and before we make this change,
we need to introduce several new notions.  (While the correctness of
our first domain change relied on the decision rule being symmetric,
this second domain change, similarly to the case considered by
Young~\shortcite{you:j:scoring-functions}, uses our further axioms.)

We distinguish one specific voting situation, $e = \langle 1, 1,
\ldots, 1 \rangle$, called the \emph{null profile}, describing the
setting where each possible vote is cast exactly once. It immediately
follows that under each symmetric $k$-decision rule $f_k$, each two
committees are ranked equally in $e$, i.e., for each two committees
$C_1, C_2$ we have $C_1 =_e C_2$.

\begin{definition}[Independence of Symmetric Profiles]
  A symmetric $k$-decision rule $f_k$ is \emph{independent of
    symmetric profiles} if for every voting situation $P \in
  \naturals^{m!}$ and for every $\ell \in \naturals$, we have that
  $f_k(P + \ell e) = f_k(P)$.
\end{definition}

\begin{definition}[Homogeneity]
  A symmetric $k$-decision rule $f_k$ is \emph{homogeneous} if for
  every voting situation $P \in \naturals^{m!}$ and for every $\ell
  \in \naturals$, we have $f_k(\ell P) = f_k(P)$.
\end{definition}

Intuitively, independence of symmetric profiles says that if we add
one copy of each possible vote then they will all cancel each other
out. Homogeneity says that the result of an election depends only on
the relative proportions of the linear orders in the voting situation
and not on the exact numbers of linear orders.  One can verify that
each symmetric and consistent $k$-decision rule satisfies both
independence of symmetric profiles and homogeneity (indeed, the
requirement in the definition of homogeneity is a special case of the
requirement from the definition of consistency).\footnote{The reader
  may ask why do we introduce independence of symmetric profiles and
  homogeneity, when what we require from them already follows from
  consistency. The reason is that, we believe, these two properties
  better explain---on the intuitive level---why the second domain
  change is allowed.}

\paragraph{Second Domain Change.}
Now we are ready to extend our
domain from $\naturals^{m!}$ to $\rationals^{m!}$. To this end, we use
the following result. It was originally stated for single-winner rules
but it can be adapted to the multiwinner setting in a straightforward
way.

\begin{lemma}[Young~\shortcite{you:j:scoring-functions}, Merlin~\shortcite{merlinAxiomatic}]
\label{thm:uniqe-ext}
  % Let $f_k$ be a symmetric, homogeneous, and independent of
  % symmetric profiles $k$-winner rule defined on
  % $\naturals^{m!}$. Then
Suppose a $k$-decision rule $f_k \colon \naturals^{m!} \to \big(\powA
\times \powA \to \{-1, 0, 1\}\big) $ is symmetric, independent of
symmetric profiles and homogeneous. There exists a unique extension of
$f_k$ to the domain $\rationals^{m!}$ (which we also denote by $f_k$),
satisfying for each positive $\ell \in \naturals$, and $P \in \naturals^{m!}$
the following two conditions:
  \begin{enumerate}
  \item $f_k(P - \ell e) = f_k(P)$,
  \item  $f_k\left(\frac{P}{\ell}\right) = f_k(P)$.
  \end{enumerate}
\end{lemma}

%FIN

Lemma~\ref{thm:uniqe-ext} allows us to consider voting situations with
fractional numbers of linear orders. From now on, when we speak of
voting situations, we mean voting situations from our new domain,
$\rationals^{m!}$.  We note that within our new domain, the score of a
pair $(C_1, C_2)$ of committees relative to a voting situation $P$
under decision scoring function $d$ can still be expressed as in
Equation~\eqref{def_D_P_vot_sit}.
%\[
%  \pairscorefull{d}{C_1}{C_2}{P} = \sum_{\pi \in \orders} P(\pi) \cdot d(\pos_{\pi}(C_1), \pos_{\pi}(C_2)).
%\]
Indeed, for decision scoring rules, this definition gives the unique
extension that Lemma~\ref{thm:uniqe-ext} speaks of. Thus
Theorem~\ref{thm:decision-rules} extends to decision rules with domain
$\rationals^{m!}$.

% \begin{proof}
%   See~\shortcite{merlinAxiomatic}, for example.
% \end{proof}

\paragraph{Constructing a Tool for Committee Scoring Rules.}
%\paragraph{Completing the First Part of the Proof.}
%Throughout most of the discussion so far, $C_1$ and $C_2$ were two
%fixed committees. However, since their choice was arbitrary, all the
%results we have established so far hold for all pairs of committees.
%From now on we no longer assume $C_1$ and $C_2$ to be these two fixed
%committees.

Since $\rationals^{m!}$ is a vector space over the field of rational
numbers, from Theorem~\ref{thm:decision-rules} (extended to
$\rationals^{m!}$) we infer that for each two committees $C_1$ and
$C_2$, the space of voting situations $P$ such that $C_1 =_P C_2$ is a
hyperplane in the $m!$-dimensional vector space of all voting
situations. This is so, because if we treat a voting situation $P$ as
a vector of $m!$ variables, then condition
$\pairscorefull{d}{C_1}{C_2}{P} = 0$ turns out to be a single linear
equation.  Hence, the space of voting situations $P$ such that $C_1
=_P C_2$ is a hyperplane in $\rationals^{m!}$ and has dimension
$m!-1$. This can be summarized as the following
corollary.

\begin{corollary}\label{cor:hyperplane}
  The set $\{P \in \rationals^{m!} \colon C_1 =_P C_2\}$ is a
  hyperplane in the vector space of all voting situations~$\rationals^{m!}$.
\end{corollary}

From now on, we assume that our $k$-decision rule $f_k$ is transitive,
that is, we require that for each voting situation $P$ and each three
committees $C_1$, $C_2$, and $C_3$ it holds that:
\[
(C_1\succeq_P C_2)\ \text{and}\ (C_2\succeq_P C_3)\
\text{implies}\ (C_1\succeq_P C_3).
\]
In other words, from now on we require $f_k$ to be a $k$-winner
election rule.

\begin{lemma}\label{thm:nontransitive2}
  Let $f_{k}$ be a symmetric, consistent,
  committee-dominant, continuous $k$-winner election rule, and let
  $\lambda\colon [m]_k \to \reals$ be a committee scoring function.
  If it holds that for each two committees $C_1$ and $C_2$ and each
  voting situation $P$ it holds that the committee scores of $C_1$ and
  $C_2$ are equal (according to $\lambda$) if and only $C_1$ and $C_2$
  are equivalent according to ${f_k}$, then it holds that: For each
  two committees $C_1$ and $C_2$ and each voting situation $P$, if the
  committee score of $C_1$ is greater than that of $C_2$ (according to
  $\lambda$) then $C_1$ is preferred over $C_2$ according to ${f_k}$
  (i.e., $C_1 \succ_P C_2$).
\end{lemma}

\begin{proof}
%  We assume that $m > k$ (if $m=k$ then the lemma holds trivially).
  Based on $\lambda$, we build a decision scoring function $g$
  as follows. For each two committee positions $I_1$ and $I_2$, we
  have $g(I_1, I_2) = \lambda(I_1) - \lambda(I_2)$.  The score of a
  committee pair $(C_1, C_2)$ in voting situation $P$ under $g$ is
  given by:
  \begin{align*}
    \pairscorefull{g}{C_1}{C_2}{P} = \sum_{\pi \in \orders} P(\pi) \cdot g(\pos_{\pi}(C_1), \pos_{\pi}(C_2)) \textrm{.}
  \end{align*}

  Let us fix $x \in [k-1]$ and two arbitrary committees $C_1^*$ and
  $C_2^*$ such that $|C_1^* \cap C_2^*| = x$.  We note that, by the
  assumptions of the theorem, if it holds that:
  \[
   \pairscorefull{g}{C_1^*}{C_2^*}{P} = 0 \iff C_1^* =_P C_2^*,
  \] 
  then, by Corollary~\ref{cor:hyperplane}, $H = \{P \in
  \rationals^{m!}  \colon C_1^* =_P C_2^*)\}$ is an $(m! -
  1)$-dimensional hyperplane.  More so, this is the same hyperplane as
  the following two (where $d=d_x$ is the decision scoring function
  from the thesis of Lemma~\ref{thm:nontransitive}, built for $f_k$):
  \[
   \{P \in \rationals^{m!} \colon  \pairscorefull{g}{C_1^*}{C_2^*}{P} = 0\} 
   \quad \text{and} \quad 
   \{P \in \rationals^{m!} \colon  \pairscorefull{d}{C_1^*}{C_2^*}{P} = 0\},
  \]  
  We claim that for $C_1^*$ and $C_2^*$ one of the following
  conditions must hold:
  \begin{enumerate}
  \item For each voting situation $P$, if
    $\pairscorefull{g}{C_1^*}{C_2^*}{P} > 0$ then $C_1^* \succ_P
    C_2^*$.
  \item For each voting situation $P$, if
    $\pairscorefull{g}{C_1^*}{C_2^*}{P} > 0$ then $C_2^* \succ_P
    C_1^*$.
  \end{enumerate}
  Why is this so? For the sake of contradiction, let us assume that
  there exist two voting situations, $P$ and $Q$, such that
  $\pairscorefull{g}{C_1^*}{C_2^*}{P} > 0$ and
  $\pairscorefull{g}{C_1^*}{C_2^*}{Q} > 0$, but $C_1^* \succeq_{P}
  C_2^*$ and $C_2^* \succeq_{Q} C_1^*$. From the fact that
  $\pairscorefull{g}{C_1^*}{C_2^*}{P} > 0$ and
  $\pairscorefull{g}{C_1^*}{C_2^*}{Q} > 0$, we see that the points $P$
  and $Q$ lie on the same side of hyperplane $H$ and neither of them
  lies on $H$. From $C_1^* \succeq_{P} C_2^*$, $C_2^* \succeq_{Q}
  C_1^*$, and from Lemma~\ref{thm:nontransitive}, we see that
  $\pairscorefull{d}{C_1^*}{C_2^*}{P} \geq 0$ and
  $\pairscorefull{d}{C_1^*}{C_2^*}{Q} \leq 0$. That is, at least one
  of the voting situations $P$ and $Q$ lies on the hyperplane, or they
  both lie on different sides of the hyperplane. This gives a
  contradiction and proves our claim.

  Now, using the committee dominance axiom, we exclude the second possibility.
  For each $i \in [m-k+1]$ we set $I_i = \{i, i+1, \ldots, i+k-1\}$.
  Let $I$ and $J$ denote, respectively, the best possible and the
  worst possible position of a committee, i.e., $I = I_1$ and $J =
  I_{m-k+1}$.  For the sake of contradiction, let us assume that there
  exists a profile $P'$, where $\pairscorefull{g}{C_1^*}{C_2^*}{P'} >
  0$ and $C_2^* \succ_{P'} C_1^*$.  Since there exists a profile with
  $\pairscorefull{g}{C_1^*}{C_2^*}{P'} > 0$, it must be the case that
  $\lambda(I) > \lambda(J)$ (otherwise $\lambda$ would be a constant
  function). Thus there must exist $p$ such that $\lambda(I_p) >
  \lambda(I_{p+k-x})$. Let us consider a profile $S$ consisting of a
  single vote where $C_1^*$ stands on position $I_p$ and $C_2^*$
  stands on position $I_{p+k-x}$ (as
  $|C_1^* \cap C_2^*| = x$, this is possible).  Since $\lambda(I_p) >
  \lambda(I_{p+k-x})$, we have that
  $\pairscorefull{g}{C_1^*}{C_2^*}{S} > 0$. By committee-dominance of
  $f_k$, it follows that $C_1^* \succeq_S C_2^*$. However, from the
  reasoning in the preceding paragraph (applied to profile $S$), we
  know that either $C_1^* \succ_S C_2^*$ or $C_2^* \succ_S C_1^*$.
  Putting these two facts together, we conclude that $C_1^* \succ_S
  C_2^*$.  Since we have shown a single profile $S$ such that
  $\pairscorefull{g}{C_1^*}{C_2^*}{S} > 0$ and $C_1^* \succ_S C_2^*$,
  by the argument from the previous paragraph, we know that for every
  profile $P$ it holds that:
 \[
   \text{If } \pairscorefull{g}{C_1^*}{C_2^*}{P} > 0 \text{ then } C_1^* \succ_P C_2^*.
 \]
 Our choice of committees $C_1^*$ and $C_2^*$ was arbitrary and, thus,
 the above implication holds for all pairs of committees. This
 completes the proof.
\end{proof}

Due to Lemma~\ref{thm:nontransitive2}, in our further discussion,
given a symmetric, consistent, committee-dominant, continuous
$k$-winner election rule $f_k$ we can focus solely on the subspace
$\{P\colon C_1 =_P C_2\}$.
%
%Given Lemma~\ref{thm:nontransitive2}, in our following discussion we
%can focus on the space $\{P\colon C_1 =_P C_2\}$. 
If we manage to show that committees $C_1$ and $C_2$ are equivalent if
and only if the score of $C_1$ is equal to the score of $C_2$
according to some committee scoring function $\lambda$, then we can
conclude that $f_k$ is a committee scoring rule defined by
this committee scoring function $\lambda$.
This important observation concludes the first part of the proof.

%Until the end of Section~\ref{sec:mainResult} we assume $f_k$ to be an anonymous,
%neutral, and consistent.

\subsection{Second Part of the Proof: Committees with All but One
  Candidate in Common}\label{sec:second-part-a}

We now start the second part of the proof. The current section is
independent from the results of the previous one, but we do use all
the notation that was introduced and, in particular, we consider
voting situations over $\rationals^{m!}$. We will use results from
Sections~\ref{sec:nontransitiveRules}~and~\ref{sec:toolsforcsrs} only in
Section~\ref{sec:specialCase_base}, where we conclude the whole proof.

\paragraph{The Setting and Our Goal.}
As before, the size of committees is denoted as $k$. Throughout this
section we assume $f_k$ to be a $k$-winner election rule that is
symmetric, consistent, committee-dominant, and continuous.  Our goal
is to show that as long as we consider committees that contain some
$k-1$ fixed members and can differ only in the final one, $f_k$ acts
on such committee pairs as a committee scoring rule.  The discussion
in this section is inspired by that of
Young~\shortcite{you:j:scoring-functions} and
Merlin~\shortcite{merlinAxiomatic}, but the main part of our analysis
is original (in particular Lemma~\ref{lemma:zeroAlpha}).

\paragraph{Position-Difference Function.}

% Let us define several more concepts that will be useful in our
% further analysis.

Let $P$ be a voting situation in $\rationals^{m!}$, $C$ be some
size-$k$ committee, and $I$ be a committee position. We define the
weight of position $I$ with respect to $C$ within $P$ as:
\[
\posweight_I(C,P) = \sum_{\pi \in \orders \colon \pos_\pi(C) = I} P(\pi),
\]
That is, $\posweight_I(C,P)$ is the (rational) number of votes in
which committee $C$ is ranked on position $I$.

For each two committees $C_1, C_2$ such that $|C_1 \cap C_2| = k-1$,
we define a committee position-difference function
$\alpha_{C_1,C_2}\colon \rationals^{m!} \to \rationals^{m \choose k}$
that for each voting situation $P \in \rationals^{m!}$ returns a
vector of ${m \choose k}$ elements, indexed by committee positions
(i.e., elements of $[m]_k$), such that for each committee position
$I$, we have: 
\[
   \alpha_{C_1,C_2}(P)[I] \: = \: \posweight_I(C_1,P) - \posweight_I(C_2,P).
\]
%
% , where $P_I[C] = \sum_{v \in
%   \orders: \text{the set of positions of $C$ in $v$ is
%     $I$}}P[v]$. 
%
Naturally, $\alpha_{C_1,C_2}(P)$ is a linear function of $P$.
We %also
claim that for each voting situation $P$, we have:
\begin{equation}
\label{sumI=0}
\sum_{I \in [m]_k} \alpha_{C_1,C_2}(P)[I] = 0.
\end{equation}
To see why this is the case, we note that $\sum_{I \in [m]_k}
\posweight_I(C_1,P) = \sum_{\pi \in \orders} P(\pi)$ because every
vote is accounted exactly once. Thus, we have that:
\begin{align*}
  \sum_{I \in [m]_k} \alpha_{C_1,C_2}(P)[I]
  & = \sum_{I \in [m]_k}\bigg( \posweight_I(C_1,P) - \posweight_I(C_2,P) \bigg)\\
  & = \sum_{I \in [m]_k} \posweight_I(C_1,P) - \sum_{J \in [m]_k}\posweight_J(C_2,P) \\
  & = \sum_{\pi \in \orders} P(\pi) - \sum_{\pi' \in \orders} P(\pi')
  = 0.
\end{align*}

Position-difference functions will be important technical tools that
we will soon use in the proof (in particular, in
Lemma~\ref{lemma:zeroAlpha} we will show that if $\alpha_{C_1,C_2}(P)
= \langle 0, \ldots, 0\rangle$ then $C_1 =_P C_2$ ).  However, we need
to provide some more tools first.

\paragraph{Johnson Graphs and Hamiltonian Paths.}

We will need the following graph-theoretic results to build certain
votes and preference profiles in our following analysis. We mention
that the graphs that Lemmas~\ref{lemma:hamiltonian1}
and~\ref{lemma:hamiltonian2} speak of are called Johnson
graphs. Lemma~\ref{lemma:hamiltonian1} was known before (we found the
result in the work of Asplach~\cite{Asp13} and could not trace an
earlier reference\footnote{We suspect the results might have been
  known before the work of Asplach. Indeed, similar results appear in
  the form of algorithms that output all size-$k$ subsets of a given
  set in the order so that each two consecutive sets differ in only
  one element. Yet, we need the specific variants provided in
  Lemmas~\ref{lemma:hamiltonian1} and~\ref{lemma:hamiltonian2} that
  finish the Hamiltonian path on a specific
  vertex. Asplach~\cite{Asp13} does not mention directly that his
  proofs provide this property, but close inspection shows that this
  is the case.}), and we provide the proof for the sake of
completeness.

\begin{lemma}\label{lemma:hamiltonian1}
  Let $p$ and $j$ be integers such that $1 \leq j \leq p$. Let $G(j,
  p)$ be a graph constructed in the following way. We associate
  $j$-element subsets of $\{1, \ldots p\}$ with vertices and we say
  that two vertices are connected if the corresponding subsets differ
  by exactly one element (they have $j-1$ elements in common). Such a
  graph contains a Hamiltonian path, i.e., a path that visits each
  vertex exactly once, that starts from the set $\{1, \ldots, j\}$ and
  ends in the set $\{p-j+1, \ldots, p\}$.
\end{lemma}
\begin{proof}
  We prove this lemma by induction over $j$ and $p$.
  % For $j = p$ the graph is empty or contains a single vertex, and by
  % convention such graphs contain Hamiltonian paths.
  For $j =1$ and for each $p \geq 1$, it is easy to see that the
  required path exists (in this case, the graph is simply a full
  clique). This provides the induction base.  For the inductive step,
  we assume that there are two numbers, $p'$ and $j'$, such that for
  each $p$ and $j$ ($j \leq p$) such that $p < p'$ and $j < j'$ it
  holds that graph $G(j, p)$ contains a Hamiltonian path satisfying
  the constraints from the lemma.  We will prove that such a path also
  exists for $G(j',p')$.

  We partition the set of vertices of $G(j', p')$ into $p'-j+1$ groups
  $V(j', p', 1), \ldots, V(j', p', p'-j+1)$, where for each $x \in
  \{1, \ldots, p'-j+1\}$, group $V(j', p', x)$ consists of all sets of
  $j$ elements (vertices of the graph) such that $x$ is the lowest
  among them.

  We build our Hamiltonian path for $G(j', p')$ as follows. We start
  with the vertex $\{1, \ldots, j'\}$. By our inductive hypothesis, we
  know that there is a path that starts with $\{1, \ldots, j'\}$,
  traverses all vertices in $V(j', p', 1)$, and ends in $\{1, p'-j'+2,
  \ldots, p'\}$. From $\{1, p'-j'+2, \ldots, p'\}$ we can go, over a
  single edge, to $\{2, p'-j'+2, \ldots, p'\}$. Starting with this
  vertex, by our inductive hypothesis, we can traverse all the
  vertices of $V(j', p', 2)$. Then, over a single edge, we can move to
  some vertex from $V(j', p', 3)$, traverse all the vertices there,
  and so on.  By repeating this procedure, we will eventually reach
  some vertex in the set $V(j',p', p'-j'+1)$. However, $V(j',p',
  p'-j'+1)$ contains exactly one vertex, $\{p'-j'+1, \ldots,
  p'\}$. This means that we have found the desired Hamiltonian path.
\end{proof}

\begin{lemma}\label{lemma:hamiltonian2}
  Let $r, p$ and $j$ be integers such that $1 \leq r \leq p$ and $1
  \leq j \leq p-1$. Let $\tilde{G}(j, p, r)$ be a graph constructed in
  the following way: (i) A $j$-element subset of $\{1, \ldots, p\}$ is
  a vertex of $\tilde{G}(j, p, r)$ if and only if it contains at least
  one element smaller than $r$. (ii) There is an edge between two
  vertices if they differ in exactly one element (i.e., if they have
  $j-1$ elements in common). Such a graph contains a Hamiltonian path.
\end{lemma}
\begin{proof}
  The proof is very similar to the previous one. We partition the set
  of vertices of $\tilde{G}(j, p, r)$ into $r-1$ groups $V(j, p, 1),
  \ldots, V(j, p, r-1)$, where for each $x \in \{1, \ldots, r-1\}$,
  group $V(j, p, x)$ consists of all the sets (i.e., all the vertices)
  such that $x$ is their smallest member.

  We build our Hamiltonian path for $\tilde{G}(j, p, r)$ as
  follows. We start with the vertex $\{1, \ldots, j\}$. By
  Lemma~\ref{lemma:hamiltonian1}, we can continue the path from $\{1,
  \ldots, j\}$, traverse all vertices in $V(j, p, 1)$, and end in
  $\{1, p-j+2, \ldots, p\}$. From $\{1, p-j+2, \ldots, p\}$ we can go,
  over a single edge, to $\{2, p-j+2, \ldots, p\}$, and we can
  traverse all vertices in $V(j, p, 2)$. Then we can go, over a single
  edge, to some vertex from $V(j, p, 3)$, and we can continue in the
  same way as in the proof of Lemma~\ref{lemma:hamiltonian1}.
\end{proof}

\paragraph{The Range of $\boldsymbol{\alpha_{C_1,C_2}}$.}

Let us consider two distinct committees $C_1$ and $C_2$. Using
Lemma~\ref{lemma:hamiltonian1}, we establish the dimension of the
range of function $\alpha_{C_1,C_2}$. This result will be useful in
the proof of Lemma~\ref{lemma:zeroAlpha}.

\begin{lemma}\label{lem:codomain}
  For two committees, $C_1$ and $C_2$, the range of the function
  $\alpha_{C_1,C_2}$ has dimension ${m \choose k} - 1$.
\end{lemma}
\begin{proof}
  From Equation~\eqref{sumI=0}, we get that the dimension of the
  range of function $\alpha_{C_1,C_2}$ is at most ${m \choose k} -
  1$.  Now, let us consider graph $G = G(k, m)$ from
  Lemma~\ref{lemma:hamiltonian1} and the Hamiltonian path specified in
  this lemma. Note that we can understand each vertex in $G$ as a
  committee position.  For each edge $(I, I')$ on our Hamiltonian
  path, consider a single vote where $C_1$ stands on position $I$ and
  $C_2$ stands on position $I'$. For such a vote, $\alpha_{C_1,C_2}$
  returns a vector with all zeros except a single $1$ on position $I$
  and a single $-1$ on position $I'$. It is easy to observe that there
  are ${m \choose k} - 1$ such votes and that so constructed vectors
  are linearly independent.
\end{proof}

\paragraph{($\boldsymbol{C_1}$,$\boldsymbol{C_2}$)-Symmetric Profiles.}
The final tool that we need to provide before we prove
Lemma~\ref{lemma:zeroAlpha} is the definition of ($C_1$, $C_2$)-symmetric
profiles.  Suppose $\sigma$ is a permutation of $A$. Then we can
extend its action to linear orders and voting situations in the
natural way.

\begin{definition}\label{def:c1c2symetric}
  Let $C_1$ and $C_2$ be two size-$k$ committees. We say that a voting
  situation $P$ is \emph{($C_1$,~$C_2$)-symmetric} if there exists a
  permutation of the set of candidates $\sigma$ and a sequence of
  committees $F_1, F_2, \ldots, F_x$ such that $P = \sigma(P)$ and:
  \begin{enumerate}
%  \item $P = \sigma(P)$,
  \item $C_1 = F_1 = F_x$ and $C_2 = F_2$,
  \item for each $i \in [x-1]$ it holds that $\sigma(F_i) = F_{i+1}$.
  \end{enumerate} 
\end{definition}
\noindent
If a voting situation $P$ is ($C_1$,~$C_2$)-symmetric then we know
that $C_1 =_P C_2$. Why is this the case? For the sake of
contradiction let us assume that $C_1 \neq_P C_2$, and, without loss
of generality, that $C_1 \succ_P C_2$.  From $C_1 \succ_P C_2$ (which
translates to $F_1 \succ_P F_2$) by neutrality of $f_k$ we infer that
$F_2 \succ_{\sigma(P)} F_3$, thus that $F_2 \succ_{P} F_3$. By the
same arguments, we get that $F_1 \succ_P F_2 \succ_P F_3 \succ_P
\cdots \succ_P F_x$.
In consequence, we get that $C_1 \succ_P C_1$, a contradiction.

Further, we observe that for each ($C_1$,~$C_2$)-symmetric voting
situation $P$ it holds that $\alpha_{C_1,C_2}(P) =
\zerovector$. Indeed, %using notation from
if $\sigma$ is as in Definition~\ref{def:c1c2symetric}, we note that
since $\sigma(C_1) = C_2$ and since $\sigma(P) = P$, for each
(fractional) vote in $P$ where committee $C_1$ stands on some position
$I$ we can uniquely assign a (fractional) vote in $P$ where committee
$C_2$ stands on the same position $I$.
This shows that $\alpha_{C_1,C_2}(P)[I]$ is a
vector of non-positive numbers. By an analogous argument (using the
fact that $\sigma^{(-1)}(C_2) = C_1$ and $\sigma^{(-1)}(P) = P$) we
infer that $\alpha_{C_1,C_2}(P)[I]$ is a vector of nonnegative numbers,
and, so, we conclude that
$\alpha_{C_1,C_2}(P) = \zerovector$.

\paragraph{Inferring Committee Equivalence Using $\boldsymbol{\alpha_{C_1,C_2}}$.}
We are ready to present
Lemma~\ref{lemma:zeroAlpha}, our main technical tool required in this
part of the proof. On the intuitive level, it says that for $|C_1 \cap
C_2| = k-1$ the information provided by the function
$\alpha_{C_1,C_2}$ in relation to a profile $P$ is sufficient to
distinguish whether $C_1$ is equivalent to $C_2$ with respect to $P$.

% We write $\overline{0}$ to denote the vector $\langle 0, \ldots, 0
% \rangle$ (its dimension will follow from the context).

\begin{lemma}\label{lemma:zeroAlpha}
  For each two committees $C_1, C_2 \in \powA$ such that $|C_1 \cap
  C_2| = k-1$ and for each voting situation $P \in \rationals^{m!}$,
  if $\alpha_{C_1,C_2}(P) = \zerovector$ then $C_1 =_P C_2$.
\end{lemma}

\begin{proof}
  The kernel of a linear function is the space of all vectors for
  which this function returns the zero vector. In particular, the
  kernel of $\alpha_{C_1,C_2}$, denoted $\ker(\alpha_{C_1,C_2})$, is
  the space of all voting situations $P$ such that
  $\alpha_{C_1,C_2}(P) = \zerovector$.
  Since the domain of function $\alpha_{C_1,C_2}$ has dimension $m!$
  and, by Lemma~\ref{lem:codomain}, its range has dimension ${m
    \choose k} - 1$, the kernel of $\alpha_{C_1,C_2}$ has dimension
  $m! - {m \choose k} + 1$.  We will construct a base of this kernel
  that will consists of ($C_1$,~$C_2$)-symmetric voting situations
  only. Since for each ($C_1$,~$C_2$)-symmetric voting situation $P$
  it holds that $C_1 =_P C_2$ and $\alpha_{C_1,C_2}(P) = \zerovector$,
  by consistency of $f_k$ and linearity of $\alpha_{C_1,C_2}$ we will
  %get the thesis 
  prove the conclusion of the theorem.

  We prove the statement by a two-dimensional induction on $k$
  (the committee size) and $m$ (the size of the set of candidates). As a base
  for the induction we will show that the property holds for $k = 1$
  and all values of $m$. For the inductive step we will show that from
  the fact that the property holds for committee size $j-1$ and for
  $p-1$ candidates it follows that the property also holds for
  committee size $j$ and for $p$ candidates. This will allow us to
  conclude that the property holds for all values of $m$ and $k$ with
  $m \geq k$.

  For $k = 1$ and for an arbitrary value of $m$, the problem collapses
  to the single-winner setting. It has been shown by
  Young~\shortcite{you:j:scoring-functions} (and by
  Merlin~\shortcite{merlinAxiomatic}) that for each two candidates
  $c_1$ and $c_1'$, there exists a base of
  $\ker(\alpha_{\{c_1\},\{c_1'\}})$ that consists of $m! - (m-1)$
  voting situations which are ($\{c_1\}$,~$\{c_1'\}$)-symmetric. This
  gives us the base for the induction.

  Let us now prove the inductive step.  We want to show that the
  statement is satisfied for $A_p= \{a_1, a_2, \ldots a_{p}\}$, $C_{1,
    j} = \{a_1, a_2, \ldots, a_{j}\}$ and $C_{2, j} = \{a_1', a_2,
  \ldots, a_{j}\}$, where we set $a'_1 = a_{j+1}$. (We note that since
  $f_k$ is symmetric, the exact names of the candidates we use here
  are irrelevant, and we picked these for notational convenience.)
  From the sets $A_p$, $C_{1, j}$ and $C_{2, j}$ we take out element
  $a_j$ and get $A_{p-1} = \{a_1, a_2, \ldots, a_{j-1}, a_{j+1},
  \ldots a_{p}\}$, $C_{1, (j-1)} = \{a_1, a_2, \ldots, a_{j-1}\}$ and
  $C_{2, (j-1)} = \{a_1', a_2, \ldots, a_{j-1}\}$.  Let $V_{j-1}$ be a
  base of $\ker(\alpha_{C_{1, (j-1)}, C_{2, (j-1)}})$ that consists of
  ($C_{1, (j-1)}$,~$C_{2, (j-1)}$)-symmetric voting situations. We know
  that it exists from the induction hypothesis. We also know that it
  consists of $(p-1)! - {p - 1 \choose j-1} + 1$ voting situations. We
  now build the desired base for $\ker(\alpha_{C_{1, j}, C_{2, j}})$
  using $V_{j-1}$ as the starting point. Our base has to consist of
  $p! - {p \choose j} + 1$ linearly independent,
  ($C_{1,j}$,~$C_{2,j}$)-symmetric voting situations. \medskip

  First, for each voting situation $P \in V_{j-1}$ and for each $r \in
  \{1, \ldots p\}$ we create a voting situation $P_r$ as follows. We
  take each vote $v$ in $P$ and we put $a_j$ in the $r$-th position of
  $v$, pushing the candidates on positions $r, r+1, r+2, \ldots$ back
  by one position, but keeping their relative order unchanged. There
  are $p! - p{p - 1 \choose j-1} + p$ such vectors and it is easy to
  see that they are linearly independent. Let us refer to the set of
  these vectors as $B_1$. Naturally, the vectors from $B_1$ do not
  span the whole space $\ker(\alpha_{\{a_1, \ldots, a_j\},\{a_1',
    \ldots, a_j\}})$; there is simply too few of them. However, there
  is also a certain structural reason for this and understanding this
  reason will help us further in the proof.  Let $\lin(B_1)$ denote
  the set of linear combinations of voting situations from $B_1$. For
  each $r \in \{1, \ldots, p\}$ and each $T \in \lin(B_1)$, let $T(a_j
  \rightarrow r)$ denote the voting situation that consists of all
  votes from $T$ which have $a_j$ on the $r$-th position. We can see
  that for each $r \in \{1, \ldots, p\}$ and each $T \in \lin(B_1)$, it
  holds that $\alpha_{C_{1,j},C_{2,j}}(T(a_j \rightarrow r)) =
  \zerovector$ (the reason for this is that $T(a_j \rightarrow r)$ is,
  in essence, a linear combination of voting situations from
  $V_{j-1}$, with $a_j$ inserted at position $r$) . This property
  certainly does not hold for all the voting situations in
  $\ker(\alpha_{\{a_1, \ldots, a_j\},\{a_1', \ldots,
    a_j\}})$. \medskip

  We now form the second part of our base, denoted $B_2$ and
  consisting of $p {p-1 \choose j-1} \cdot \frac{j-1}{j} - (p-1)$
  voting situations (($C_{1,j}$,~$C_{2,j}$)-symmetric and linearly
  independent from each other and all the voting situations in
  $B_1$). We start constructing each voting situation in $B_2$ by
  constructing its distinctive vote. To construct a distinctive vote,
  we first select the position for candidate $a_j$; we consider each
  position from $\{1, \ldots, p\}$. Let us fix $r \in \{1, \ldots, p\}$
  as the position that we picked. Next, we select a set of $j$
  positions for the candidates from $\{a_1, \ldots, a_{j-1},
  a_1'\}$. To do that, we first construct the following graph. We
  associate all sets of $j-1$ positions such that $r$ is
  greater\footnote{There is a possible point of confusion here. By
    ``greater'' we mean greater as a number. So, for example, position
    $7$ is greater than position $5$ (even though we would say that a
    candidate ranked on position $5$ is ranked higher than candidate
    ranked on position $7$).} than at least one of them with vertices
  (for a fixed $r$ there are ${p-1 \choose j-1} - {p-r \choose j-1}$
  such vertices; we choose $j-1$ positions out of $p-1$ still
  available, but we omit the situations where all these $j-1$
  positions are greater than $r$). We say that two vertices are
  connected if the corresponding sets differ by exactly one
  element. From Lemma~\ref{lemma:hamiltonian2} it follows that such a
  graph contains a Hamiltonian path.
  % Let us treat this path as a path of directed edges.
  Now, for each edge $(X, X')$ on the considered Hamiltonian path we
  do the following. Let $B = X \cap X'$, and let $b$ and $b'$ be the
  two elements such that $b < b'$ and $\{b, b'\} = (X \setminus B)
  \cup (X' \setminus B)$.  (In other words, $b$ and $b'$ are the two
  elements on which $X$ and $X'$ differ.)  Note that $|B| = j-2$.  We
  form a distinctive vote by putting candidate $a_j$ on position $r$,
  candidates $a_2, \ldots, a_{j-1}$ on the positions from $B$ (in some
  arbitrary order), $a_1$ on position $b$, $a_1'$ on position $b'$,
  and all the other candidates on the remaining positions (in some
  arbitrary order).

  How many distinctive votes have we constructed? There are $p$
  possible values for the position of $a_j$, and for each such
  position we consider a graph. If the position of $a_j$ is $r$, then
  the graph has ${p-1 \choose j-1} - {p-r \choose j-1}$ vertices.
  Thus, altogether, the number of vertices is:
  %
 % \footnote{The second
 %   equality follows from a standard property of binomial
  %  coefficients: for $m,n \in \naturals$ we have $\sum_{k=0}^n{k
  %    \choose m} = {n+1 \choose m+1}.$}
  \begin{align*}
    \sum_{r=1}^{p} \left( {p-1 \choose j-1} - {p-r \choose j-1}
    \right)
    & = p {p-1 \choose j-1} - \sum_{r=1}^p {p-r \choose j-1} \\
    & = p {p-1 \choose j-1} - {p \choose j} = p {p-1 \choose j-1} - \frac{p}{j}{p-1 \choose j-1} = p{p-1 \choose j-1}\frac{j-1}{j},
  \end{align*}
  where the second equality follows from the following property of
  binomial coefficients: for $m,n \in \naturals$ we have
  $\sum_{k=0}^n{k \choose m} = {n+1 \choose m+1}.$ (An intuitive way
  to obtain the same result is as follows.  Let us fix the value $r$
  chosen uniformly at random. The vertices for the graph for this
  value of $r$ are size-$(j-1$) subsets of $p-1$ positions, except
  those subsets that contain only elements greater than $r$.  By
  symmetry, on the average the number of subsets that we omit is a
  $\nicefrac{1}{j}$ fraction of all the subsets. Since we have all the
  graphs for all values of $r$, altogether we have $p{p-1 \choose
    j-1}\frac{j-1}{j}$ vertices.)  One of the graphs is empty (it is
  the one that is constructed for $r = 1$, because there is no element
  in $\{1, \ldots, p\}$ lower than $r=1$). Thus we have $p-1$
  non-empty graphs. As a result, the total number of edges in the
  considered Hamiltonian paths is $p{p-1 \choose j-1}\frac{j-1}{j} -
  (p-1)$.  Every edge corresponds to a distinctive vote, so this is
  also the number of distinctive votes constructed.

 For each distinctive vote $v$ constructed, we build the following voting
  situation:

  \begin{description}
  \item[Case 1.] If $a_1$ and $a_1'$ are both ranked ahead of
  $a_j$, then we let $\tau$ be permutation $\tau:= (a_1, a_j, a_1')$
  (i.e., we let $\tau$ be the identity permutation except that
  $\tau(a_1) = a_j$, $\tau(a_j) = a_1'$, $\tau(a_1') = a_1$)
 %$\{a_1 \rightarrow a_1', a_1' \rightarrow a_j, a_j \rightarrow a_1\}$
  and we let the voting situation consist of three votes, $v$, $\tau(v)$, and
  $\tau^{(2)}(v)$:
  \begin{align*}
    v&\colon \cdots \succ a_1 \succ \cdots \succ a_1' \succ \cdots \succ a_j \succ \cdots \\
    \tau(v)&\colon \cdots \succ a_j \succ \cdots \succ a_1 \succ
    \cdots \succ a_1' \succ \cdots \\
    \tau^{(2)}(v)&\colon \cdots \succ a_1' \succ \cdots \succ a_j \succ \cdots \succ a_1 \succ \cdots
  \end{align*}
  Note that permutation $\tau$ and the sequence $F_1 = \{a_1, \ldots, a_j\}$, $F_2 = \{a_2, \ldots, a_j, a_1'\}$, $F_3 = \{a_1, \ldots, a_{j-1}, a_1'\}$, $F_4 = \{a_1, \ldots, a_{j}\}$ witness that this voting situation is
  ($C_{1,j}$,~$C_{2,j}$)-symmetric.

\item[Case 2.] If it is not the case that $a_1$ and $a'_1$ are both
  ranked ahead of $a_j$ in distinctive vote $v$, then we know that
  there is some other candidate $a \in \{a_2, \ldots, a_{j-1}\}$
  ranked ahead of $a_j$. This is due to our construction of
  distinctive votes---we always put $a_j$ on position $r$ and make
  sure that there is some candidate ranked on a position ahead of
  $r$. If all the candidates $a_2, \ldots, a_{j-1}$ were ranked behind
  $a_j$, then it would have to be the case that both $a_1$ and $a'_1$
  are ranked ahead of $a_j$.\footnote{To see why this is the case,
    recall how the distinctive votes are produced. We have an edge
    $(X,X')$ on a Hamiltonian path in our graph. We set $B = X \cap
    X'$ and $\{b,b'\} = (X \setminus B) \cup (X' \setminus B)$.  $B$
    contains positions of the candidates $a_2, \ldots, a_{j-1}$,
    whereas $b$ and $b'$ are positions of $a_1$ and $a'_1$. Without
    loss of generality, we can take $X = B \cup \{b\}$ and $X' = B
    \cup \{b'\}$.  Since---by our assumption here---the positions of
    $a_2, \ldots, a_{j-1}$ (i.e., the positions in $B$) are greater
    than the position of $a_j$ (denoted $r$ in the description of
    distinctive votes construction), for $X$ and $X'$ to be vertices
    in the graph, we need both $b$ and $b'$ to be smaller than $r$
    (and, in effect, both $a_1$ and $a'_1$ precede $a_j$).}  Since it
  is not the case that both $a_1$ and $a'_1$ are ranked ahead of
  $a_j$, there must be some other candidate from $\{a_2, \ldots,
  a_{j-1}\}$ that is. We call this candidate $a$.
  We let $\rho$ be permutation $\rho:= (a_1, a_1')(a, a_j)$ (i.e., we
  let $\rho$ be the identity permutation, except that it swaps $a_1$
  with $a'_1$ and $a$ with $a_j$). We form a voting situation that
  consists of $v$ and $\rho(v)$:
  \begin{align*}
    v&\colon \cdots \succ a \succ \cdots \succ a_j \succ \cdots \succ a_1 \succ \cdots \succ a_1' \succ \cdots \\
    \rho(v)&\colon \cdots \succ a_j \succ \cdots \succ a \succ
    \cdots \succ a'_1 \succ \cdots \succ a_1 \succ \cdots
  \end{align*}
  Permutation $\rho$ and the sequence $F_1 = \{a_1, \ldots, a_j\}$, $F_2 = \{a_2, \ldots, a_j, a_1'\}$, $F_3 = \{a_1, \ldots, a_{j}\}$ witness that this is a
  ($C_{1,j}$,~$C_{2,j}$)-symmetric voting situation.  
  \end{description}
  Let $B_2$ consists of all the voting situations constructed from the
  distinctive votes.\medskip

  For each $r \in \{1, \ldots, p\}$, each set of $j-1$ positions $R$
  from $\{1, \ldots, p\} \setminus \{r\}$, and each voting situation
  $P$, we define $\gamma_{r, R}(P)$ to be the total (possibly
  fractional) number of votes from $P$ that have $a_j$ on the $r$-th
  position and that have candidates from $\{a_1, a_2, \ldots,
  a_{j-1}\}$ on positions from $R$. We define $\gamma'_{r, R}(P)$
  analogously, for the votes where $a_j$ is on position $r$ and
  candidates $a'_1, a_2, \ldots, a_{j-1}$ take positions from $R$.  We
  define $\beta_{r,R}(P)$ to be $\gamma_{r,R}(P) - \gamma'_{r,R}(P)$.
  For example, for each $P \in B_1$ we have $\beta_{r, R}(P) = 0$.
  
  Let us consider voting situations from $B_2$ which were created from
  a single Hamiltonian path in one of the graphs. The distinctive
  votes for all these voting situations have $a_j$ on the same
  position; we denote this position by $r$. For each such voting
  situation $P$, each non-distinctive vote belonging to $P$ has $a_j$
  on a position ahead of position $r$. Further, we see that there
  exist exactly two sets $R_1$ and $R_2$ such that $\beta_{r, R_1}(P)
  \neq 0$ and $\beta_{r, R_2}(P) \neq 0$. These are the sets that
  correspond to the vertices connected by the edge from which the
  distinctive vote for $P$ was created (for one of them, let us say
  $R_1$, we have $\beta_{r, R_1}(P) = 1$, and for the other we have
  $\beta_{r, R_2}(P) = -1$; to see that this holds, recall that $a_j$
  is ranked on positions ahead of $r$ in non-distinctive votes and,
  thus, it suffices to consider the distinctive vote only).

  Now we are ready to explain why the vectors from $B_1 \cup B_2$ are
  linearly independent. For each nontrivial linear combination $L$ of
  the vectors from $B_1 \cup B_2$ we will show that $L$ cannot be
  equal to the zero vector. For the sake of contradiction let us
  assume that $L = \zerovector$. We start by showing that all
  coefficients of vectors from $B_2$ in $L$ are equal to zero. Again,
  for the sake of contradiction let us assume that this is not the
  case. Let $B'_2$ consist of those vectors from $B_2$ that appear in
  $L$ with non-zero coefficients. Let $r$ be the largest position of
  $a_j$ in some vote in $B'_2$ (by ``largest position'' we mean
  largest numerically, i.e., for each vote $v$ that occurs in some
  voting situation from $B'_2$ it holds that $\pos_v(a_j) \leq
  r$). Let $B'_{2,r}$ be the set of all voting situations from $B'_2$
  that have some votes which have $a_j$ on position $r$. Each voting
  situation in $B'_{2,r}$ consists of either two or three
  votes. However, the votes belonging to those voting situations which
  have $a_j$ on position $r$ must be distinctive votes (all
  non-distinctive votes for voting situations in $B_2$ have $a_j$ on
  positions ahead of $r$). Each such distinctive vote is built from an
  edge of a single Hamiltonian path (they come from the same
  Hamiltonian path because otherwise they would not have $a_j$ on the
  same position). Let $S$ be a voting situation in $B'_{2,r}$ that has
  a distinctive vote built from the latest edge on the path, among the
  edges that contributed voting situations to $B'_{2,r}$ (to make this
  notion meaningful, we orient the path in one of the two possible
  ways). Let $R_1$ and $R_2$ be the sets of $j-1$ positions that form
  this edge. By the reasoning from the previous paragraph we have that
  $\beta_{r, R_1}(S) \neq 0$, $\beta_{r, R_2}(S) \neq 0$, and one of
  the following two conditions must hold (depending on the orientation
  of the Hamiltonian path that we chose):
  \begin{enumerate}
  \item For each voting situation $Q'$ in $B'_2$ other than $S$ we have
    $\beta_{r, R_1}(Q') = 0$.
  \item For each voting situation $Q'$ in $B'_2$ other than $S$ we have
    $\beta_{r, R_2}(Q') = 0$.
  \end{enumerate}
  Further, for each $Q \in B_1$ we have $\beta_{r, R_1}(Q) =
  \beta_{r, R_2}(Q) = 0$.  Thus, since $\beta_{r, R_1}$ and
  $\beta_{r, R_2}$ are linear functions, we have that either
  $\beta_{r, R_1}(L) \neq 0$ or $\beta_{r, R_2}(L) \neq 0$.  Thus,
  $L$ cannot be a zero-vector, which gives a contradiction.

  We have shown that all coefficients of vectors from $B_2$ used to
  form $L$ are equal to zero. Thus $L$ must be a linear combination of
  vectors from $B_1$. However, the vectors from $B_1$ are linearly
  independent, which means that if $L$ is $\zerovector$, then the
  coefficients of all the vectors from $B_1$ are zeros.  Thus we
  conclude that the vectors from $B_1 \cup B_2$ are linearly
  independent.

  It remains to show that $B_1 \cup B_2$ indeed forms a base of the
  kernel of $\alpha_{C_{1,j},C_{2,j}}$. Since vectors in $B_1$ and
  $B_2$ are linearly independent, it suffices to check that the
  cardinality of $B_1 \cup B_2$ is equal to the dimension of
  $\ker(\alpha_{C_{1,j},C_{2,j}})$.  The number of vectors in $B_1
  \cup B_2$ is equal to:
  \begin{align*}
    \underbrace{\left( p! - p{p - 1 \choose j-1} + p \right)}_{|B_1|}
    + \underbrace{\left( p {p-1 \choose j-1} \cdot \frac{j-1}{j} - p +
        1\right)}_{|B_2|} = p! - \frac{p}{j} {p-1 \choose j-1} + 1 =
    p! - {p \choose j} + 1.
  \end{align*}
  This completes our induction. The proof works for arbitrary
  committees $C_1$ and $C_2$ with $|C_1\cap C_2|=k-1$ due to symmetry
  of $f_k$.
\end{proof}

We are almost ready to show that for committees that differ by one
candidate only, $f_k$ is a committee scoring rule, and to derive its
committee scoring function. However, before we do that we need to
change the domain once again. We will also need some notions from
topology.

\paragraph{Topological Definitions.}
For every set $S$ in some Euclidean space $\reals^n$, by $\inter(S)$
we mean the interior of $S$, i.e., the largest (in terms of inclusion)
open set contained in $S$. By $\cvx(S)$ we mean the convex hull of
$S$, i.e., the smallest (in terms of inclusion) convex set that
contains $S$. Finally, by $\overline{S}$ we define the closure of $S$,
i.e., the smallest (in terms of inclusion) closed set that contains
$S$.
We use the concept of $\rationals$-convex sets of
Young~\shortcite{you:j:scoring-functions} and we recall his two
observations.

\begin{definition}[$\rationals$-convex sets]
  A set $S \subseteq \reals^n$ is $\rationals$-convex if $S \subseteq
  \rationals^n$ and for each $s_1, s_2 \in S$ and each $q \in
  \rationals$, $0 \leq q \leq 1$, it holds that $q \cdot s_1 + (1 - q)
  \cdot s_2 \in S$.
\end{definition}

\begin{lemma}[Young~\shortcite{you:j:scoring-functions}]\label{lemma:qconvex1}
  Set $S \subseteq \reals^n$ is $\rationals$-convex if and only if $S
  = \rationals^n \cap \cvx(S)$.
\end{lemma}

\begin{lemma}[Young~\shortcite{you:j:scoring-functions}]\label{lemma:qconvex2}
  If a set $S$ is $\rationals$-convex, then $\overline{S} =
  \overline{\cvx(S)}$; moreover, $\overline{S}$ is convex.
\end{lemma}

\paragraph{Third Domain Change.}

In the following arguments, we fix two arbitrary committees $C_1$ and
$C_2$ such that $|C_1 \cap C_2| = k-1$ and focus on them.  (In other
words, we consider function $f_{C_1, C_2}$ instead of $f_k$.) In this
case, Lemma~\ref{lemma:zeroAlpha} allows us to change the domain of
the function.

Let us consider two voting situations $P$ and $Q$ such that
$\alpha_{C_1,C_2}(P) = \alpha_{C_1,C_2}(Q)$. Since $\alpha_{C_1,C_2}$
is a linear function, we have $\alpha_{C_1,C_2}(P - Q) = \zerovector$.
Thus, by Lemma~\ref{lemma:zeroAlpha}, we know that $C_1 =_{P - Q}
C_2$.  We can express $Q$ as $Q = P + (Q - P)$ and thus, by
consistency of $f_{C_1,C_2}$, we have that:
\[
  C_1 \succ_{P} C_2 \iff C_1 \succ_{Q} C_2.
\]
Consequently, to answer the question ``what is the relation between
committees $C_1$ and $C_2$ according to $f_{C_1, C_2}$ in voting
situation $P$?'' it suffices to know the value $\alpha_{C_1,C_2}(P)$.
This is exactly because for any two profiles, $P$ and $Q$, with the
same values of function $\alpha_{C_1,C_2}$ the result of comparison of
committees $C_1$ and $C_2$ according to $f_{C_1, C_2}$ is the same in
$P$ and $Q$.

In effect, we can restrict the domain of $f_{C_1, C_2}$ to an
$\left({m \choose k} - 1\right)$-dimensional space $D$:
\begin{align*}
  D = \left\{P \in \rationals^{m \choose k}: \sum_{I \in [m]_k}P[I] = 0 \right\} \textrm{.}
\end{align*}
We interpret elements of $D$ as the values of the committee
position-difference function $\alpha_{C_1,C_2}$ and, so, the condition
$\sum_{I \in [m]_k}P[I] = 0$ corresponds to the property of committee
position-difference functions given in
Equation~\eqref{sumI=0}. By the argument given prior to the definition
of $D$, we know that from the point of view of comparing committees
$C_1$ and $C_2$ using function $f_{C_1,C_2}$, the vector of values
$\alpha_{C_1,C_2}$ provides the same information as a voting situation
from which it is obtained. Thus, we can think of elements of $D$ as
corresponding to voting situations.

\paragraph{Separating Two Committees.}
We proceed by defining two sets, $D_1, D_2 \subseteq D$, such that:
\[
  D_1 = \{P \in D: C_1 \succ_P C_2\} 
  \quad \text{and} \quad 
  D_2 = \{P \in D: C_2 \succ_P C_1\}.
\]
That is, $D_1$ corresponds to situations where, according to $f_{C_1,
  C_2}$, committee $C_1$ is preferred over $C_2$, and $D_2$
corresponds to the situations where it is the other way round.
From consistency of $f_{C_1, C_2}$, it follows that $D_1$ and $D_2$
are $\rationals$-convex.

Let us consider the case
where $f_{C_1,C_2}$ is trivial, i.e., for each voting situation it
ranks $C_1$ and $C_2$ as equal. By neutrality, it follows that $f_k$
ranks equally each two committees $C'_1$ and $C'_2$, such that $|C'_1
\cap C'_2| = k-1$. This means that $f_k$ (for committees with
intersection $k-1$) can be expressed by means of the trivial committee
scoring function $\lambda \equiv 0$.
% (Note that we still consider $f_{C_1, C_2}$ for fixed committees
% $C_1$ and $C_2$. The fact that $\lambda \equiv 0$ describes function
% $f_k$, i.e., the fact that the function $f_k$ must also be trivial in
% such a case, will follow from the discussions in the subsequent
% subsection). 
 So let us assume that $f_{C_1, C_2}$ is
nontrivial and there is some voting situation where it does not rank
$C_1$ and $C_2$ equally. In this case one of the sets $D_1$ and $D_2$
is nonempty. From neutrality it follows that so is the other one. Now,
we move our analysis from $\rationals^{m \choose k}$ to $\reals^{m
  \choose k}$, by analyzing the closures of the sets $D_1$ and $D_2$.

\begin{lemma}\label{lemma:disjointSets}
  The sets $\inter(\overline{D_1})$ and $\overline{D_2}$ are disjoint,
  convex, and nonempty relative to $D$ (i.e., $\inter(\overline{D_1})
  \cap D \neq \emptyset$ and $\overline{D_2} \cap D \neq \emptyset$).
\end{lemma}
\begin{proof}
  This lemma follows from the results given by
  Young~\shortcite{you:j:scoring-functions} and
  Merlin~\shortcite{merlinAxiomatic}. However, in their cases the
  proofs are implicit in the text. We include an explicit proof for
  the sake of completeness.

  From Lemma~\ref{lemma:qconvex2}, it follows that the sets
  $\overline{D_1}$ and $\overline{D_2}$ are convex and, thus, the
  interior $\inter(\overline{D_1})$ is also convex.  Now, we prove
  that $\overline{D_1} \cup \overline{D_2} = \overline{D}$, a fact
  that will be useful in our further analysis. If this is not the
  case, then $\overline{D} - (\overline{D_1} \cup \overline{D_2})$ is
  open in $\overline{D}$. Thus, there exists a point $P$ and an
  $\left({m \choose k} - 1 \right)$-dimensional ball $\mathcal{B}$
  such that $P \in \mathcal{B} \subseteq \overline{D} -
  (\overline{D_1} \cup \overline{D_2})$. Naturally, $C_1 =_P
  C_2$. Thus, for some $S \in D_1$, there exists a (small) $x \in
  \rationals$, such that $Q = x \cdot S + (1 - x) \cdot P$ belongs to
  the ball $\mathcal{B}$. Since $Q$ belongs to $\mathcal{B}$, it must
  be the case that $C_1 =_Q C_2$.  However, by consistency of
  $f_{C_1,C_2}$, we have that $C_1 \succ_{Q} C_2$ and, so, we have $Q
  \in D_1$. This is a contradiction.

  Next, we show that the set $\inter(\overline{D_1})$ is nonempty,
  relatively to $D$. For the sake of contradiction, assume that
  $\inter(\overline{D_1}) \cap D = \emptyset$. Then, from neutrality,
  it follows that also $\inter(\overline{D_2}) \cap D =
  \emptyset$. Thus, $D_1$ and $D_2$ are nowhere dense in
  $D$,\footnote{A subset $A$ of a topological space $X$ is called
    nowhere dense (in $X$) if there is no neighborhood in $X$ on which
    $A$ is dense.} and so are $\overline{D_1}$, $\overline{D_2}$, and
  $\overline{D_1} \cup \overline{D_2} = \overline{D}$. Consequently,
  we get that $\overline{D}$ is nowhere dense in $D$, a contradiction
  with the density of $D$ in $\overline{D}$\footnote{A subset $A$ of a
    topological space $X$ is dense in $X$ if for every point $x$ in
    $X$, each neighborhood of $x$ contains at least one point from $A$
    (i.e., $A$ has non-empty intersection with every non-empty open
    subset of $X$).} (density of $D$ follows immediately from its
  definition).

  To see that $\overline{D_2}\cap D$ is nonempty, it suffices to note
  that $f_k$ is nontrivial (by assumptions just ahead of the statement
  of the lemma) and, so, $D_2$ is nonempty. Since $D_2$ is a subset of
  both $\overline{D_2}$ and $D$, we get that $\overline{D_2}\cap D
  \neq \emptyset$.

  % Before we show that the sets $\inter(\overline{D_1})$ and
  % $\overline{D_2}$ are disjoint, 

  Now we show that $\inter(\overline{D_1})$ and $D_2$ are
  disjoint. For the sake of contradiction, let us assume that there
  exists $P \in \overline{D}$, such that $P \in
  \inter(\overline{D_1})$ and $P \in D_2$. From
  Lemma~\ref{lemma:qconvex2}, we get that $\inter(\overline{D_1}) =
  \inter(\overline{\cvx(D_1)}) = \inter(\cvx(D_1))$. This means that
  $P \in \inter(\cvx(D_1)) \cap D_2$ and, so, $P \in \cvx(D_1) \cap
  D_2$. Since $P \in D_2$, we know that $P \in \rationals^{m \choose
    k}$.  By Lemma~\ref{lemma:qconvex1} we know that $D_1 =
  \rationals^{m \choose k} \cap \cvx(D_1)$. Thus, since $P \in
  \rationals^{m \choose k}$ and $P \in \cvx(D_1)$, we know that $P \in
  D_1$. All in all, it must be the case that $P \in D_1 \cap D_2$,
  which is a contradiction because $D_1 \cap D_2 = \emptyset$.

  Finally, for the sake of contradiction, let us assume that there
  exists $Q \in \overline{D}$, such that $Q \in \inter(\overline{D_1})$ and $Q
  \in \overline{D_2}$. Since $Q \in \overline{D_2}$, this means that
  every open set containing $Q$ must have nonempty intersection with
  $D_2$. Consequently, $\inter(\overline{D_1})$ has nonempty
  intersection with $D_2$, which---by the previous paragraph---gives a
  contradiction. This completes the proof of the lemma.
\end{proof}

\paragraph{Recovering the Scoring Function.}
We are finally ready to derive our committee scoring function.  From
the classic hyperplane separation theorem, it follows that there
exists a vector $\eta \in \reals^{m \choose k}$ such that (for $P \in
D$, by $\eta \cdot P$ we mean the dot product of $P$ and $\eta$, both
treated as $m \choose k$ dimensional vectors):
\begin{enumerate}
\item For each voting situation $P \in \overline{D_2}$ it holds that
  $\eta \cdot P \leq 0$ .
\item For each voting situation $P \in
  \inter(\overline{D_1})$ it holds that $\eta \cdot P > 0$.  
\end{enumerate}
We note that Lemma~\ref{lemma:disjointSets} allows us to directly
apply the hyperplane separation theorem as the sets
$\inter(\overline{D_1})$ and $\overline{D_2}$ are
disjoint.\footnote{This is different from
  Young's~\shortcite{you:j:scoring-functions} and
  Merlin's~\shortcite{merlinAxiomatic} approach, who operate on sets
  with disjoint interiors, but which do not have to be disjoint on
  their own.}

We now show that if $P \in D$ and $\eta \cdot P > 0$, then $P \in
D_1$. Since $\eta \cdot P > 0$, $P$ cannot belong to $D_2$, but it
might be the case that $C_1 =_{P} C_2$. For the sake of contradiction,
let us assume that this is the case.  We observe that there exists an
$\left({m \choose k} - 1 \right)$-dimensional ball $\mathcal{B}$ in
$D$ with $P \in \mathcal{B}$, such that for each $S \in \mathcal{B}$
we have $C_1 \succeq_{S} C_2$ (this is because $P$ does not belong to
$\overline{D_2}$). Let us now consider two cases.
\begin{description}
\item[Case 1.] If for each $S \in \mathcal{B}$ we have $C_1 =_{S}
  C_2$, then we proceed as follows. Let us take some $Q$ such that
  $C_1 \succ_{Q} C_2$.  There must exist some (possibly very small)
  $x$ such that $S = x \cdot Q + (1-x) \cdot P \in
  \mathcal{B}$. However, from consistency we would get that $C_1
  \succ_{S} C_2$, a contradiction.
\item[Case 2.] If there exists $Q \in \mathcal{B}$ such that $C_1
  \succ_{Q} C_2$, then we observe that there exists $0 < \epsilon < 1$
  such that $S = \frac{P - \epsilon Q}{1 - \epsilon} \in \mathcal{B}$.
  Since $S \in \mathcal{B}$, we have that $C_1 \succeq_S C_2$.
  Further, we have that $P = \epsilon Q + (1-\epsilon) S$.  By
  consistency of $f_{C_1,C_2}$ we get that $C_1 \succ_{P}
  C_2$. However, this is a contradiction.
\end{description}

Next, we show that if $\eta \cdot P < 0$, then $P \in D_2$. For the
sake of contradiction, let us assume that there is $P$ such that $\eta
\cdot P < 0$ but $C_1 \succeq_P C_2$. Then there exists such
$\epsilon$ that if $|Q - P| < \epsilon$ then $\eta \cdot Q < 0$ (and
so $Q \notin \inter(\overline{D_1})$). Thus there exists a ball
$\mathcal{B}$ in $D$ with $P \in \mathcal{B}$, such that $\mathcal{B}
\cap \inter(\overline{D_1}) = \emptyset$. Thus, $\mathcal{B} \cap D_1
= \emptyset$. We infer that some point $S$ in $\mathcal{B}$ could be
represented as a linear combination of $P$ and some point from
$D_1$. From consistency we would get that $C_1 \succ_S C_2$, a
contradiction.

\begin{remark}
  We have shown that for each $P \in D$, (a) $\eta \cdot P > 0$
  implies that $P \in D_1$ (and, so, $C_1 \succ_P C_2$), and (b) $\eta
  \cdot P < 0$ implies that $P \in D_2$ (and, so, $C_2 \succ_P
  C_1$). From symmetry, the same vector $\eta$ works for each pair of
  committees $C_1$ and $C_2$ such that $|C_1 \cap C_2| = k-1$.
\end{remark}

Now we will use continuity to prove that if $\eta \cdot P = 0$ then
$C_1 =_P C_2$.  For the sake of contradiction let us assume that this
is not the case, i.e., that there exists a voting situation $P \in D$
such that $\eta \cdot P = 0$ but $C_1 \neq_P C_2$.  Without loss of
generality, let us assume that $C_1 \succ_P C_2$. Let $Q$ be a voting
situation such that $\eta \cdot Q < 0$ and so $C_2 \succ_{Q} C_1$.
For each $x$ it holds that $\eta \cdot (xP + Q) < 0$ and so $C_2
\succ_{xP + Q} C_1$. However, this contradicts continuity of $f_k$.
Thus, for every $P \in D$, if $\eta \cdot P = 0$ then $C_1 =_P C_2$.

From vector $\eta$, we retrieve a committee scoring function
$\lambda$. For each committee position $I \in [m_k]$ we set
$\lambda(I) = \eta[I]$. Now, we can see that for each two committees
$C_1, C_2$, and for each voting situation $P \in \rationals^{m!}$ it
holds that (see the comment below for an
explanation of what $Q$ is):
\begin{align*}
  \scorefull{\posf}{C_1}{P} - \scorefull{\posf}{C_2}{P}
  &= \sum_{I \in [m]_k} \left(\lambda(I) \cdot \posweight_I(C_1,P) - \lambda(I) \cdot \posweight_I(C_2,P)\right) \\
  &= \sum_{I \in [m]_k} \lambda(I) \cdot \alpha_{C_1, C_2}(P)[I] =
  \sum_{I \in [m]_k} \eta[I] \cdot \alpha_{C_1, C_2}(P)[I] = \eta
  \cdot Q \textrm{,}
\end{align*}
where $Q \in D$ is the representation of $P$ in the space $D$ (i.e.,
$Q$ is the vector of values of the committee position-difference
function $\alpha_{C_1, C_2}$ for profile $P$). From the above inequality we see
that $\scorefull{\posf}{C_1}{P} > \scorefull{\posf}{C_2}{P}$ implies
that $C_1 \succ_P C_2$ and that $\scorefull{\posf}{C_1}{P} =
\scorefull{\posf}{C_2}{P}$ implies that $C_1 =_P C_2$. From neutrality
we get that the same committee scoring function $\lambda$ works for
every two committees $C_1'$ and $C_2'$ with $|C_1' \cap C_2'| = k-1$

There is one more issue we need to deal with.  So far, we gave no
argument as to why $\lambda$ should satisfy the dominance property of
committee scoring functions (i.e., that if $I$ and $J$ are two
committee positions such that $I$ dominates $J$, then $\posf(I) \geq
\posf(J)$). However, to get this property it suffices to assume the
committee dominance axiom for $f_k$.

% Naturally, the reasoning from this section applies to each value of
% the committee size $k$, $1 \leq k \leq m$.  Consequently, it allows
% to define committee scoring functions $\{\eta_i\}_{i \in [m]}$ on
% all subsets of $m$.
%It is clear that committee scoring rules are symmetric,
%consistent, continuous, and have the committee
%dominance property. 
Summarizing our discussion from this section, we get our main result,
Theorem~\ref{thm:theMainTheorem}, for the committees $C_1$ and $C_2$,
with $|C_1\cap C_2| = k-1$. We continue our analysis in the next
section.

\subsection{Putting Everything Together: Comparing Arbitrary
  Committees}\label{sec:specialCase_base}

In this section we conclude the proof of
Theorem~\ref{thm:theMainTheorem} by extending the reasoning from the
previous section to apply to every two committees $C_1$ and $C_2$
irrespective of the size of their intersection.

\paragraph{Setting Up the Proof.}

Let $f_k$ be a $k$-winner election rule that is symmetric,
consistent, continuous, and has the committee
dominance property. Let $\lambda$ be the scoring function derived for
this $f_k$ as described at the end of the previous section.  We know
that for each two committees $C_1$ and $C_2$ such that $|C_1 \cap C_2|
= k-1$ and each voting situation $P \in \rationals^{m!}$ it holds that
$\scorefull{\posf}{C_1}{P} > \scorefull{\posf}{C_2}{P}$ if and only if
$C_1 \succ_P C_2$, and $\scorefull{\posf}{C_1}{P} =
\scorefull{\posf}{C_2}{P}$ if and only if $C_1 =_P C_2$.  We will show
that the same holds for all committees $C_1$ and $C_2$, irrespective
of the size of their intersection.  We will show this by induction
over $k - |C_1 \cap C_2|$.  

Let us fix some value $k' < k-1$ and let us assume that $\lambda$ can
be used to distinguish whether some committee $C_1$ is preferred over
some committee $C_2$ whenever $|C_1 \cap C_2| > k'$. We will show that
the same $\lambda$ can be used to distinguish whether committee $C_1$
is preferred over committee $C_2$ when $|C_1 \cap C_2| = k'$.

Let $C_1$ and $C_2$ be two arbitrary committees such that $|C_1 \cap
C_2| = k'$. Let us rename the candidates so that $C_1 \setminus C_2 =
\{c_1, \ldots, c_{k-k'}\}$, $C_1 \cap C_2 = \{c_{k-k'+1}, \ldots,
c_{k}\}$ and $C_2 \setminus C_1 = \{c_{k+1}, \ldots, c_{2k-k'}\}$.

\paragraph{The Case Where $\bf k - k'$ Is Even.}
 
If $k - k'$ is even, we consider the following two cases:

\begin{description}
\item[Case 1:] There exists a vector of $2k-k'$ positions $\langle
  p_1, \ldots, p_{2k-k'} \rangle$ such that:
  \begin{align}\label{eq:existsnontransitive_C1}
    \lambda(\{p_1, \ldots, p_k\}) + \lambda(\{p_{k-k'+1}, \ldots,
    p_{2k-k'}\}) \neq 2\lambda(\{p_{ \frac{k-k'}{2}+1}, \ldots, p_{
      \frac{k-k'}{2} + k }\}) \textrm{.}
  \end{align}
  Let us consider the committee $C_3 = \{c_{ \frac{k-k'}{2}+1}, \ldots, c_{ \frac{k-k'}{2} +
    k}\}$. We consider the vector space of voting
  situations $P \in \rationals^{m!}$ such that $C_1 =_P C_3$ and $C_3
  =_P C_2$ (the fact that this is a vector space follows from the
  inductive assumption; $|C_1 \cap C_3| = |C_2 \cap C_3| > k'$). 
  The conditions $C_1 =_P C_3$ and $C_3 =_P C_2$ are not contradictory
  (consider the profile in which each vote is cast exactly once---in such
  profile all size-$k$ committees are equivalent with respect to $f_k$).
  This space has dimension either $m! - 2$ or $m! - 1$. This is so, because
  each of the conditions $C_1 =_P C_3$ and $C_2 =_P C_3$ boils down to
  a single linear equation. If these equations are independent then
  the dimension is $m!-2$. Otherwise, it is $m!-1$.  By transitivity
  of $f_k$ we get that in each voting situation $P$ from this space it
  holds that $C_1 =_P C_2$ and that the committee score of $C_1$
  (according to $\lambda$) is equal to the committee score of
  $C_2$. Let $B$ be a base of this space. Further, let $v$ be a vote
  where each candidate $c_i$, $i \in \{1, \ldots, 2k-k'\}$, stands on
  position $p_i$ (recall Equation~\eqref{eq:existsnontransitive_C1}
  above), and let $v'$ be an identical vote except that candidates
  from $C_1 \cup C_2$ are listed in the reverse order (i.e., $c_1$ is
  on position $p_{2k-k'}$, $c_2$ is on position $p_{2k-k'-1}$ and so
  on). Let $S_b$ be a voting situation that consists of $v$ and
  $v'$. The positions of $C_1$ and $C_3$ in $v$ are: 
  \begin{align*}
  \pos_v(C_1) = \{p_1, \ldots, p_k\} 
  \quad & \text{and} \quad
  \pos_v(C_3) = \{p_{ \frac{k-k'}{2} +1}, \ldots, p_{\frac{k-k'}{2} + k }\}  \\
  \intertext{The positions of $C_1$ and $C_3$ in $v'$ are:}
  \pos_{v'}(C_1) = \{p_{k-k'+1}, \ldots, p_{2k-k'}\}
  \quad & \text{and} \quad
  \pos_{v'}(C_3) = \{p_{\frac{k-k'}{2} +1}, \ldots, p_{ \frac{k-k'}{2} + k }\}  
  \end{align*}
  %
  % In $v$, $C_1$ stands on positions $\{p_1, \ldots, p_k\}$ and $C_3$
  % on positions $\{p_{ \frac{k-k'}{2} +1}, \ldots, p_{ \frac{k-k'}{2}
  % + k }\}$. In $v'$, $C_3$ stands on positions $\{p_{ \frac{k-k'}{2}
  % +1}, \ldots, p_{ \frac{k-k'}{2} + k }\}$ and $C_1$ on positions
  % $\{p_{k-k'+1}, \ldots, p_{2k-k'}\}$.
  %
  Consequently, according to Equation~\eqref{eq:existsnontransitive_C1}, in voting
  situation $S_b$ the committee score of $C_1$ is not equal to that of
  $C_3$. By the inductive assumption, it must be the case that $C_1
  \neq_{S_b} C_3$. This means that the voting situations in $B \cup
  \{S_b\}$ are linearly independent.

  We now show that $C_1 =_{S_b} C_2$.  Consider a permutation $\sigma$
  (over the candidate set) that swaps $c_1$ with $c_{2k-k'}$, $c_2$
  with $c_{2k-k'-1}$, and so on. We note that $\sigma(C_1) = C_2$,
  $\sigma(C_2) = C_1$, and $S_b = \sigma(S_b)$. Thus, by symmetry of
  $f_k$, it must be the case that $C_1 =_{S_b} C_2$.  Further, the
  committee scores of $C_1$ and $C_2$ are equal in $S_b$.

  Altogether, the base $B \cup \{S_b\}$ defines an $(m! -
  1)$-dimensional space of voting situations $P$ such that $C_1 =_P
  C_2$ and the committee scores of $C_1$ and $C_2$ are equal.  From
  Corollary~\ref{cor:hyperplane} we know that the set of voting
  situations $P$ such that $C_1 =_P C_2$ forms a vector space of
  dimension $m! -1$. As a result, we get that for each voting
  situation $P$ the condition $C_1 =_P C_2$ is equivalent to the
  condition that $C_1$ has the same committee score as $C_2$ according
  to $\lambda$.

  The fact that $C_1 \succ_S C_2$ whenever the committee score of
  $C_1$ is greater than that of $C_2$
  follows from Lemma~\ref{thm:nontransitive2}.

\item[Case 2:] For each vector of $2k-k'$ positions $\langle p_1,
  \ldots, p_{2k-k'} \rangle$ it holds that (note that the condition
  below is a negation of the condition from Case~1):
  \begin{align*} %\label{eq:existsnontransitive2}
      \lambda(\{p_1, \ldots, p_k\}) - &\lambda(\{p_{ \frac{k-k'}{2}+1}, \ldots, p_{ \frac{k-k'}{2}  + k }\}) = \\
      &\lambda(\{p_{ \frac{k-k'}{2}+1}, \ldots, p_{ \frac{k-k'}{2} + k
      }\}) - \lambda(\{p_{k-k'+1}, \ldots, p_{2k-k'}\}) \textrm{.}
  \end{align*}
  As before, let $C_3 = \{c_{ \frac{k-k'}{2}+1}, \ldots, c_{
    \frac{k-k'}{2} + k }\}$.  Since the above equality must hold for
  each vector of $2k-k'$ positions, we see that if the committee score
  of $C_1$ is equal to the committee score of $C_3$, then the
  committee score of $C_3$ is equal to the committee score of
  $C_2$. Consequently, by the inductive assumption, we get that $C_1
  =_P C_3$ implies that $C_3 =_P C_2$. Thus, by $f_k$'s transitivity,
  we get that for each voting situation $P$, the condition $C_1 =_P
  C_3$ implies that $C_1 =_P C_2$. As a consequence of this reasoning,
  there exists an $(m! - 1)$-dimensional space of voting situations
  $P$ such that $C_1 =_P C_2$ and such that $C_1$ has the same
  committee score as $C_2$. Similarly as in Case~1, we conclude that
  for each voting situation $P$ the condition $C_1 =_P C_2$ is
  equivalent to the condition that $C_1$ has the same committee score
  as $C_2$ according to $\lambda$, and that it holds that $C_1 \succ_P
  C_2$ whenever the committee score of $C_1$ is greater than
  that of $C_2$ (by Lemma~\ref{thm:nontransitive2}).
\end{description}

\paragraph{The Case Where $\bf k - k' \geq 3$ and $\bf k - k'$ is Odd.}
Similarly as before we consider two cases:

\begin{description}
\item[Case 1:] There exists a vector of $2k-k'$ positions $\langle
  p_1, \ldots, p_{2k-k'} \rangle$ and a number $x \in \{1, \ldots
  k-k'\}$ such that:
  \begin{align*} %\label{eq:existsnontransitive}
    \lambda(\{p_1, \ldots, p_k\}) \;+\; &\lambda(\{p_{k-k'+1}, \ldots, p_{2k-k'}\}) \neq \\
    & \lambda(\{p_x, \ldots, p_{k + x - 1}\}) + \lambda(\{p_{k-k'+ 2 -
      x}, \ldots, p_{2k-k'+1-x}\}) \textrm{.}
  \end{align*}
  In this case we can repeat the reasoning from Case 1 from the
  previous subsection (it suffices to take $C_3 = \{c_x, \ldots,
  c_{k+x-1}\}$).

\item[Case 2:] For each vector of $2k-k'$ positions $\langle p_1,
  \ldots, p_{2k-k'} \rangle$ and each number $x \in \{1, \ldots
  k-k'\}$ it holds that:
  \begin{align*} %\label{eq:existsnontransitive}
    \lambda(\{p_1, \ldots, p_k\}) \;+\; &\lambda(\{p_{k-k'+1}, \ldots, p_{2k-k'}\}) = \\
    &\lambda(\{p_x, \ldots, p_{k + x - 1}\}) + \lambda(\{p_{k-k'+ 2 -
      x}, \ldots, p_{2k-k'+1-x}\}) \textrm{.}
  \end{align*}
  The above inequality for $x = \lfloor \frac{k-k'}{2} \rfloor$ and
  for $x = \lfloor \frac{k-k'}{2} \rfloor + 1$ gives, respectively
  (note that $k - k' - \lfloor \frac{k-k'}{2} \rfloor = \lceil
  \frac{k-k'}{2} \rceil$):
  \begin{align*}
    \lambda(\{p_1, \ldots, p_k\}) \;+\; & \lambda(\{p_{k-k'+1}, \ldots, p_{2k-k'}\}) =\\
    &\lambda(\{p_{\lfloor \frac{k-k'}{2} \rfloor} \ldots, p_{k +
      \lfloor \frac{k-k'}{2} \rfloor - 1}\}) + \lambda(\{p_{\lceil
      \frac{k-k'}{2} \rceil + 2} \ldots, p_{k + \lceil \frac{k-k'}{2}
      \rceil +1}\}) \textrm{,}
  \end{align*}
  and:
  \begin{align*}
    \lambda(\{p_1, \ldots, p_k\}) \;+\; & \lambda(\{p_{k-k'+1}, \ldots, p_{2k-k'}\}) =\\
    &\lambda(\{p_{\lfloor \frac{k-k'}{2} \rfloor + 1} \ldots, p_{k +
      \lfloor \frac{k-k'}{2} \rfloor }\}) + \lambda(\{p_{\lceil
      \frac{k-k'}{2} \rceil + 1} \ldots, p_{k + \lceil \frac{k-k'}{2}
      \rceil }\}) \textrm{.}
  \end{align*}
  Together, these two equalities give that:
  \begin{align*}
    \lambda(\{p_{\lfloor \frac{k-k'}{2} \rfloor} \ldots, p_{k + \lfloor \frac{k-k'}{2} \rfloor - 1}\}) \; + \; & \lambda(\{p_{\lceil \frac{k-k'}{2} \rceil + 2} \ldots, p_{k + \lceil \frac{k-k'}{2} \rceil +1}\}) = \\
    \lambda(\{p_{\lfloor \frac{k-k'}{2} \rfloor + 1} \ldots, p_{k +
      \lfloor \frac{k-k'}{2} \rfloor }\}) \: + \: & \lambda(\{p_{\lceil
      \frac{k-k'}{2} \rceil + 1} \ldots, p_{k + \lceil \frac{k-k'}{2}
      \rceil }\}) \textrm{.}
  \end{align*}
  Since the above equality holds for each vector of $2k-k'$ positions,
  after renaming the positions, we get that for each set of $k+3$
  positions $\langle q_1, \ldots, q_{k+3} \rangle$ it holds that:
  \begin{align*}
    \lambda(\{q_1, \ldots, q_k\}) + \lambda(\{q_4, \ldots, q_{k+3}\})
    = \lambda(\{q_2, \ldots, q_{k+1}\}) + \lambda(\{q_3, \ldots,
    q_{k+2}\}) \textrm{.}
  \end{align*}
  After reformulation we get:
  \begin{align}\label{eq:cyclicMove}
    \lambda(\{q_1, \ldots, q_k\}) - \lambda(\{q_2, \ldots, q_{k+1}\})
    = \lambda(\{q_3, \ldots, q_{k+2}\}) - \lambda(\{q_4, \ldots,
    q_{k+3}\}) \textrm{.}
  \end{align}

  If $k$ is odd, we obtain the following series of equalities (the
  consecutive equalities, except for the last one, are consequences of
  applying Equation~\eqref{eq:cyclicMove} to the cyclic shifts
  of the list $\langle q_1, q_2, \ldots, q_{k+3} \rangle$; the last
  equality breaks the pattern and is a consequence of applying
  Equation~\eqref{eq:cyclicMove} to the list $\langle q_{k+2}, q_{k+3},
  q_1, q_2, \ldots, q_{k-1}, q_{k+1}, q_k \rangle$):
  \begin{align*}
    \lambda(\{q_1, \ldots, q_k\}) \;-\; & \lambda(\{q_2, \ldots, q_{k+1}\}) 
    = \lambda(\{q_3, \ldots, q_{k+2}\}) - \lambda(\{q_4, \ldots, q_{k+3}\})  \\
    &= \lambda(\{q_5, \ldots, q_{k+3}, q_1\}) - \lambda(\{q_6, \ldots, q_{k+3}, q_1, q_2\}) \\
    &= \lambda(\{q_7, \ldots, q_{k+3}, q_1, q_2, q_3\}) - \lambda(\{q_8, \ldots, q_{k+3}, q_1, q_2, q_3, q_4\}) \\
    &\quad \quad \vdots \\
    &= \lambda(\{q_{k+2}, q_{k+3}, q_1, \ldots, q_{k-2}\}) - \lambda(\{q_{k+3}, \ldots, q_{1}, q_{k-1}\}) \\
    &= \lambda(\{q_1, \ldots, q_{k-1}, q_{k+1}\}) - \lambda(\{q_{2},
    \ldots, q_{k+1}\}) \textrm{,}
  \end{align*}
  In consequence, it must be the case that $\lambda(\{q_1, \ldots,
  q_k\}) = \lambda(\{q_1, \ldots, q_{k-1}, q_{k+1}\})$.  Thus, by
  transitivity, we get that $\lambda$ is a constant function (in
  essence, what we have shown is that we can replace positions in the set of $k$ positions, one
  by one, without changing the value of the committee scoring
  function).  Let $C_3 = \{c_2, \ldots, c_{k+1}\}$. Since $\lambda$ is
  a constant function, then by the inductive assumption we have that
  for every voting situation $P$ it holds that $C_1 =_P C_3$ and $C_3
  =_P C_2$. By transitivity we get that for each voting situation $P$
  it holds that $C_1 =_P C_2$. Thus our trivial scoring function works
  correctly on $C_1$ and $C_2$.\medskip

  Let us now assume that $k$ is even. Now we obtain the following
  series of equalities (in this case all the consecutive equalities
  are consequences of applying Equation~\eqref{eq:cyclicMove} to the cyclic shifts of
  the sequence $\langle q_1, q_2, \ldots, q_{k+3} \rangle$):
  \begin{align*}
    \lambda(\{q_1, \ldots, q_k\}) \;-\; & \lambda(\{q_2, \ldots, q_{k+1}\})  
    = \lambda(\{q_3, \ldots, q_{k+2}\}) - \lambda(\{q_4, \ldots, q_{k+3}\})  \\
    &= \lambda(\{q_5, \ldots, q_{k+3}, q_1\}) - \lambda(\{q_6, \ldots, q_{k+3}, q_1, q_2\}) \\
    &= \lambda(\{q_7, \ldots, q_{k+3}, q_1, q_2, q_3\}) - \lambda(\{q_8, \ldots, q_{k+3}, q_1, q_2, q_3, q_4\}) \\
    & \quad \quad \vdots \\
    &= \lambda(\{q_{k+3}, q_1, \ldots, q_{k-1}\}) - \lambda(\{q_1, \ldots, q_k\}) \\
    &= \lambda(\{q_2, \ldots, q_{k+1}\}) - \lambda(\{q_3, \ldots,
    q_{k+2}\}) \textrm{,}
  \end{align*}
  In consequence, it is the case that:
  \begin{align*}
    \lambda(\{q_1, \ldots, q_k\}) - \lambda(\{q_2, \ldots, q_{k+1}\}) = 
    \lambda(\{q_2, \ldots, q_{k+1}\}) - \lambda(\{q_3, \ldots, q_{k+2}\}) \textrm{,}
  \end{align*}
  and this holds for every sequence $\langle q_1, \ldots, q_{k+2}
  \rangle$ of positions.  Thus, we get that for each voting situation
  in which $\{c_{1}, \ldots, c_{k}\}$ is equivalent to $\{c_{2},
  \ldots, c_{k+1}\}$, it also holds that $\{c_{2}, \ldots, c_{k+1}\}$
  is equivalent to $\{c_{3}, \ldots, c_{k+2}\}$, it also holds that
  $\{c_{3}, \ldots, c_{k+2}\}$ is equivalent to $\{c_{4}, \ldots,
  c_{k+3}\}$, etc.  Let $C_3 = \{c_{2}, \ldots, c_{k+1}\}$. From the
  preceding reasoning we have that for each voting situation $P$ the
  fact that it holds that $C_1 =_P C_3$ implies that $C_1 =_P C_2$. We
  conclude the proof in the same way as in the case of even $k-k'$
  (Case 2).  Specifically, we conclude that there exists an $(m! -
  1)$-dimensional space of voting situations $P$ such that $C_1 =_P
  C_2$ and such that $C_1$ has the same committee score as $C_2$. This
  means that for each voting situation $P$ the condition $C_1 =_P C_2$
  is equivalent to the condition that $C_1$ has the same committee
  score as $C_2$ according to $\lambda$, and that it holds that $C_1
  \succ_P C_2$ whenever the committee score of $C_1$ is greater than
  that of $C_2$ (by Lemma~\ref{thm:nontransitive2}).
\end{description}

\paragraph{The End.}
We have shown that if a $k$-winner rule is symmetric,
consistent, continuous, and has the
committee-dominance property, then it is a committee scoring rule. On
the other hand, committee scoring rules satisfy all these conditions.
This completes our proof of Theorem~\ref{thm:theMainTheorem}.

\section{Conclusions}\label{sec:conclusions}

We have provided an axiomatic characterization of committee scoring
rules, a new class of multiwinner voting rules recently introduced by
Elkind et al.~\shortcite{elk-fal-sko-sli:c:multiwinner-rules}.
Committee scoring rules form a remarkably general class of multiwinner
systems that consists of many nontrivial rules with a variety of
applications.  Thus, our characterization constitutes a fundamental
framework for further axiomatic studies of this fascinating class and
makes an important step towards their understanding.  We mention that
various properties of committee scoring rules, and the internal
structure of the class, were already studied by Elkind et
al.~\cite{elk-fal-sko-sli:c:multiwinner-rules} and Faliszewski et
al.~\cite{fal-sko-sli-tal:c:top-k-counting,fal-sko-sli-tal:c:classification}.
However, they mostly focused on specific rules and on subclasses of
the whole class, while this work distinguishes the class of committee
scoring rules among the universe of multiwinner voting rules.

Our Theorem~\ref{thm:theMainTheorem} required developing a set of useful tools and new
concepts, such as decision rules.  We believe that they are an
interesting notion that deserves further study.

\bibliographystyle{plain}
\bibliography{main}

\begin{thebibliography}{10}

\bibitem{arrow1963}
K.~J. Arrow.
\newblock {\em Social Choice and Individual Values}.
\newblock Wiley, New York, 2nd edition, 1963.

\bibitem{arrow2010handbook}
K.~J. Arrow, A.~Sen, and K.~Suzumura, editors.
\newblock {\em Handbook of Social Choice \& Welfare}, volume~2.
\newblock Elsevier, 2010.

\bibitem{arrow2002handbook}
Kenneth~J Arrow, Amartya Sen, and Kotaro Suzumura, editors.
\newblock {\em Handbook of social choice and welfare}, volume~1.
\newblock Elsevier, 2002.

\bibitem{Asp13}
B.~Asplach.
\newblock Johnson graphs are hamilton-connected.
\newblock {\em Ars Mathematica Contemporanea}, 6(1):21--23, 2013.

\bibitem{justifiedRepresenattion}
H.~Aziz, M.~Brill, V.~Conitzer, E.~Elkind, R.~Freeman, , and T.~Walsh.
\newblock Justified representation in approval-based committee voting.
\newblock In {\em Proceedings of the 29th Conference on Artificial Intelligence
  (AAAI-2015)}, 2015.

\bibitem{azi-gas-gud-mac-mat-wal:c:multiwinner-approval}
H.~Aziz, S.~Gaspers, J.~Gudmundsson, S.~Mackenzie, N.~Mattei, and T.~Walsh.
\newblock Computational aspects of multi-winner approval voting.
\newblock In {\em Proceedings of the 14th International Conference on
  Autonomous Agents and Multiagent Systems}, May 2015.

\bibitem{bar-coe:j:non-controversial-k-names}
S.~Barber\`a and D.~Coelho.
\newblock How to choose a non-controversial list with $k$ names.
\newblock {\em Social Choice and Welfare}, 31(1):79--96, 2008.

\bibitem{Barbera198249}
S.~Barberá and B.~Dutta.
\newblock Implementability via protective equilibria.
\newblock {\em Journal of Mathematical Economics}, 10(1):49--65, 1982.

\bibitem{fullyProportionalRepr}
N.~Betzler, A.~Slinko, and J.~Uhlmann.
\newblock On the computation of fully proportional representation.
\newblock {\em Journal of Artificial Intelligence Research}, 47:475--519, 2013.

\bibitem{Bran13a}
F.~Brandl, F.~Brandt, and H.~G. Seedig.
\newblock Consistent probabilistic social choice.
\newblock {\em Econometrica}, 84(5):1839--1880, 2016.

\bibitem{BrandtEtAlChapter}
Felix Brandt, Markus Brill, and Paul Harrenstein.
\newblock Tournament solutions.
\newblock In F.~Brandt, V.~Conitzer, U.~Endriss, J.~Lang, and A.~D. Procaccia,
  editors, {\em Handbook of Computational Social Choice}, chapter~3. Cambridge
  University Press, 2016.

\bibitem{ccElection}
B.~Chamberlin and P.~C{ourant}.
\newblock Representative deliberations and representative decisions:
  {P}roportional representation and the borda rule.
\newblock {\em American Political Science Review}, 77(3):718--733, 1983.

\bibitem{che-sha:j:scoring-rules}
P.~Chebotarev and E.~Shamis.
\newblock Characterizations of scoring methods for preference aggregation.
\newblock {\em Annals of Operations Research}, 80:299--332, 1998.

\bibitem{Ching1996298}
S.~Ching.
\newblock A simple characterization of plurality rule.
\newblock {\em Journal of Economic Theory}, 71(1):298--302, 1996.

\bibitem{cop:m:copeland}
A.~Copeland.
\newblock A ``reasonable'' social welfare function.
\newblock Mimeographed notes from a Seminar on Applications of Mathematics to
  the Social Sciences, University of Michigan, 1951.

\bibitem{deb:j:k-borda}
B.~Debord.
\newblock An axiomatic characterization of {B}orda's $k$-choice function.
\newblock {\em Social Choice and Welfare}, 9(4):337--343, 1992.

\bibitem{elk-fal-sko-sli:c:multiwinner-rules}
E.~Elkind, P.~Faliszewski, P.~Skowron, and A.~Slinko.
\newblock Properties of multiwinner voting rules.
\newblock In {\em Proceedings of the 13th International Conference on
  Autonomous Agents and Multiagent Systems (AAMAS-2014)}, May 2014.
\newblock Also presented in FIT-2013.

\bibitem{fal-saw-sch-smo:c:multiwinner-genetic-algorithms}
P.~Faliszewski, J.~Sawicki, R.~Schaefer, and M.~Smolka.
\newblock Multiwinner voting in genetic algorithms for solving ill­posed
  global optimization problems.
\newblock In {\em Proceedings of the 19th International Conference on the
  Applications of Evolutionary Computation}, 2016.
\newblock To appear.

\bibitem{fal-sko-sli-tal:c:classification}
P.~Faliszewski, P.~Skowron, A.~Slinko, and N.~Talmon.
\newblock Committee scoring rules: Axiomatic classification and hierarchy.
\newblock In {\em Proceedings of the 25th International Joint Conference on
  Artificial Intelligence (IJCAI-2016)}, 2016.

\bibitem{fal-sko-sli-tal:c:top-k-counting}
P.~Faliszewski, P.~Skowron, A.~Slinko, and N.~Talmon.
\newblock Multiwinner analogues of the plurality rule: {Axiomatic} and
  algorithmic views.
\newblock In {\em Proceedings of the 30th Conference on Artificial Intelligence
  (AAAI-2016)}, 2016.

\bibitem{fel-mao:j:norms}
D.S. Felsenthal and Z.~Maoz.
\newblock Normative properties of four single-stage multi-winner electoral
  procedures.
\newblock {\em Behavioral Science}, 37:109--127, 1992.

\bibitem{fishburn73SocChoice}
P.~Fishburn.
\newblock {\em The Theory of Social Choice,}.
\newblock Princeton University Press, 1973.

\bibitem{fishburn77socChooiceFunctions}
P.~Fishburn.
\newblock Condorcet social choice functions.
\newblock {\em SIAM Journal on Applied Mathematics}, 33(3):469--489, 1977.

\bibitem{fis:j:majority-committees}
P.~Fishburn.
\newblock Majority committees.
\newblock {\em Journal of Economic Theory}, 25(2):255--268, 1981.

\bibitem{fishburnBorda}
P.~Fishburn and W.~Gehrlein.
\newblock Borda's rule, positional voting, and condorcet's simple majority
  principle.
\newblock {\em Public Choice}, 28(1):79--88.

\bibitem{Fish84a}
P.~C. Fishburn.
\newblock Probabilistic social choice based on simple voting comparisons.
\newblock {\em Review of Economic Studies}, 51(4):683--692, 1984.

\bibitem{conf/aaai/FreemanBC14}
R.~Freeman, M.~Brill, and V.~Conitzer.
\newblock On the axiomatic characterization of runoff voting rules.
\newblock In {\em Proceedings of the 28th Conference on Artificial Intelligence
  (AAAI-2014)}, pages 675--681, 2014.

\bibitem{gardenfors73:scoring-rules}
P.~G{\"a}rdenfors.
\newblock Positionalist voting functions.
\newblock {\em Theory and Decision}, 4(1):1--24, 1973.

\bibitem{gib:j:polsci:manipulation}
A.~Gibbard.
\newblock Manipulation of voting schemes.
\newblock {\em Econometrica}, 41(4):587--601, 1973.

\bibitem{10.2307/1911681}
A.~Gibbard.
\newblock Manipulation of schemes that mix voting with chance.
\newblock {\em Econometrica}, 45(3):665--681, 1977.

\bibitem{gilboa2002utility}
Itzhak Gilboa, David Schmeidler, and Peter~P Wakker.
\newblock Utility in case-based decision theory.
\newblock {\em Journal of Economic Theory}, 105(2):483--502, 2002.

\bibitem{hansson76}
B.~Hansson and H.~Sahlquist.
\newblock A proof technique for social choice with variable electorate.
\newblock {\em Journal of Economic Theory}, 13:193--200, 1976.

\bibitem{kay-san:j:condorcet-winners}
B.~Kaymak and R.~Sanver.
\newblock Sets of alternatives as condorcet winners.
\newblock {\em Social Choice and Welfare}, 20(3):477--494, 2003.

\bibitem{kendall1938measure}
M.~G. Kendall.
\newblock A new measure of rank correlation.
\newblock {\em Biometrika}, 30(1/2):81--93, 1938.

\bibitem{kil-handbook}
D.~Kilgour.
\newblock Approval balloting for multi-winner elections.
\newblock In J.~Laslier and R.~Sanver, editors, {\em Handbook on Approval
  Voting}, pages 105--124. Springer, 2010.

\bibitem{kil-mar:j:minimax-approval}
M.~Kilgour and E.~Marshall.
\newblock Approval balloting for fixed-size committees.
\newblock In {\em Electoral Systems, Studies in Choice and Welfare}, volume~12,
  pages 305--326, 2012.

\bibitem{las:b:tournaments}
J.~Laslier.
\newblock {\em Tournament Solutions and Majority Voting}.
\newblock Springer, 1997.

\bibitem{budgetSocialChoice}
T.~Lu and C.~Boutilier.
\newblock Budgeted social choice: {F}rom consensus to personalized decision
  making.
\newblock In {\em Proceedings of the 22nd International Joint Conference on
  Artificial Intelligence (IJCAI-2011)}, pages 280--286, 2011.

\bibitem{bou-lu:c:value-directed}
T.~Lu and C.~Boutilier.
\newblock Value directed compression of large-scale assignment problems.
\newblock In {\em Proceedings of the 29th Conference on Artificial Intelligence
  (AAAI-2015)}, pages 1182--1190, 2015.

\bibitem{mayAxiomatic1952}
K.~O. May.
\newblock A set of independent necessary and sufficient conditions for simple
  majority decision.
\newblock {\em Econometrica}, 20(4):680--684, 1952.

\bibitem{merlinAxiomatic}
V.~Merlin.
\newblock The axiomatic characterization of majority voting and scoring rules.
\newblock {\em Mathematical Social Sciences}, 41(161):87--109, 2003.

\bibitem{nitzan81BordaChar}
S.~Nitzan and A.~Rubinstein.
\newblock A further characterization of borda ranking method.
\newblock {\em Public Choice}, 36(1):153--158, 1981.

\bibitem{fishburn78Approval}
Peter~C. P.~Fishburn.
\newblock {Axioms for approval voting: Direct proof}.
\newblock {\em Journal of Economic Theory}, 19(1):180--185, 1978.

\bibitem{pattanaik2002positional}
P.~K Pattanaik.
\newblock Positional rules of collective decision-making.
\newblock In K.~J. Arrow, A.~Sen, and K.~Suzumura, editors, {\em Handbook of
  social choice and welfare}, volume~1, pages 361--394. Elsevier, 2002.

\bibitem{per:pav-dhondt}
T.~Pereira.
\newblock Proportional approval method using squared loads, approval removal
  and coin-flip approval transformation (pamsac) - a new system of proportional
  representation using approval voting.
\newblock Technical Report arXiv:1602.05248v2 [cs.GT], arXiv.org, March 2016.

\bibitem{complexityProportionalRepr}
A.~Procaccia, J.~Rosenschein, and A.~Zohar.
\newblock On the complexity of achieving proportional representation.
\newblock {\em Social Choice and Welfare}, 30(3):353--362, April 2008.

\bibitem{Puke14a}
F.~Pukelsheim.
\newblock {\em Proportional Representation: Apportionment Methods and Their
  Applications}.
\newblock Springer, 2014.

\bibitem{rat:j:condorcet-inconsistencies}
T.~Ratliff.
\newblock Some startling inconsistencies when electing committees.
\newblock {\em Social Choice and Welfare}, 21(3):433--454, 2003.

\bibitem{RichelsonPlurality}
J.~T. Richelson.
\newblock A characterization result for the plurality rule.
\newblock {\em Journal of Economic Theory}, 19(2):548--550, 1978.

\bibitem{rubinstein80Tournament}
A.~Rubinstein.
\newblock Ranking the participants in a tournament.
\newblock {\em SIAM Journal on Applied Mathematics}, 38(1):108--111, 1980.

\bibitem{sat:j:polsci:manipulation}
M.~Satterthwaite.
\newblock Strategy-proofness and {Arrow's} conditions: Existence and
  correspondence theorems for voting procedures and social welfare functions.
\newblock {\em Journal of Economic Theory}, 10(2):187--217, 1975.

\bibitem{sertel88Approval}
M.~Sertel.
\newblock Characterizing approval voting.
\newblock {\em Journal of Economic Theory}, 45(1):207--211, 1988.

\bibitem{skow:multiwinner-models}
P.~Skowron.
\newblock What do we elect committees for? a voting committee model for
  multi-winner rules.
\newblock In {\em Proceedings of the 24th International Joint Conference on
  Artificial Intelligence (IJCAI-2015)}, pages 1141--1148, 2015.

\bibitem{owaWinner}
P.~Skowron, P.~Faliszewski, and J.~Lang.
\newblock Finding a collective set of items: From proportional
  multirepresentation to group recommendation.
\newblock {\em Artificial Intelligence}, 241:191--216, 2016.

\bibitem{sko-fal-sli:j:multiwinner}
P.~Skowron, P.~Faliszewski, and A.~Slinko.
\newblock Achieving fully proportional representation: {Approximability}
  result.
\newblock {\em Artificial Intelligence}, 222:67--103, 2015.

\bibitem{smi:j:scoring-rules}
J.~Smith.
\newblock Aggregation of preferences with variable electorate.
\newblock {\em Econometrica}, 41(6):1027--1041, 1973.

\bibitem{Thie95a}
T.~N. Thiele.
\newblock Om flerfoldsvalg.
\newblock In {\em Oversigt over det Kongelige Danske Videnskabernes Selskabs
  Forhandlinger}, pages 415--441. 1895.

\bibitem{yager1988}
R.~Yager.
\newblock On ordered weighted averaging aggregation operators in multicriteria
  decisionmaking.
\newblock {\em IEEE Transactions on Systems, Man and Cybernetics},
  18(1):183--190, 1988.

\bibitem{youngBorda}
H.~Young.
\newblock {An axiomatization of Borda's rule}.
\newblock {\em Journal of Economic Theory}, 9(1):43--52, 1974.

\bibitem{you:j:scoring-functions}
H.~Young.
\newblock Social choice scoring functions.
\newblock {\em SIAM Journal on Applied Mathematics}, 28(4):824--838, 1975.

\bibitem{lev-you:j:condorcet}
H.~Young and A.~Levenglick.
\newblock A consistent extension of {Condorcet}'s election principle.
\newblock {\em SIAM Journal on Applied Mathematics}, 35(2):285--300, 1978.

\bibitem{young74}
H.~P. Young.
\newblock A note on preference aggregation.
\newblock {\em Econometrica}, 42(6):1129--1131, 1974.

\end{thebibliography}

\end{document}